\documentclass[bst/sn-mathphys]{sn-jnl}

\usepackage{physics}
\usepackage[numbers,sort&compress]{natbib}
\usepackage[inline]{enumitem}
\usepackage{subcaption}

\jyear{2021}%

\theoremstyle{thmstyleone}%
\newtheorem{theorem}{Theorem}
\newtheorem*{theorem*}{Theorem}
\newtheorem{corollary}[theorem]{Corollary}
\newtheorem*{corollary*}{Corollary}
\newtheorem{claim}[theorem]{Claim}
\newtheorem*{claim*}{Claim}
\newtheorem{lemma}[theorem]{Lemma}
\newtheorem*{lemma*}{Lemma}

\theoremstyle{thmstyletwo}%

\theoremstyle{thmstylethree}%
\newtheorem{definition}{Definition}%

\newcommand{\m}[1]{{\mathbf{#1}}}
\newcommand{\ve}[1]{{\mathbf{#1}}}
\newcommand{\Sig}{{\text{\boldmath$\Sigma$}}}

\renewcommand{\R}{{\mathbb{R}}}

\newcommand{\algrule}[1][.2pt]{\par\vskip.5\baselineskip\hrule height #1\par\vskip.5\baselineskip}

\begin{document}

\title[Quantum algorithms for SVD-based data representation and analysis]{Quantum algorithms for SVD-based data representation and analysis}


\author*[1]{\fnm{Armando} \sur{Bellante}}\email{armando.bellante@polimi.it}

\author[2,3]{\fnm{Alessandro} \sur{Luongo}}\email{ale@nus.edu.sg}

\author[1]{\fnm{Stefano} \sur{Zanero}}\email{stefano.zanero@polimi.it}

\affil*[1]{\orgdiv{DEIB}, \orgname{Politecnico di Milano}, \orgaddress{\street{Via Ponzio, 34/5 – Building 20}, \city{Milan}, \postcode{20133}, \country{Italy}}}

\affil[2]{\orgdiv{IRIF}, \orgname{Universit{\'e} Paris Diderot}, \orgaddress{\street{Place Aur{\'e}lie Nemours, 8}, \city{Paris}, \postcode{75205}, \country{France}}}

\affil[3]{\orgdiv{CQT}, \orgname{National University of Singapore}, \orgaddress{\street{Science Drive 2, 3}, \city{City}, \postcode{117543}, \country{Singapore}}}


\abstract{This paper narrows the gap between previous literature on quantum linear algebra and practical data analysis on a quantum computer, formalizing quantum procedures that speed-up the solution of eigenproblems for data representations in machine learning.
The power and practical use of these subroutines is shown through new quantum algorithms, sublinear in the input matrix's size, for principal component analysis, correspondence analysis, and latent semantic analysis. 
We provide a theoretical analysis of the run-time and prove tight bounds on the randomized algorithms' error. 
We run experiments on multiple datasets, simulating PCA's dimensionality reduction for image classification with the novel routines. 
The results show that the run-time parameters that do not depend on the input's size are reasonable and that the error on the computed model is small, allowing for competitive classification performances.}

\keywords{Quantum computing, machine learning, data analysis, data representations, singular value decomposition, principal component analysis, correspondence analysis, latent semantic analysis}



\maketitle

\section{Introduction}
\label{Sec:introduction}
Quantum computation is a computing paradigm that promises substantial speed-ups in a plethora of tasks that are computationally hard for classical computers. 
In 2009, Harrow, Hassidim, and Lloyd \cite{HHL2009} presented quantum procedures to create a quantum state proportional to the solution of a linear system of equations $\m{A}\ve{x}=\ve{b}$ in time logarithmic in the size of $\m{A}$. 
This result has promoted further research on optimization, linear algebra, and machine learning problems, leading to faster quantum algorithms for linear regressions \cite{chakraborty2018power}, support vector machines \cite{rebentrost2014svm}, k-means \cite{kerenidis2019qmeans}, and many others \cite{biamonte2017quantum}.
Following this research line, in this work, we focus on quantum algorithms for singular value based data analysis and representation. When handling big data, it is crucial to learn effective representations that reduce the data's noise and help the learner perform better on the task. 
Many data representation methods for machine learning, such as principal component analysis \cite{partridge1997fastpca}, correspondence analysis \cite{greenacre1984corranalysis}, slow feature analysis \cite{kerenidis2020mnist}, or latent semantic analysis \cite{deerwester1990indexing}, heavily rely on singular value decomposition and are impractical to compute on classical computers for extensive datasets.

We have gathered and combined state-of-the-art quantum techniques to present a useful and easy-to-use framework for solving eigenvalue problems at large scale. 
While we focus on machine learning problems, these subroutines can be used for other problems that are classically solved via an SVD of a suitable matrix. More specifically, we formalize novel quantum procedures to compute classical estimates of the most relevant singular values, factor scores, factor score ratios of an $n\times m$ matrix in time poly-logarithmic in $nm$, and the most relevant singular vectors sub-linearly in $nm$.
We show how to use these procedures to obtain a classical description of the models of three machine learning algorithms: principal component analysis, correspondence analysis, and latent semantic analysis.
We also discuss how to represent the data in the new feature space with a quantum computer.
We provide a thorough theoretical analysis for all these algorithms bounding the run-time, the error, and the failure probability. 

The remainder of the paper is organized as follows. 
Section \ref{sec:preliminaries} introduces our notation and discusses the relevant quantum preliminaries.
Section \ref{Section:novel_quantum_subroutines} presents the novel quantum algorithms.
In Section \ref{Section:applications}, we show applications of the algorithms to principal component analysis, correspondence analysis, and latent semantic analysis.
Section \ref{Section:experiments} presents numerical experiments assessing the run-time parameters. 
Finally, we provide detailed information on the experiments and extensively discuss related work in quantum and classical literature in the appendix. 


\section{Quantum preliminaries and notation}\label{sec:preliminaries}
\subsection{Notation} Given a matrix $\m{A}$, we write $\ve{a}_{i,\cdot}$ to denote its $i^{th}$ row, $\ve{a}_{\cdot, j}$ for its $j^{th}$ column, and $a_{ij}$ for the element at row $i$, column $j$.
We write its singular value decomposition as $\m{A}=\m{U}\Sig \m{V}^T$. 
$\m{U}$ and $\m{V}$ are orthogonal matrices, whose column vectors $\ve{u}_i$ and $\ve{v}_i$ are respectively the left and right singular vectors of $\m{A}$.
$\Sig$ is a diagonal matrix with positive, non-negative, entries $\sigma_i$: the singular values. 
The row/column size of $\Sig$ is the \emph{rank} of $\m{A}$ and is denoted as $r$. We use $\lambda_i$ to denote the $i^{th}$ eigenvalue of the covariance matrix $\m{A}^T\m{A}=\m{V}\Sig^2\m{V}^T$, and $\lambda^{(i)}=\frac{\lambda_i}{\sum_j^r\lambda_j}$ to denote the relative magnitude of each eigenvalue. 
Using the notation of \citet{hsu2019correspondence} for correspondence analysis, we refer to $\lambda_i$ as factor scores and to $\lambda^{(i)}$ as factor score ratios. 
Note that $\lambda_i=\sigma_i^2$ and $\lambda^{(i)}=\frac{\sigma_i^2}{\sum_j^r\sigma_j^2}$. We denote the number of non-zero elements of a matrix/vector with $nnz()$. Given a scalar $a$, $\abs{a}$ is its absolute value. 
The $\ell_\infty$ and $\ell_0$ norm of a vector $\ve{a}$ are defined as $\norm{\ve{a}}_\infty = \max_i (\norm{a_i})$, $\norm{\ve{a}}_0 = nnz(\ve{a})$. If the vector norm is not specified, we refer to the $\ell_2$ norm. The Frobenius norm of a matrix is $\norm{\m{A}}_F = \sqrt{\sum_i^r \sigma_i^2}$, its spectral norm is $\norm{\m{A}} = \max_{\ve{x} \in \R^m}\frac{\norm{\m{A}\ve{x}}}{\norm{\ve{x}}} = \sigma_{max}$, and finally $\norm{\m{A}}_\infty = \max_i (\norm{\ve{a}_{i,\cdot}}_1)$.
A contingency table is a matrix that represents categorical variables in terms of the observed frequency counts.
Finally, when stating the complexity of an algorithm, we use $\widetilde{O}$ instead of $O$ to omit the poly-logarithmic terms on the size of the input data (i.e., $O(\text{polylog}(nm)) = \widetilde{O}(1)$), on the error, and the failure probability. 

\subsection{Quantum preliminaries}
We represent scalars as states of the computational basis of $\mathcal{H}_n$, where $n$ is the number of bits required for binary encoding. The quantum state corresponding to a vector $\ve{v} \in\mathbb{R}^m$ is defined as a \emph{state-vector} $\ket{\ve{v}}=\frac{1}{\norm{v}}\sum_j^m v_j \ket{j}$. 
Note that to build $\ket{\ve{v}}$ we need $\lceil \log m\rceil$ qubits.

\paragraph{Data access.} To access data in the form of state-vectors, we use the following definition of  quantum access.

\begin{definition} [Quantum access to a matrix]
\label{Def:quantum_access}
    We have quantum access to a matrix $\m{A} \in \R^{n \times m}$, if there exists a data structure that allows performing the mappings $\ket{i} \ket{0} \mapsto \ket{i}\ket{\ve{a}_{i, \cdot}} = \ket{i}\frac{1}{\norm{\ve{a}_{i,\cdot}}}\sum_j^m a_{ij}\ket{j}$, for all $i$, and $\ket{0} \mapsto \frac{1}{\norm{\m{A}}_F}\sum_i^n \norm{\ve{a}_{i,\cdot}} \ket{i}$ in time $\widetilde{O}(1)$.
\end{definition}
By combining the two mappings we can create the state $\ket{\m{A}} = \frac{1}{\norm{\m{A}}_F}\sum_i^n\sum_j^m a_{ij}\ket{i}\ket{j}$ in time $\widetilde{O}(1)$.

\citet{kerenidis2016recommendation, kerenidis2020gradient} have described one implementation of such quantum data access. 
Their implementation is based on a classical data structure such that the cost of updating/deleting/inserting one element of the matrix is poly-logarithmic in the number of its entries. 
In addition, their structure gives access to the Frobenius norm of the matrix and the norm of its rows in time $O(1)$.
The cost of creating this data structure is $\widetilde{O}(\text{nnz}(\m{A}))$.
This input model requires the existence of a QRAM \cite{giovannetti2008quantum}.
While there has been some skepticism on the possibility of error-correcting such a complex device, recent results show that bucket-brigade QRAMs are highly resilient to generic noise \cite{hann2021resilience}.

Sometimes it is desirable to normalize the input matrix to have a spectral norm smaller than one. \citet{kerenidis2020gradient} provide an efficient routine to estimate the spectral norm.
\begin{theorem}[Spectral norm estimation \cite{kerenidis2020gradient}]
\label{Theo:spectral_norm_estimation} 
    Let there be quantum access to the matrix $\m{A} \in \R ^{n\times m}$, and let $\epsilon>0$ be a precision parameter. There exists a quantum algorithm that estimates $\norm{\m{A}}$ to additive error $\epsilon\norm{\m{A}}_F$ in time $\widetilde{O}\left(\frac{\log(1/\epsilon)}{\epsilon}\frac{\norm{\m{A}}_F}{\norm{\m{A}}}\right)$.
\end{theorem}
If we have $\norm{\m{A}}$, we can create quantum access to $\m{A}'=\frac{\m{A}}{\norm{\m{A}}} = \m{U}\frac{\Sig}{\sigma_{max}}\m{V}^T$ in time $\widetilde{O}\left(\text{nnz}(\m{A})\right)$ by dividing each entry of the data structure. 
Once we have quantum access to a dataset, it is possible to apply a pipeline of quantum machine learning algorithms for data representation, analysis, clustering, and classification \cite{rebentrost2014quantum,kerenidis2020quantum,wang2017quantum,allcock2020quantum, kerenidis2019qmeans, kerenidis2020mnist}. 
Since the cost of each step of the pipeline should be evaluated independently, we consider $\widetilde{O}(nnz(\m{A}))$ to be a pre-processing cost and do not include it in our run-times.

We conclude this section by stating a useful claim that connects errors on classical vectors with errors on quantum states. 
\begin{claim} [Closeness of state-vectors \cite{kerenidis2020gradient}]
\label{claim:vector_quantum_distance}
    Let $\theta$ be the angle between vectors $\ve{x}, \overline{\ve{x}}$ and assume that $\theta < \pi/2$.
    Then, $\norm{\ve{x} - \overline{\ve{x}}} \leq \epsilon$ implies $\norm{\ket{\ve{x}} - \ket{\overline{\ve{x}}}} \leq \sqrt{2}\frac{\epsilon}{\norm{\ve{x}}}$.
\end{claim}

\paragraph{Useful subroutines}

We state two relevant quantum linear algebra results: quantum singular value estimation (SVE) and quantum matrix-vector multiplication.

\begin{theorem}[Singular value estimation \cite{kerenidis2020gradient}] 
\label{Theo:singular_value_estimation} 
    Let there be quantum access to $\m{A} \in \R ^{n\times m}$, with singular value decomposition $\m{A}=\sum_i^r\sigma_i \ve{u}_i\ve{v}_i^T$ and $r=min(n,m)$. 
    Let $\epsilon>0$ be a precision parameter. 
    It is possible to perform the mapping 
    $\ket{b}=\sum_i\alpha_i\ket{\ve{v}_i} \mapsto \sum_i\alpha_i\ket{\ve{v}_i}\ket{\overline{\sigma}_i}$, such that $\abs{\frac{\sigma_i}{\mu(\m{A})}-\overline{\sigma}_i} \leq \epsilon$ with probability at least $1 - 1/poly(m)$, in time $\widetilde{O}(\frac{1}{\epsilon})$ where $\mu(\m{A}) = \min\limits_{p \in [0,1]}(\norm{\m{A}}_F, \sqrt{s_{2p}(\m{A})s_{2(1-p)}(\m{A}^T)})$ and
    $s_p(\m{A}) = \max\limits_{i}\norm{\m{a}_{i,\cdot}}_p^p$. Similarly, we can have $\abs{\sigma_i-\overline{\sigma}_i}\leq \epsilon$ in time $\widetilde{O}(\frac{\mu(\m{A})}{\epsilon})$.
\end{theorem}

Unlike previous results with Hamiltonian simulations \cite{rebentrost2014quantum}, this algorithm enables performing conditional rotations using the singular values of a matrix without any special requirement (e.g., sparsity, being square, Hermitian, etc.). 
By choosing the same matrix $\ket{\m{A}} =  \frac{1}{\norm{\m{A}}_F}\sum_i^n\sum_j^m a_{ij}\ket{i}\ket{j} = \frac{1}{\norm{\m{A}}_F}\sum_i^k\sigma_i\ket{\ve{u}_i}\ket{\ve{v}_i}$ as starting state $\ket{b}$, we obtain a superposition of all the singular values entangled with the respective left and right singular vectors $\frac{1}{\norm{\m{A}}_F}\sum_i^r \sigma_i \ket{\ve{u}_i}\ket{\ve{v}_i}\ket{\overline{\sigma}_i}$.
In this case, the requirement $r=\min (n,m)$ is not needed anymore as $\ket{b} = \ket{\m{A}}$ can be fully decomposed in terms of $\m{A}$'s right singular vectors. 

This algorithm uses phase estimation.
In this work, we consider this algorithm to use a consistent version of phase estimation, so that the errors in the estimates of the singular values are consistent across multiple runs \citep{ta2013inverting, kerenidis2020gradient}. 

\begin{theorem}[Matrix-vector multiplication \cite{chakraborty2018power} (Lemma 24, 25)]
\label{Theo:matrix_vec_mul} 
    Let there be quantum access to the matrix $\m{A} \in \R ^{n\times n}$, with $\sigma_{max} \leq 1$, and to a vector $\ve{x} \in \R^{n}$.
    Let $\norm{\m{A}\ve{x}} \geq \gamma$. 
    There exists a quantum algorithm that creates a state $\ket{\ve{z}}$ such that $\norm{\ket{\ve{z}} - \ket{\m{A}\ve{x}}} \leq \epsilon$ in time $\widetilde{O}(\frac{1}{\gamma}\mu(\m{A})\log(1/\epsilon))$, with probability at least $1-1/poly(n)$.
    Increasing the run-time by a multiplicative factor $\widetilde{O}\left(\frac{1}{\eta}\right)$ one can retrieve an estimate of $\norm{\m{A}\ve{x}}$ to relative error $\eta$.
\end{theorem}

\paragraph{Data output.} Finally, to read out the quantum states, we state one version of amplitude amplification and estimation, and two \emph{state-vector} tomographies. 
\begin{theorem}[Amplitude amplification and estimation \cite{brassard2002quantum, kerenidis2019portf}]
\label{Theo:ampamp}
    Let there be a unitary that performs the mapping $U_x: \ket{0} \mapsto sin(\theta)\ket{\ve{x},0} + cos(\theta)\ket{\ve{G},0^\perp}$, where $\ket{\ve{G}}$ is a garbage state, in time $T(U_x)$. Then, $sin(\theta)^2$ can be estimated to multiplicative error $\eta$ in time $O(\frac{T(U_x)}{\eta sin(\theta)})$ or to additive error $\eta$ in time $O(\frac{T(U_x)}{\eta})$, and $\ket{\ve{x}}$ can be generated in expected time $O(\frac{T(U_x)}{sin(\theta)})$.  
\end{theorem}

\begin{theorem} [$\ell_2$ state-vector tomography \cite{kerenidis2018interiorpoint, kerenidis2019portf}] 
\label{Theo:l2_vector_tomography}
    Given a unitary mapping $U_x: \ket{0} \mapsto \ket{\ve{x}}$ in time $T(U_x)$ and $\delta>0$, there is an algorithm that produces an estimate $\overline{\ve{x}} \in \R^m$ with $\norm{\overline{\ve{x}}} = 1$ such that $\norm{\ve{x} - \overline{\ve{x}}} \leq \delta$ with probability at least $1 - 1/poly(m)$ in time $O(T(U_x)\frac{m \log m}{\delta^2})$. 
\end{theorem}

\begin{theorem} [$\ell_\infty$ state-vector tomography \cite{kerenidis2019infty}] 
\label{Theo:linf_vector_tomography} 
    Given access to a unitary mapping $U_x: \ket{0} \mapsto \ket{\ve{x}}$ and its controlled version in time $T(U_x)$, and  $\delta>0$, there is an algorithm that produces an estimate $\overline{\ve{x}} \in \R^m$ with $\norm{\overline{\ve{x}}}=1$ such that $\norm{\ve{x} - \overline{\ve{x}}}_\infty \leq \delta$ with probability at least $1 - 1/poly(m)$ in time $O(T(U_x)\frac{\log m}{\delta^2})$. 
\end{theorem}


\section{Novel quantum methods}
\label{Section:novel_quantum_subroutines}
Building from the previous section's techniques, we formalize a series of quantum algorithms that allow us to retrieve a classical description of the singular value decomposition of a matrix to which we have quantum access.

\subsection{Estimating the quality of the representation}
\label{Subsec:quality_repr}
Algorithms such as principal component analysis and correspondence analysis are often used for visualization or dimensionality reduction purposes. 
These applications work better when a small subset of factor scores have high factor score ratios. 
We provide a fast procedure that allows verifying if this is the case: given efficient quantum access to a matrix $\m{A} \in \R^{n \times m}$, it retrieves the most relevant singular values, factor scores, and factor score ratios in time poly-logarithmic in the number of elements of $\m{A}$, with no strict dependencies on its rank. 

The main intuition behind this algorithm is that it is possible to create the state
$\sum_{i}^{r} \sqrt{\lambda^{(i)}} \ket{\ve{u}_i}\ket{\ve{v}_i}\ket{\overline{\sigma}_i}$.
The third register, when measured in the computational basis, outputs the estimate $\overline{\sigma}_i$ of a singular value with probability equal to its factor score ratio $\lambda^{(i)}$. 
This enables sampling the singular values of $\m{A}$ directly from the factor score ratios' distribution.
When a matrix has a huge number of small singular values and only a few of them that are very big, the ones with the greatest factor score ratios will appear many times during the measurements.
In contrast, the negligible ones are not likely to be measured.  
This intuition has already appeared in literature \cite{gyurik2020towards,cade2018quantum}. 
Nevertheless, the analysis and the problem solved in these works are different from ours. 
In the context of data representation and analysis, this intuition has only been sketched for sparse or low rank square symmetric matrices by \citet{lloyd2014qpca}, without a precise formalization.
We thoroughly formalize it for any real matrix.

\begin{algorithm}[H]
    \caption{Quantum factor score ratio estimation.}
    \label{alg_main:factor_score_estimation}
    \begin{algorithmic}[1]
        \Statex \textbf{Input:} Quantum access to a matrix $\m{A} \in \R^{n\times m}$. Two precision parameters $\gamma, \epsilon \in \R_{>0}$.
        \Statex \textbf{Output:} An estimate of the factor score ratios $\lambda^{(i)} > \gamma$. An estimate of the corresponding singular values and factor scores.
        \algrule
        \For{\texttt{$N \sim O(\frac{1}{\gamma^2})$ times}}
            \State Prepare the state
            $\frac{1}{\norm{\m{A}}_F}\sum_{i}^{n}\sum_{j}^{m}a_{ij}\ket{i}\ket{j}$. \label{alg1:QRAMquery}
            \State Apply SVE to get $\frac{1}{\sqrt{\sum_{j}^{r}\sigma_j^2}}\sum_{i}^{r} \sigma_i \ket{\ve{u}_i}\ket{\ve{v}_i}\ket{\overline{\sigma}_i}$.\label{alg1:SVEstep}
            \State Measure the last register and store it in a set data structure. \label{alg1:Measurestep}
        \EndFor
        \State For each $\overline{\sigma}_i$ measured, output $\overline{\sigma}_i$, its factor score $\overline{\lambda}_i = \overline{\sigma}^2_i$, and its factor score ratio 
        $\overline{\lambda}^{(i)} = \frac{\overline{\lambda}_i}{\norm{A}^2_F}$.
    \end{algorithmic}
\end{algorithm} 

\begin{theorem}[Quantum factor score ratio estimation] 
\label{TheoMio:factor_score_estimation} 
    Let there be quantum access to a matrix $\m{A} \in \R^{n \times m}$, with singular value decomposition $\m{A} = \sum_i\sigma_i\ve{u}_i\ve{v}^T_i$. 
    Let $\gamma, \epsilon$ be precision parameters. 
    There exists a quantum algorithm that runs in time $\widetilde{O}\left(\frac{1}{\gamma^2}\frac{\mu(\m{A})}{\epsilon}\right)$ and estimates:\\
    \begin{itemize*}
    \item all the factor score ratios $\lambda^{(i)} > \gamma$, with probability at least $1-1/\text{poly}(r)$, such that $\abs{\lambda^{(i)} - \overline{\lambda}^{(i)}} \leq 2\epsilon\frac{\sigma_i}{\norm{A}^2_F}$,  with probability at least $1-1/\text{poly}(n)$; \\
    \item the corresponding singular values $\sigma_i$, such that $\abs{\sigma_i - \overline{\sigma}_i} \leq \epsilon$ with probability at least $1-1/\text{poly}(n)$; \\
    \item the corresponding factor scores $\lambda_i$, such that $\abs{\lambda_i - \overline{\lambda}_i} \leq 2\epsilon\sqrt{\lambda_i}$  with probability at least $1-1/\text{poly}(n)$.
    \end{itemize*}
\end{theorem}

The proof consists in bounding the run-time, the error, and the probability of failure of Algorithm \ref{alg_main:factor_score_estimation}.

\begin{proof}
By the definition of quantum access, the cost of step \ref{alg1:QRAMquery} is $\widetilde{O}(1)$.
The singular value estimation in step \ref{alg1:SVEstep} can be performed using Theorem \ref{Theo:singular_value_estimation} in time 
$\widetilde{O}\left(\frac{\mu(\m{A})}{\tau}\right)$,
such that $\norm{\sigma_i - \overline{\sigma}_i}\leq \epsilon$ with probability at least $1-1/\text{poly}(n)$.
A measurement of the third register at step \ref{alg1:Measurestep} can output any $\overline{\sigma}_i$ with probability $\lambda^{(i)} = \frac{\sigma_i^2}{\sum_j^r\sigma_j^2}$. 

Theorem \ref{Theo:linf_vector_tomography} guarantees that with $O(1/\gamma^2)$ measurements we can get estimates $\abs{\lambda^{(i)} - \overline{\lambda_i}}\leq \gamma$.
In particular, \citet{kerenidis2019infty} estimate that $N=36\log(r)/\gamma^2$ measures should suffice for our goal.

Alternatively, we could consider the measurement process as performing $r$ Bernoulli trials: one for each $\overline{\lambda}^{(i)}$, so that if we measure $\overline{\sigma}_i$ it is a success for the $i^{th}$ Bernoulli trial and a failure for all the others. 
Given a confidence level $z$, it is possible to use the Wald confidence interval to determine a value for $N$ such that $\abs{\lambda^{(i)} - \frac{\zeta_{\overline{\sigma}_i}}{N}} \leq \gamma$ with confidence level $z$, where $\zeta_{\overline{\sigma}_i}$ is the number of times that $\overline{\sigma}_i$ has appeared in the measurements.
In this case, it suffice to choose $N=\frac{z^2}{4\gamma^2}$ \cite[Section 5.1.3]{quantumsupervisedschuld}. Having $\abs{\lambda^{(i)} - \overline{\lambda}^{(i)}}\leq \gamma$ means measuring all the $\overline{\sigma}_i$ whose factor score ratio is greater than $\lambda$.

We now proceed with the error analysis.
We can compute $\overline{\lambda}_i=\overline{\sigma}_i^2$. 
\begin{align}
    \abs{\lambda_i - \overline{\sigma}_i^2} \leq \abs{\lambda_i - (\sigma_i \pm \epsilon)^2} = \abs{ \pm 2\epsilon\sigma_i + \epsilon^2 } \leq 2\epsilon\sigma_i + \epsilon^2.
\end{align}
If we keep the error analysis at the first order and consider that $\sigma_i = \sqrt{\lambda_i}$, we can conclude the bound as $\abs{\lambda_i - \overline{\sigma}_i^2} \leq 2\epsilon\sqrt{\lambda_i}$. Similarly, we can compute $\overline{\lambda}^{(i)}=\frac{\overline{\sigma}_i^2}{\norm{A}_F^2}$. 
\begin{align}
    \abs{\lambda^{(i)} - \overline{\lambda}^{(i)}} = \frac{\abs{\lambda_i - \overline{\lambda}_i}}{\norm{A}_F^2} \leq 2\epsilon\frac{\sigma_i}{\norm{A}_F^2}.
\end{align}
\end{proof}

The parameter $\gamma$ is the one that controls how big a factor score ratio should be for the singular value/factor score to be measured. 
If we choose $\gamma$ bigger than the least factor scores ratio of interest, the estimate for the smaller ones is likely to be $0$, as $\abs{\lambda^{(i)}-0}\leq \gamma$ would be a plausible estimation. 

Often in data representations, the cumulative sum of the factor score ratios is a measure of the quality of the representation. 
By slightly modifying Algorithm \ref{alg_main:factor_score_estimation} to use Theorem \ref{Theo:linf_vector_tomography}, it is possible to estimate this sum such that $\abs{\sum_i^k \lambda^{(i)} - \sum_i^k \overline{\lambda}^{(i)}} \leq k\epsilon$ with probability $1-1/\text{poly}(r)$.
However, a slight variation of Algorithm \emph{IV.3} for spectral norm estimation in \citet{kerenidis2020gradient} provides a more accurate estimation in less time, given a threshold $\theta$ for the smallest singular value to retain.

\begin{algorithm}[H]
	\caption{Quantum check on the factor score ratios' sum.}
	\label{alg:factor_score_sum_check}
	\begin{algorithmic}[1]
	    \Statex \textbf{Input:} Quantum access to a matrix $\m{A} \in \R^{n\times m}$. A threshold parameter $\theta$. Two precision parameters $\epsilon, \eta \in \R_{>0}$, such that the greatest singular value smaller than $\theta$ is not more than $\epsilon$ distant to it.
        \Statex \textbf{Output:} An estimate $\overline{p}$ of the factor score ratios' sum $p=\sum_{i: \sigma_i \geq \theta} \lambda^{(i)}$. 
        \algrule

	    \State Prepare the state
    	$\frac{1}{\norm{\m{A}}_F}\sum_{i}^{n}\sum_{j}^{m}a_{ij}\ket{i}\ket{j}$. \label{AlgCheckVar:QRAM}
	    \State Apply SVE with precision $\epsilon$ to get $\frac{1}{\sqrt{\sum_{j}^{r}\sigma_j^2}}\sum_{i}^{r} \sigma_i \ket{\ve{u}_i}\ket{\ve{v}_i}\ket{\overline{\sigma}_i}$.  \label{AlgCheckVar:SVE}
	    \State Append a quantum register $\ket{0}$ to the state and set it to $\ket{1}$ if $\overline{\sigma}_i < \theta$. \label{AlgCheckVar:Cond_rot}
		\State Uncompute the SVE
        $$\frac{1}{\sqrt{\sum_{j}^{r}\sigma_j^2}}\sum_{i: \overline{\sigma}_i \geq \theta} \sigma_i \ket{\ve{u}_i}\ket{\ve{v}_i}\ket{0} + \frac{1}{\sqrt{\sum_{j}^{r}\sigma_j^2}}\sum_{i: \overline{\sigma}_i < \theta} \sigma_i \ket{\ve{u}_i}\ket{\ve{v}_i}\ket{1}$$ \label{AlgCheckVar:SVE_Undo}
		\State Perform amplitude estimation with precision $\eta$ on the last register being $\ket{0}$, to estimate $p=\frac{\sum_{i: \overline{\sigma}_i \geq \theta} \sigma_i^2}{\sum_j^r \sigma_j^2}= \sum_{i: \overline{\sigma}_i \geq \theta} \lambda^{(i)}$. \label{AlgCheckVar:amp_est}
	\end{algorithmic}
\end{algorithm} 

\begin{theorem}[Quantum check on the factor score ratios' sum] 
\label{TheoMio:check_explained_variance} 
    Let there be quantum access to a matrix $\m{A} \in \R^{n \times m}$, with singular value decomposition $\m{A} = \sum_i\sigma_i\ve{u}_i\ve{v}^T_i$. Let $\eta, \epsilon$ be precision parameters, and $\theta$ be a threshold for the smallest singular value to consider.
    There exists a quantum algorithm that estimates $p = \sum_{i: \overline{\sigma}_i \geq \theta} \lambda^{(i)}$,
    where $\abs{\sigma_i - \overline{\sigma}_i} \leq \epsilon$, to relative error $\eta$ in time $\widetilde{O}\left(\frac{\mu(\m{A})}{\epsilon}\frac{1}{\eta \sqrt{p}}\right)$.
\end{theorem}
\label{Proof:check_explained_variance}

\begin{proof}
As discussed in the previous proof, the cost of preparing the state at step \ref{AlgCheckVar:SVE} is $\widetilde{O}\left(\frac{\mu(\m{A})}{\epsilon}\right)$. 
The complexity of step \ref{algTop-k:cond_rot} is $\widetilde{O}(1)$, as it is an arithmetic operation that only depends on the encoding of $\ket{\overline{\sigma}_i}$.
Step \ref{AlgCheckVar:SVE_Undo} consists in uncomputing step \ref{AlgCheckVar:SVE} and has its same cost.
Finally, the cost of amplitude estimation, with relative precision $\eta$, on the last register being $\ket{0}$ is equal to $O\left(T(U_{\ref{AlgCheckVar:SVE_Undo}})\frac{1}{\eta\sqrt{p}}\right)$, 
where $p = \frac{\sum_{i: \overline{\sigma}_i \geq \theta} \sigma_i^2}{\sum_j^r \sigma_j^2}$ is the probability of measuring $\ket{0}$ (Theorem \ref{Theo:ampamp}).
The overall complexity is proven:
$\widetilde{O}\left(\frac{\mu(\m{A})}{\epsilon}\frac{1}{\eta\sqrt{p}}\right).$
\end{proof}    

Since the sum of factor score ratios $p$ is a measure of the representation quality, in problems such as PCA, CA, and LSA, this is usually a constant number bigger than $0$ (i.e., often in practice, $p \in [0.3, 1]$).
This makes the term $\sqrt{p}$ negligible in most of the practical applications.
Moreover, we further modify Algorithm \emph{IV.3} to perform a binary search of $\theta$ given the desired sum of factor score ratios.

\begin{algorithm}[H]
    \caption{Quantum binary search for the singular value threshold.}
    \label{alg_main:binary_search_theta}
    \begin{algorithmic}[1]
        \Statex \textbf{Input:} Quantum access to a matrix $\m{A} \in \R^{n\times m}$. The desired amount of factor score ratios sum $p \in [0,1]$. Two precision parameters $\epsilon, \eta \in \R_{>0}$.
        \Statex \textbf{Output:} A threshold $\theta$ such that $\abs{p - \sum_{i: \overline{\sigma}_i \geq \theta} \lambda^{(i)} } \leq \eta$, where $\abs{\overline{\sigma}_i - \sigma_i} \leq \epsilon$, or $-1$ if no such $\theta$ exists.
        \algrule
        
        \State Let $l=0$ and $u=1$ be upper and lower bounds for the binary search.
        \State If $\abs{1 - p} \leq \eta$, then return $\theta=0$.
        \State If $\abs{0 - p} \leq \eta$, then return $\theta=\mu(\m{A})$.
        \State Initialize $\tau=(l+u)/2$.
        \For{\texttt{$O(\log(\frac{\mu(\m{A})}{\epsilon}))$ times}}
            \State Prepare the state $\ket{\m{A}}$ and apply SVE to get $\frac{1}{\sqrt{\sum_{j}^{r}\sigma_j^2}}\sum_{i}^{r} \sigma_i \ket{\ve{u}_i}\ket{\ve{v}_i}\ket{\hat{\sigma}_i}$, so that $\abs{\hat{\sigma}_i - \frac{\sigma_i}{\mu(\m{A})}} \leq \frac{\epsilon}{\mu(\m{A})}$.\label{alg3:SVEstep}
            \State Append a quantum register $\ket{0}$ to the state and set it to $\ket{1}$ if $\hat{\sigma}_i < \tau$. \label{alg3:ControlledRot}
            \State Uncompute the SVE
            $$\frac{1}{\sqrt{\sum_{j}^{r}\sigma_j^2}}\sum_{i: \hat{\sigma}_i \geq \tau} \sigma_i \ket{\ve{u}_i}\ket{\ve{v}_i}\ket{0} + \frac{1}{\sqrt{\sum_{j}^{r}\sigma_j^2}}\sum_{i: \hat{\sigma}_i < \tau} \sigma_i \ket{\ve{u}_i}\ket{\ve{v}_i}\ket{1}.$$
    		\State Perform amplitude estimation with precision $\eta$ on the last register being $\ket{0}$, to estimate $p_\tau =\frac{\sum_{i: \hat{\sigma}_i \geq \tau} \sigma_i^2}{\sum_j^r \sigma_j^2}= \sum_{i: \hat{\sigma}_i \geq \tau} \lambda^{(i)}$, such that $\abs{\overline{p}_\tau - p_\tau } \leq \eta/2$. \label{alg3:amp_est}
            \State If $\abs{\overline{p}_\tau - p} \leq \eta/2$, then return $\theta=\tau\mu(\m{A})$.
            \State If $\overline{p}_\tau < p$, then set $u = \tau$ and set $l = \tau$ otherwise.
    	    \State Update $\tau=(u+l)/2$. \label{alg3:updatestep}
    	\EndFor
    	\State Return -1.
    \end{algorithmic}
\end{algorithm} 

\begin{theorem}[Quantum binary search for the singular value threshold \cite{kerenidis2020gradient}] 
\label{Theorivisto:binarysearch} 
    Let there be quantum access to a matrix $\m{A} \in \R^{n \times m}$. 
    Let $\eta, \epsilon$ be precision parameters, and $\theta$ be a threshold for the smallest singular value to consider.
    Let $p \in [0,1]$ be the factor score ratios sum to retain.
    There exists a quantum algorithm that runs in time $\widetilde{O}\left(\frac{\mu(\m{A})\log(\mu(\m{A})/\epsilon)}{\epsilon\eta}\right)$ and outputs an estimate $\theta$ such that $\abs{p - \sum_{i: \overline{\sigma}_i \geq \theta} \lambda^{(i)}} \leq \eta$, where $\abs{\overline{\sigma}_i - \sigma_i} \leq \epsilon$, or detects whether such $\theta$ does not exists.
\end{theorem}

The proof consists in proving the correctness and the run-time of Algorithm \ref{alg_main:binary_search_theta}.
\begin{proof}
The algorithm searches for $\theta$ using $\tau$ as an estimate between $0$ and $1$.
The search is performed using $\text{sign}(p_\tau - p)$ as an oracle that tells us whether to update the lower or upper bound for $\tau$.

The algorithm terminates when $\abs{\overline{p}_\tau - p} \leq \eta/2$ or when it is not possible to update $\tau$ anymore (i.e., there are not enough qubits to express the next $\tau$).
In this last case, there is no $\theta$ that satisfies the requisites and the algorithm returns $-1$.

In the first case, instead, we need to guarantee that $\abs{p - \sum_{i: \overline{\sigma}_i \geq \theta} \lambda^{(i)} } = \abs{p - p_\tau} \leq \eta$. 
Since we run amplitude estimation with additive error $\eta/2$ we have $\abs{\overline{p}_\tau - {p}_\tau} \leq \eta/2$, and we require
$\abs{\overline{p}_\tau - p} \leq \eta/2$ to stop.
This two conditions entail
\begin{align}
    \abs{p-p_\tau} \leq \abs{\abs{p - \overline{p}_\tau} + \eta/2} \leq \eta.
\end{align}

If we want $\theta$ to be comparable with the singular values of $\m{A}$ and use $\tau$ for the binary search, we have to use Theorem \ref{Theo:singular_value_estimation} with error $\frac{\mu(\m{A})}{\epsilon}$, meaning that Step \ref{alg3:SVEstep} can be done in time $O(\mu(\m{A})/\epsilon)$.
The total cost of the inner loop has to be evaluated at the end of Step \ref{alg3:amp_est}, which runs in time $O(\frac{\mu(\m{A})}{\epsilon \eta})$.

The maximum number of updates of $\tau$ is bounded by the number of qubits that we use to store the singular values $\hat{\sigma}_i$.
This is given by the logarithm of the error used in Step \ref{alg3:SVEstep}, and is $O\left(\log\left(\frac{\mu{(A)}}{\epsilon}\right)\right)$.

The run-time of this algorithm is bounded by $\widetilde{O}\left(\frac{\mu(\m{A})\log(\mu(\m{A})/\epsilon)}{\epsilon\eta}\right)$.
\end{proof}

Using the quantum counting algorithms of \citet{brassard2002quantum} after step \ref{AlgCheckVar:Cond_rot} of Algorithm \ref{alg:factor_score_sum_check}, it is possible to count the number of singular values retained by a certain threshold $\theta$.

\begin{corollary}[Quantum reduced rank estimation] 
\label{TheoMio:counting} 
    Let there be quantum access to a matrix $\m{A} \in \R^{n \times m}$, with singular value decomposition $\m{A} = \sum_i^r\sigma_i\ve{u}_i\ve{v}^T_i$ and rank $r$. Let $\epsilon$ be a precision parameter, and $\theta$ be a threshold for the smallest singular value to consider.
    There exists a quantum algorithm that estimates the exact number $k$ of singular values such that $\overline{\sigma}_i \geq \theta$,
    where $\abs{\sigma_i - \overline{\sigma}_i} \leq \epsilon$, in time $\widetilde{O}\left(\frac{\mu(\m{A})}{\epsilon}\sqrt{(k+1)(r-k+1)}\right)$ with probability at least $\frac{2}{3}$.
    
    Similarly, given a parameter $\eta$, it is possible to produce an estimate $\overline{k}$ such that $\abs{\overline{k}-k} \leq \eta k$ in time $\widetilde{O}\left(\frac{\mu(\m{A})}{\epsilon\eta}\sqrt{\frac{r}{k}}\right)$ with probability at least $\frac{2}{3}$.
\end{corollary}

Estimating the number of singular values retained by $\theta$ is helpful.
When the singular values are dense around $\theta$, this Corollary, together with Theorem \ref{TheoMio:check_explained_variance}, can help the analyst evaluate trade-offs between big $p$ and small $k$.
On the one hand, the bigger $p$ is, the more information on the dataset one can retain. 
On the other hand, the bigger $k$ is, the slower will the algorithms in the next section be.

\subsection{Extracting the SVD representation}\label{Subsec:extraction}
After introducing the procedures to test for the most relevant singular values, factor scores and factor score ratios of $\m{A}$, we present a routine to extract the corresponding right/left singular vectors.
The inputs of this algorithm, other than the matrix, are a parameter $\delta$ for the precision of the singular vectors, a parameter $\epsilon$ for the precision of the singular value estimation, and a threshold $\theta$ to discard the non interesting singular values/vectors.
The output guarantees a unit estimate $\overline{\ve{x}}_i$ of each singular vector such that $\norm{\ve{x}_i -\overline{\ve{x}}_i} \leq \delta$, ensuring that the estimate has a similar orientation to the original vector. 
Additionally, this subroutine can provide an estimation of the singular values greater than $\theta$, to absolute error $\epsilon$.

\begin{algorithm}[t]
    \caption{Quantum top-k singular vectors extraction.}
    \label{alg_main:topk_sv}
    
    \begin{algorithmic}[1]
        \Statex \textbf{Input:} Quantum access to a matrix $\m{A} \in \R^{n\times m}$. A threshold $\theta$ that captures the top-k singular values. Two precision parameters $\delta, \epsilon \in \R_{>0}$.
        \Statex \textbf{Output:} The top-k singular vectors such that $\norm{\ve{u}_i - \overline{\ve{u}}_i} \leq \delta$ and $\norm{\ve{v}_i - \overline{\ve{v}}_i} \leq \delta$. Optionally, the top-k singular values, such that $\norm{\sigma_i - \overline{\sigma}_i} \leq \epsilon$.
        \algrule

        \State Prepare the state $\frac{1}{\norm{\m{A}}_F}\sum_{i}^{n}\sum_{j}^{m}a_{ij}\ket{i}\ket{j}$. 
        
        \State Apply SVE 
        to get $\frac{1}{\sqrt{\sum_{j}^{r}\sigma_j^2}}\sum_{i}^{r} \sigma_i \ket{\ve{u}_i}\ket{\ve{v}_i}\ket{\overline{\sigma}_i}$, where $\abs{\sigma_i - \overline{\sigma_i}} \leq \epsilon$. 

        \State Append a quantum register $\ket{0}$ to the state and set it to $\ket{1}$ if $\ket{\overline{\sigma}_i} < \theta$. \label{algTop-k:cond_rot}

        \State Perform amplitude amplification for $\ket{0}$, to get the state
       $\frac{1}{\sqrt{\sum_j^k \sigma_j^2}}\sum_i^k \sigma_i\ket{\ve{u}_i}\ket{\ve{v}_i}\ket{\overline{\sigma}}$. \label{algTop-k:ampamp1}

        \State Append a second ancillary register $\ket{0}$ and perform the controlled rotation
        $\frac{C}{\norm{\m{A}^{(k)}}_F}\sum_i^k \frac{\sigma_i}{\overline{\sigma}_i} \ket{\ve{u}_i}\ket{\ve{v}_i}\ket{\overline{\sigma}_i}\ket{0} +\frac{1}{\norm{\m{A}^{(k)}}_F}\sum_i^k \sqrt{1-\frac{C^2}{\overline{\sigma}_i^2}}\ket{\ve{u}_i}\ket{\ve{v}_i}\ket{\overline{\sigma}_i}\ket{1}$ 
        where $C$ is a normalization constant.\label{algTop-k:cond_rot2}
        
        \State Perform again amplitude amplification for $\ket{0}$ to get the uniform superposition
        $\frac{1}{\sqrt{k}}\sum_i^k \ket{\ve{u}_i}\ket{\ve{v}_i}\ket{\overline{\sigma}_i}.$ \label{algTop-k:ampamp2}

        \State Measure the last register and, according to the measured $\ket{\overline{\sigma}_i}$, apply state-vector tomography on $\ket{\ve{u}_i}$ for the $i^{th}$ left singular vector or on $\ket{\ve{v}_i}$ for the right one. \label{algTop-k:vector_tomography}
        
        \State Repeat 1-7 until the tomography requirements are met.
        \State Output the $k$ singular vectors $\ve{u}_i$ or $\ve{v}_i$ and, optionally, the singular values $\overline{\sigma}_i$.
    \end{algorithmic}
\end{algorithm} 

\begin{theorem} [Top-k singular vectors extraction] 
\label{TheoMio:top-k_sv_extraction}
    Let there be efficient quantum access to a matrix $\m{A} \in \R^{n \times m}$, with singular value decomposition $\m{A} = \sum_i^r \sigma_i \ve{u}_i \ve{v}_i^T$. Let $\delta > 0$ be a precision parameter for the singular vectors, $\epsilon>0$ a precision parameter for the singular values, and  $\theta>0$ be a threshold such that $\m{A}$ has $k$ singular values greater than $\theta$. Define $p=\frac{\sum_{i: \overline{\sigma}_i \geq \theta} \sigma_i^2}{\sum_j^r \sigma_j^2}$. There exist quantum algorithms that estimate:\\
    \begin{itemize*}
        \item The top $k$ left singular vectors $\ve{u}_i$ of $\m{A}$ with unit vectors $\overline{\ve{u}}_i$
        such that $\norm{\ve{u}_i-\overline{\ve{u}}_i}_2 \leq \delta$ with probability at least $1-1/poly(n)$, in time $\widetilde{O}\left(\frac{\norm{A}}{\theta}\frac{1}{\sqrt{p}}\frac{\mu(\m{A})}{\epsilon}\frac{kn}{\delta^2}\right)$; \\
        \item The top $k$ right singular vectors  $\ve{v}_i$ of $\m{A}$ with unit vectors $\overline{\ve{v}}_i$
        such that $\norm{\ve{v}_i-\overline{\ve{v}}_i}_2 \leq \delta$ with probability at least $1-1/poly(m)$, in time $\widetilde{O}\left(\frac{\norm{A}}{\theta}\frac{1}{\sqrt{p}}\frac{\mu(\m{A})}{\epsilon}\frac{km}{\delta^2}\right)$.\\
        \item The top $k$ singular values $\sigma_i$, factor scores $\lambda_i$, and factor score ratios $\lambda^{(i)}$ of $\m{A}$ to precision $\epsilon$, $2\epsilon\sqrt{\lambda_i}$, and $\epsilon\frac{\sigma_i}{\norm{A}^2_F}$ respectively, with probability at least $1 - 1/\text{poly}(m)$, in time $\widetilde{O}\left(\frac{\norm{A}}{\theta}\frac{1}{\sqrt{p}}\frac{\mu(\m{A})k}{\epsilon}\right)$ or during any of the two procedures above.
    \end{itemize*}
\end{theorem}
The proof consists in proving the time complexity and the error of Algorithm \ref{alg_main:topk_sv}.
\begin{proof}
Like in the previous proofs, the cost of preparing the state at step \ref{algTop-k:ampamp1}, is  $\widetilde{O}\left(\frac{1}{\sqrt{p}}\frac{\mu(\m{A})}{\epsilon}\right)$, where $\widetilde{O}\left(\frac{\mu(\m{A})}{\epsilon}\right)$ is the cost of singular value estimation and $\widetilde{O}\left(\frac{1}{\sqrt{p}}\right)$ is the one of amplitude amplification.
Step \ref{algTop-k:cond_rot2} is a conditional rotation and similarly to step \ref{algTop-k:cond_rot} it has a negligible cost.
The next step is to analyze the amplitude amplification at \ref{algTop-k:ampamp2}. 
The constant $C$ is a normalization factor in the order of $\widetilde{O}(1/\kappa(\m{A^{(k)}}))$ where $\kappa(\m{A^{(k)}}) = \frac{\sigma_{max}}{\sigma_{min}}$
is the condition number of the low-rank matrix 
$\m{A}^{(k)}$.
Since for construction $\sigma_{min} \geq \theta$,
we can bound the condition number 
$\kappa(\m{A}^{(k)}) \leq \frac{\norm{\m{A}}}{\theta}$.
From the famous work of Harrow, Hassidim and Lloyd \cite{HHL2009} we know that applying amplitude amplification on the state above, with the the third register being 
$\ket{0}$,
would cost 
$T(U_{\ref{algTop-k:ampamp2}}) \sim \widetilde{O}(\kappa(\m{A}^{(k)}) T(U_{\ref{algTop-k:cond_rot2}})) \sim \widetilde{O}\left(\frac{\norm{\m{A}}}{\theta}\frac{1}{\sqrt{p}}\frac{\mu(A)}{\epsilon}\right)$.

This last amplitude amplification leaves the registers in the state 
\begin{align}
\label{Eq:coupon_collect}
    \frac{1}{\sqrt{\sum_{i}^k \frac{\sigma_i^2}{\overline{\sigma}_i^2}}}\sum_i^k\frac{\sigma_i}{\overline{\sigma}_i} \ket{\ve{u}_i}\ket{\ve{v}_i}\ket{\overline{\sigma}_i} \sim 
    \frac{1}{\sqrt{k}}\sum_i^k \ket{\ve{u}_i}\ket{\ve{v}_i}\ket{\overline{\sigma}_i}
\end{align}
where
$\overline{\sigma}_i \in [\sigma_i - \epsilon, \sigma_i + \epsilon]$ 
and 
$\frac{\sigma_i}{\sigma_i \pm \epsilon} \rightarrow 1$ 
for
$\epsilon \rightarrow 0$.

When measuring the last register of state \ref{algTop-k:ampamp2} in the computational basis, we measure $\ket{\overline{\sigma}_i}$ and the first two registers collapse in the state $\ket{\ve{u}_i}\ket{\ve{v}_i}$.
It is possible to perform vector-state tomography on $\ket{\ve{u}_i}\ket{\ve{v}_i}$, using Theorem \ref{Theo:l2_vector_tomography} on the first register to retrieve $\overline{\ve{u}}_i$, or on the second one to retrieve $\overline{\ve{v}}_i$.
The costs are
$O(\frac{n\log{n}}{\delta^2})$
and
$O(\frac{m\log{m}}{\delta^2})$, respectively.
Using a coupon collector's argument \cite{erdHos1961couponcollector}, if the $k$ states
$\ket{\overline{\sigma}_i}$
are uniformly distributed, to get all the $k$ possible couples 
$\ket{\ve{u}_i}\ket{\ve{v}_i}$
at least once, we would need 
$k\log k$
measurements on average. 
This proves that it is possible to estimate all the singular values, factor scores and factor score ratios with the guarantees of Theorem \ref{Theo:singular_value_estimation} in time $\widetilde{O}(\frac{\norm{A}}{\theta}\frac{1}{\sqrt{p}}\frac{\mu(\m{A})k}{\epsilon})$.

To perform tomography on each state-vector, one should satisfy the coupon collector the same number of times as the measurements needed by the tomography procedure. 
The costs of the tomography for all the vectors $\{\overline{\ve{u}}_i\}_i^k$ and $\{\overline{\ve{v}}_i\}_i^k$ are
$O\left(T(U_{\ref{algTop-k:ampamp2}})\frac{k\log{k}\cdot n\log{n}}{\delta^2}\right)$, 
and 
$O\left(T(U_{\ref{algTop-k:ampamp2}})\frac{k\log{k}\cdot m\log{m}}{\delta^2}\right).$
Therefore, the following complexities are proven:
$\widetilde{O}\left(\frac{\norm{\m{A}}}{\theta}\frac{1}{\sqrt{p}}\frac{\mu(\m{A})}{\epsilon}\frac{kn}{\delta^2}\right), \widetilde{O}\left(\frac{\norm{\m{A}}}{\theta}\frac{1}{\sqrt{p}}\frac{\mu(\m{A})}{\epsilon}\frac{km}{\delta^2}\right).$
\end{proof}

In the appendix, Section \ref{SuppExperiments:run-time_param}, we provide experiments that show that the coupon collector's argument of Eq. \ref{Eq:coupon_collect} is accurate for practical $\epsilon$.  Besides $1/\sqrt{p}$ being negligible, it is interesting to note that the parameter $\theta$ can be computed using: \begin{enumerate*}
    \item the procedures of Theorems \ref{TheoMio:factor_score_estimation}, \ref{TheoMio:check_explained_variance};
    \item the binary search of Theorem \ref{Theorivisto:binarysearch};
    \item the available literature on the type of data stored in the input matrix $\m{A}$.
\end{enumerate*}
About the latter, the original paper of latent semantic indexing \cite{deerwester1990indexing} states that the first $k=100$ singular values are enough for a good representation. 
We believe that, in the same way, fixed thresholds $\theta$ can be defined for different machine learning applications.
The experiments of \citet{kerenidis2020mnist} on the run-time parameters of the polynomial expansions of the MNIST dataset support this expectation:
even though in qSFA they keep the $k$ smallest singular values and refer to $\theta$ as the biggest singular value to retain, this value does not vary much when the the dimensionality of their dataset grows.
In our experiments, we observe that different datasets for image classification have similar $\theta$s. For completeness, we also state a different version of Theorem \ref{TheoMio:top-k_sv_extraction}, with $\ell_\infty$ guarantees on the vectors.
\begin{corollary}[Fast top-k singular vectors extraction] 
\label{Prop:top-k_sv_linfty}
    The run-times of \ref{TheoMio:top-k_sv_extraction} can be improved to
    $\widetilde{O}\left(\frac{\norm{\m{A}}}{\theta}\frac{1}{\sqrt{p}}\frac{\mu(\m{A})}{\epsilon}\frac{k}{\delta^2}\right)$
    with estimation guarantees on the $\ell_\infty$ norms.
\end{corollary}
\begin{proof}
The proof consists in using $\ell_\infty$ tomography (Theorem \ref{Theo:linf_vector_tomography}) at step \ref{algTop-k:vector_tomography} of Algorithm \ref{alg_main:topk_sv}.
\end{proof}
Note that, given a vector with $d$ non-zero entries, performing $\ell_\infty$ tomography with error $\frac{\delta}{\sqrt{d}}$ provides the same guarantees of $\ell_2$ tomography with error $\delta$.
This implies that the extraction of the singular vectors with $\ell_2$ guarantees can be faster if we can make assumptions on their sparseness: $\widetilde{O}\left(\frac{\norm{\m{A}}}{\theta}\frac{1}{\sqrt{p}}\frac{\mu(\m{A})}{\epsilon}\frac{kd}{\delta^2}\right)$.

\section{Applications to machine learning}
\label{Section:applications}
The new quantum procedures can be used for principal component analysis, correspondence analysis, and latent semantic analysis.
Besides extracting the orthogonal factors and measuring their importance, we provide a procedure to represent the data in PCA's reduced feature space on a quantum computer. 
In a similar way, it is possible to compute the representations of CA and LSA.

\subsection{Principal Component Analysis}
\label{subsec:PCArelw}
Principal component analysis is a widely-used multivariate statistical method for continuous variables with applications in machine learning. 
Its uses range from outlier detection to dimensionality reduction and data visualization. 
Given a matrix $\m{A} \in \R^{n \times m}$ storing information about $n$ data points with $m$ coordinates, its \emph{principal components} are the set of orthogonal vectors along which the variance of the data points is maximized.
The goal of PCA is to compute the principal components with the amount of variance they capture and rotate the data points to express their coordinates along the principal components.
It is possible to represent the data using only the $k$ coordinates that express the most variance for dimensionality reduction.

\paragraph{PCA Model} The model of PCA is closely related to the singular value decomposition of the data matrix $\m{A}$, shifted to row mean 0. 
The model consists of the principal components and the amount of variance they explain.
The principal components coincide with the right singular vectors $\ve{v}_i$, the factor scores $\lambda_i = \sigma_i^2$ represent the amount of variance along each of them, and the factor score ratios $\lambda^{(i)}=\frac{\lambda_i}{\sum_j^r\lambda_j}$ express the percentage of retained variance.
For datasets with $0$ mean, the transformation consists in a rotation along the principal components: $\m{Y} = \m{A}\m{V} = \m{U}\Sig\m{V}^T\m{V} = \m{U}\Sig  \in \R^{n \times m}$.
When performing dimensionality reduction, it suffice to use the top $k$ singular values and vectors. \\

Using the procedures from Section \ref{Section:novel_quantum_subroutines} it is possible to extract the model for principal component analysis. In particular, 
Theorems \ref{TheoMio:factor_score_estimation}, \ref{TheoMio:check_explained_variance}, and \ref{Theorivisto:binarysearch} allow to retrieve information on the factor scores and on the factor score ratios, while Theorem \ref{TheoMio:top-k_sv_extraction} allows extracting the principal components.
The run-time of the model extraction is the sum of the run-times of the theorems: $\widetilde{O}\left(\left( \frac{1}{\gamma^2} + \frac{km}{\theta\delta^2}\right)\frac{\mu(\m{A})}{\epsilon}\right)$.
The model comes with the following guarantees:
$\abs{\sigma_i - \overline{\sigma}_i} \leq \frac{\epsilon}{2}$;
$\abs{\lambda_i - \overline{\lambda}_i} \leq \epsilon\sqrt{\lambda_i}$;
$\abs{\lambda^{(i)} - \overline{\lambda}^{(i)}} \leq \epsilon\frac{\sigma_i}{\norm{\m{A}}_f}$;
$\norm{\ve{v}_i - \overline{\ve{v}}_i} \leq \delta$ for $i \in \{0, \dots ,k-1\}$. 
This run-time is generally smaller than the number of elements of the input data matrix, providing polynomial speed-ups on the best classical routines for non-sparse matrices.
In writing the time complexity of the routines, we have omitted the term $\frac{1}{\sqrt{p}}$ because usually $p$ is chosen to be a number greater than 0.5  (generally in the order of 0.8/0.9). 

When performing dimensionality reduction, the goal is to obtain the matrix $\m{Y}=\m{U}\Sig \in \R^{n\times k}$, where $\m{U} \in \R^{n \times k}$ and $\Sig \in \R^{k \times k}$ are composed respectively of the top $k$ left singular vectors and singular values.
In Lemma \ref{Lemma:accuracyUSeVS}, we provide a theoretical error bound for $\m{Y}$, using the estimated entries of $\m{U}$ and $\Sig$.
For sake of completeness, the error bound is also stated for $\m{V}\Sig$.
These bounds stand regardless of how the singular values and vectors are extracted and hold when the multiplication is done with a classical computer.

\begin{lemma} [Accuracy of $\overline{\m{U}\Sig}$ and $\overline{\m{V}\Sig}$] 
\label{Lemma:accuracyUSeVS}
    Let $\m{A} \in \R^{n \times m}$ be a matrix. 
    Given some approximate procedures to retrieve estimates $\overline{\sigma}_i$ of the singular values $\sigma_i$ such that
    $\abs{\sigma_i - \overline{\sigma}_i} \leq \epsilon$
    and unit estimates 
    $\overline{\ve{u}}_i$
    of the left singular vectors 
    $\ve{u}_i$
    such that 
    $\norm{\overline{\ve{u}}_i - \ve{u}_i}_2 \leq \delta$,
    the error on 
    $\m{U}\Sig$
    can be bounded as $\norm{\m{U}\Sig - \overline{\m{U}}\overline{\Sig}}_F \leq \sqrt{\sum_j^k \left( \epsilon + \delta\sigma_j \right)^2}$.
    Similarly,
    $\norm{\m{V}\Sig - \overline{\m{V}}\overline{\Sig}}_F \leq \sqrt{\sum_j^k \left( \epsilon + \delta\sigma_j \right)^2}$.
    Both are bounded by $\sqrt{k}(\epsilon+\delta\norm{\m{A}})$.
\end{lemma}
We prove this result for $\norm{\m{U}\Sig - \overline{\m{U}}\overline{\Sig}}_F$. 
The proof for $\norm{\m{V}\Sig-\overline{\m{V}}\overline{\Sig}}_F$ is analogous.

\begin{proof}
We first bound the error on the columns:
\begin{align}
    \norm{\overline{\sigma}_i \overline{\ve{u}}_i - \sigma_i\ve{u}_i} \leq \norm{(\sigma_i \pm \epsilon)\overline{\ve{u}}_i - \sigma_i\ve{u}_i} = \norm{\sigma_i (\overline{\ve{u}}_i - \ve{u}_i) \pm \epsilon\overline{\ve{u}}_i }
\end{align}
Because of the triangular inequality, 
$\norm{\sigma_i (\overline{\ve{u}}_i - \ve{u}_i) \pm \epsilon\overline{\ve{u}}_i } \leq \sigma_i\norm{\overline{\ve{u}}_i - \ve{u}_i} +  \epsilon\norm{\overline{\ve{u}}_i }$. Also by hypothesis,  $\norm{(\overline{\ve{u}}_i - \ve{u}_i)} \leq \delta$ and $\norm{\overline{\ve{u}}_i} = 1$ . Thus,
$    \sigma_i\norm{ \overline{\ve{u}}_i - \ve{u}_i} +  \epsilon\norm{\overline{\ve{u}}_i } \leq \sigma_i \delta + \epsilon
$. 
Since $f(x) = \sqrt{x}$ is an increasing monotone function, it is possible to prove:
\begin{align}
    \norm{\overline{\m{U}}\overline{\Sig} - \m{U}\Sig}_F & = \sqrt{\sum_i^n\sum_j^k \norm{ \overline{\sigma}_j\overline{u}_{ij} - \sigma_j u_{ij}  }^2} = \sqrt{\sum_j^k \left( \norm{ \overline{\sigma}_j\overline{\ve{u}}_j - \sigma_j \ve{u}_j } \right)^2} \nonumber\\
    & \leq \sqrt{\sum_j^k \left( \epsilon + \delta\sigma_j \right)^2} \leq 
    \sqrt{k \left( \epsilon + \delta\sigma_{max} \right)^2} \leq
    \sqrt{k}(\epsilon + \delta\norm{\m{A}})
\end{align}
\end{proof}

Using matrix-multiplication from Theorem \ref{Theo:matrix_vec_mul}, we can have algorithms to produce quantum states proportional to the data representation in the new feature space.
Having access to $\m{V}^{(k)} \in \R^{m\times k}$, these routines create the new data points in almost constant time and are helpful when chained to other quantum machine learning algorithms that need to be executed multiple times.

\begin{corollary} [Quantum PCA: vector dimensionality reduction]
    \label{Coro:qPCAvector}
    Let $\xi$ be a precision parameter.
    Let there be efficient quantum access to the top k right singular vectors $\overline{\m{V}}^{(k)} \in \R^{m \times k}$ of a matrix $\m{A}= \m{U}\Sig\m{V}^T \in \R^{n \times m}$, such that $\norm{\m{V}^{(k)} - \overline{\m{V}}^{(k)}} \leq \frac{\xi}{\sqrt{2}}$. 
    Given quantum access to a row $\ve{a}_i$ of $\m{A}$, the quantum state $\ket{\overline{\ve{y}}_{i}} = \frac{1}{\norm{\overline{\ve{y}}_i}}\sum_i^k \overline{y}_k \ket{i}$, proportional to its projection onto the PCA space, can be created in time $\widetilde{O}\left(\mu(\m{V}^{(k)})\frac{\norm{\ve{a}_i}}{\norm{\overline{\ve{y}}_i}}\right)$ with probability at least $1-1/\text{poly}(m)$ and precision $\norm{\ket{\ve{y}_i} - \ket{\overline{\ve{y}}_i}} \leq \frac{\norm{\ve{a}_i}}{\norm{\overline{\ve{y}}_i}}\xi$.
    An estimate of $\norm{\overline{\ve{y}}_i}$, to relative error $\eta$, can be computed in $\widetilde{O}(1/\eta)$.
\end{corollary}
\label{Proof:qPCA_vector}
\begin{proof}
Here with $\m{V}$ we denote $\m{V}^{(k)} \in \R^{m\times k}$.
Given a vector $\ve{a}_i$, its projection onto the k-dimensional PCA space of $\m{A}$ is $\ve{y}_i^T = \ve{a}_i^T \m{V}$, or equivalently $\ve{y}_i = \m{V}^T\ve{a}_i$.
Note that $\norm{\ve{y}_i} = \norm{\m{V}^T\ve{a}_i}$. 

It is possible to use Theorem \ref{Theo:matrix_vec_mul} to multiply the quantum state $\ket{\ve{a}_i}$ by $\m{V}^T$, appropriately padded with 0s to be a square $\R^{m\times m}$ matrix.
In this way, we can create an approximation $\ket{\overline{\ve{y}}_i}$ of the state $\ket{\ve{y}_i} = \ket{\m{V}^T\ve{a}_i}$ in time $\widetilde{O}\left(\frac{\mu(\m{V}^T)\log(1/\epsilon)}{\gamma}\right)$ with probability $1-1/\text{poly}(m)$, such that $\norm{\ket{\ve{y}_i} - \ket{\overline{\ve{y}}_i}} \leq \epsilon$.
Since $\m{V}^T$ has rows with unit $\ell_2$ norm, we can prepare efficient quantum access to it by creating access to its rows \citep[Theorem IV.1]{kerenidis2020gradient}.
Having $\gamma = \norm{\m{V}^T\ve{a}_i}/\norm{\ve{a}_i}$, we get a run-time of $\widetilde{O}\left(\mu(\m{V})\frac{\norm{\ve{a}_i}}{\norm{\ve{y}_i}}\log(1/\epsilon)\right)$.  The term $\log(1/\epsilon)$ can be considered negligible. 
We conclude that the state $\ket{\ve{y}_i}$ can be created in time $\widetilde{O}\left(\mu(\m{V})\frac{\norm{\ve{a}_i}}{\norm{\ve{y}_i}}\right)$ with probability $1-1/\text{poly}(m)$ and that its norm can be estimated to relative error $\eta$ in time $\widetilde{O}\left(\mu(\m{V})\frac{\norm{\ve{a}_i}}{\norm{\ve{y}_i}}\frac{1}{\eta}\right)$.

For what concerns the error, we start by bounding $\norm{\ve{y}_i - \overline{\ve{y}}_i}$ and then use Claim \ref{claim:vector_quantum_distance} to bound the error on the quantum states.
Assume to have estimates $\overline{\ve{v}}_i$ of the columns of $\m{V}$ such that $\norm{\ve{v}_i - \overline{\ve{v}}_i} \leq \delta$.
\begin{align}
    \norm{ \m{V} - \overline{\m{V}} }_F =
    \sqrt{\sum_i^n\sum_j^k \left(v_{ij} - \overline{v}_{ij} \right)^2} 
     \leq
    \sqrt{k}\delta
\end{align}

Considering that 
$\norm{\ve{y}_i - \overline{\ve{y}}_i} = \norm{\ve{a}_i^T\m{V}^{(k)} - \ve{a}_i^T\overline{\m{V}}^{(k)}} \leq
\norm{\ve{a}_i}\sqrt{k}\delta$,
we can use Claim \ref{claim:vector_quantum_distance} to state
\begin{align}
    \norm{\ket{\ve{y}_i} - \ket{\overline{\ve{y}}_i}} \leq
    \frac{\norm{\ve{a}_i}}{\norm{\ve{y}_i}}\sqrt{2k}\delta = \frac{\norm{\ve{a}_i}}{\norm{\ve{y}_i}}\xi.
\end{align}
Setting $\delta = \frac{\xi}{\sqrt{2k}}$ leads to the requirement $\norm{ \m{V} - \overline{\m{V}} }_F \leq \frac{\xi}{\sqrt{2}}$.
\end{proof}

This result also holds when $\ve{a}_i$ is a previously unseen data point, not necessarily stored in $\m{A}$.
Note that from the row orthogonality of $\m{V}^{(k)}$ it follows that $\mu(\m{V}^{(k)}) \leq \norm{\m{V}^{(k)}}_F = \sqrt{k}$.
Furthermore, $\frac{\norm{\ve{y}_i}}{\norm{\ve{a}_i}}$ is expected to be close to $1$, as it is the percentage of support of $\ve{a}_i$ on the new feature space spanned by $\m{V}^{(k)}$. 
We formalize this better using Definition \ref{def:pca_assumption} below.

\begin{definition}[PCA-representable data]
\label{def:pca_assumption}
A set of $n$ data points described by $m$ coordinates, represented through a matrix $\m{A}=\sum_i^r\sigma_i \ve{u}_i\ve{v}_i^T \in \R^{n \times m}$
is said to be PCA-representable if there exists $p \in [\frac{1}{2},1], \varepsilon \in [0, 1/2], \beta \in [p-\varepsilon, p+\varepsilon], \alpha \in [0,1]$ such that:
\begin{itemize}
    \item $\exists k \in O(1)$ such that $\frac{\sum_i^k \sigma^2_i}{\sum_i^m \sigma^2_i}= p$
    \item for at least $\alpha n$ points $\ve{a}_i$ it holds $\frac{\norm{\ve{y}_i}}{\norm{\ve{a}_i}} \geq \beta$, where $\norm{\ve{y}_i} = \sqrt{\sum_i^k \abs{\braket{\ve{a}_i}{\ve{v}_j}}^2} \norm{\ve{a}_i}$. 
\end{itemize}
\end{definition}

\begin{claim}[Quantum PCA on PCA-representable datasets]
\label{claim:pcarepr}
Let $\ve{a}_i$ be a row of $\m{A} \in  \R^{n\times d}$. Then, for $p \in [1/2, 1]$, the run-time of Corollary \ref{Coro:qPCAvector} is $\mu(\m{V})\frac{\norm{\ve{a}_i}}{\norm{\overline{\ve{y}}_i}} = \mu(\m{V})\frac{1}{\beta} = O(\mu(\m{V}))$ with probability greater than $\alpha$.
\end{claim}
It is known that, in practical machine learning datasets, $\alpha$ is a number fairly close to one.
We have tested the value of $\alpha$ for the MNIST, Fashion MNIST and CIFAR-10 datasets, finding values over 0.85 for any $p \in (0,1]$.

The next corollary shows how to perform  perform dimensionality reduction on  the whole matrix, enabling quantum access to the data matrix in the reduced feature space.

\begin{corollary} [Quantum PCA: matrix dimensionality reduction]
\label{Coro:qPCAmatrix}
    Let $\xi$ be a precision parameter and $p$ be the amount of variance retained after the dimensionality reduction. 
    Let there be efficient quantum access to $\m{A}= \m{U}\Sig\m{V}^T \in \R^{n \times m}$ and to its top k right singular vectors $\overline{\m{V}}^{(k)} \in \R^{m \times k}$, such that $\norm{\m{V}^{(k)} - \overline{\m{V}}^{(k)}} \leq \frac{\xi\sqrt{p}}{\sqrt{2}}$. 
    There exists a quantum algorithm that, with probability at least $1-1/\text{poly}(m)$, creates the state 
     $\ket{\overline{\m{Y}}} = \frac{1}{\norm{\m{Y}}_F}\sum_i^n \norm{\ve{y}_{i,\cdot}}\ket{i}\ket{\ve{y}_{i,\cdot}}$,
    proportional to the projection of $\m{A}$ in the PCA subspace, 
    with error $\norm{\ket{\m{Y}} - \ket{\overline{\m{Y}}}} \leq \xi$ in time $\widetilde{O}(\mu(\m{V})/\sqrt{p})$.
    An estimate of $\norm{\overline{\m{Y}}}_F$, to relative error $\eta$, can be computed in $\widetilde{O}(\frac{\mu(\m{V})}{\sqrt{p}\eta})$.
\end{corollary}
\label{Proof:qPCA_matrix}
\begin{proof}
Here with $\m{V}$ we denote $\m{V}^{(k)} \in \R^{m\times k}$.
Using the same reasoning as the proof above and giving a closer look at the proof of Theorem \ref{Theo:matrix_vec_mul} (Lemma 24 \cite{chakraborty2018power}), we see that it is possible to create the state $\ket{0}(\frac{\overline{\m{V}}^T}{\mu(\m{V})}\ket{\ve{a}_i}) + \ket{0_\perp}$ in time $\widetilde{O}(1)$ and that the term $\frac{\mu(\m{V})}{\gamma}$ is introduced to boost the probability of getting the right state. 
Indeed, if we apply Theorem \ref{Theo:matrix_vec_mul} without the amplitude amplification step to the superposition of the rows of $\m{A}$, we obtain the following mapping in time $\widetilde{O}(1)$:
\begin{align}
    \ket{\m{A}} = \frac{1}{\norm{\m{A}}_F} \sum_i^n \norm{\ve{a}_{i,\cdot}}\ket{i} \ket{\ve{a}_{i,\cdot}} \mapsto
    \frac{1}{\norm{\m{A}}_F\mu(\m{V})} \sum_i^n (\norm{\ve{y}_{i,\cdot}}\ket{0}\ket{i}\ket{\ve{y}_{i,\cdot}} +  \norm{\ve{y}_{i,\cdot\perp}}\ket{0_\perp}),
\end{align}
where $\norm{\ve{y}_{i,\cdot\perp}}$ are normalization factors.
Keeping in mind that $\norm{\m{A}}_F = \sqrt{\sum_i^r \sigma_i^2}$ and $\norm{\m{Y}}_F = \sqrt{\sum_i^n \norm{\ve{y}_{i,\cdot}}^2} = \sqrt{\sum_i^k \sigma_i^2}$, we see that the amount of explained variance is $p = \frac{\sum_i^k \sigma_i^2}{\sum_j^r \sigma_j^2}=\left(\frac{\norm{\m{Y}}_F}{\norm{\m{A}}_F}\right)^2$.
The probability of obtaining $\ket{Y} = \frac{1}{\norm{\m{Y}}_F} \sum_i^n \norm{\ve{y}_{i,\cdot}}\ket{i}\ket{\ve{y}_{i,\cdot}}$ is $\frac{p}{\mu(\m{V})^2} = \frac{\norm{\m{Y}}_F^2}{\norm{\m{A}}_F^2}\frac{1}{\mu(\m{V})^2} = \frac{\sum_i^n \norm{\ve{y}_{i,\cdot}}^2}{\norm{\m{A}}_F^2\mu(\m{V})^2}$.
We conclude that, using $\widetilde{O}(\mu(\m{V})/\sqrt{p})$ rounds of amplitude amplification, we obtain $\ket{\m{Y}}$ with probability $1-1/\text{poly}(m)$ (Theorem \ref{Theo:ampamp}). For the error, consider that 
$\norm{\m{Y} - \overline{\m{Y}}} = \norm{\m{A}\m{V}^{(k)} - \m{A}\overline{\m{V}}^{(k)}} \leq
\norm{\m{A}}\sqrt{k}\delta$,
so we can use Claim \ref{claim:vector_quantum_distance} to state
\begin{align}
    \norm{\ket{\m{Y}} - \ket{\overline{\m{Y}}}} \leq
    \frac{\norm{\m{A}}_F}{\norm{\m{Y}}_F}\sqrt{2k}\delta = \xi.
\end{align}
We can set $\delta = \frac{\xi}{\sqrt{2k}}\frac{\norm{\m{Y}}_F}{\norm{\m{A}}_F} = \frac{\xi\sqrt{p}}{\sqrt{2k}}$, so we require $\norm{ \m{V} - \overline{\m{V}} }_F \leq \frac{\xi\sqrt{p}}{\sqrt{2}}$.
\end{proof}

The error requirements of the two corollaries propagate to the run-time of the model extraction in the following way.

\begin{corollary} [Quantum PCA: fitting time]
\label{Coro:qPCAtraining}
    Let $\epsilon$ be a precision parameter and $p=\frac{\sum_{i: \overline{\sigma}_i \geq \theta} \sigma_i^2}{\sum_j^r \sigma_j^2}$ the amount of variance to retain, where $\abs{\sigma_i - \overline{\sigma}_i}\leq \epsilon$.
    Given efficient quantum access to a matrix $\m{A} \in \R^{n \times m}$,
    the run-time to extract $\m{V}^{(k)} \in \R^{m \times k}$ for corollaries \ref{Coro:qPCAvector}, \ref{Coro:qPCAmatrix} is $\widetilde{O}\left(\frac{\mu(\m{A})k^2m}{\theta\epsilon\xi^2}\right)$.
\end{corollary}

\label{Proof:qPCA_fittingtime}
\begin{proof}
The procedure to train the model consists in using Theorem \ref{Theorivisto:binarysearch} or \ref{TheoMio:factor_score_estimation} to extract the threshold $\theta$, given the amount of variance to retain $p$, and to leverage Theorem \ref{TheoMio:top-k_sv_extraction} to extract the $k$ right singular vectors that compose $\m{V} \in \R^{m \times k}$.
The run-time of Theorem \ref{Theorivisto:binarysearch} and \ref{TheoMio:factor_score_estimation} are smaller than the one of Theorem \ref{TheoMio:top-k_sv_extraction}, so we can focus on the last one.
To have $\norm{\m{V} - \overline{\m{V}} }_F \leq \frac{\xi\sqrt{p}}{\sqrt{2}}$ we need $\norm{\ve{v}_i - \overline{\ve{v}}_i} \leq \frac{\xi\sqrt{p}}{\sqrt{2k}}$.
Substituting $\delta = \frac{\xi\sqrt{p}}{\sqrt{2k}}$ in the run-time of Theorem \ref{TheoMio:top-k_sv_extraction}, we get $\widetilde{O}(\frac{\mu(\m{A})k^2m}{p^{3/2}\theta\epsilon\xi^2})$.
If we consider that $p$ to be a reasonable number (e.g., at least grater than 0.05), we can consider it a constant factor that is independent from the input's size.
The asymptotic run-time is proven to be $\widetilde{O}(\frac{\mu(\m{A})k^2m}{\theta\epsilon\xi^2})$.
\end{proof}

When training the model for Corollary \ref{Coro:qPCAmatrix}, the run-time has a dependency on $1/p^{3/2}$. 
However, this term is constant and independent from the size of the input dataset.
With this additional $1/p^{3/2}$ cost, the error of Corollary \ref{Coro:qPCAvector} drops to $\xi$ for every row of the matrix and generally decreases in case of new  data points.

Using the same framework and proof techniques, it is possible to produce similar results for the representations of CA and LSA.

\emph{Remark:} Note that \citet[Theorem 1]{yu2019quantumdatacompression} propose a lower bound for a quantity similar to our $\alpha$. 
However, their result seems to be a loose bound: using their notation and setting $\eta=1, \theta=1$ they bound this quantity with $0$, while a tight bound should give $1$.


\subsection{Correspondence analysis}
\label{subsec:CA}
Correspondence analysis is a multivariate statistical tool from the family of \emph{factor analysis} methods.
It is used to explore relationships among categorical variables.
Given two random variables, $X$ and $Y$, with possible outcomes in $\{x_1, \cdots, x_n\}$ and $\{y_1, \cdots, y_m\}$, the model of Correspondence Analysis enables representing the outcomes as vectors in two related Euclidean spaces.
These vectors can be used for data visualization, exploration, and other unsupervised machine learning tasks.

\paragraph{Model}
Given a contingency table for $X$ and $Y$ (see Section~\ref{sec:preliminaries}), it is possible to compute the matrix $\m{A} = \m{D}_X^{-1/2}(\hat{\m{P}}_{X,Y} - \hat{\ve{p}}_X\hat{\ve{p}}_Y^T)\m{D}_Y^{-1/2} \in \R ^{n \times m}$, where $\hat{\m{P}}_{X,Y} \in \R^{n \times m}$ is the estimated matrix of joint probabilities, $\hat{\ve{p}}_X \in R^n$ and $\hat{\ve{p}}_X \in \R^m$ are the vectors of marginal probabilities, and $\m{D}_X^{-1/2} = diag(\hat{\ve{p}}_X)$, $\m{D}_Y^{-1/2} = diag(\hat{\ve{p}}_Y)$.
The computation of $\m{A}$ requires linear time in the non-zero entries of the contingency table.
The singular value decomposition of $\m{A}$ is strictly related to the model of correspondence analysis \cite{greenacre1984corranalysis, hsu2019correspondence}.
The new coordinates of $X$'s outcomes are given by the rows of $\m{D}_X^{-1/2}\m{U} \in \R^{n\times k}$, while the ones of $Y$ by the rows of $\m{D}_Y^{-1/2}\m{V}\in \R^{m\times k}$.
Like in PCA, it is possible to choose only a subset of the orthogonal factors as coordinates for the representation.
Factor scores and factor score ratios measure of how much ``correspondence'' is captured by the respective orthogonal factor, giving an estimate of the quality of the representation. \\

Similarly to what we have already discussed, it is possible to extract the model for CA by creating quantum access to the matrix $\m{A}$ and using Theorems \ref{TheoMio:factor_score_estimation}, \ref{TheoMio:check_explained_variance}, and \ref{TheoMio:top-k_sv_extraction} to extract the orthogonal factors, the factor scores and the factor score ratios in time $\widetilde{O}\left(\left( \frac{1}{\gamma^2} + \frac{k(n+m)}{\theta\delta^2}\right)\frac{\mu(\m{A})}{\epsilon}\right)$.
We provide a theoretical bound for the data representations in Lemma \ref{Lemma:accuracyDUeDV}.

\begin{lemma} [Accuracy of $\m{D}_{X}^{-1/2}\m{U}$ and $\m{D}_{Y}^{-1/2}\m{V}$]
\label{Lemma:accuracyDUeDV}
    Let $\m{A} \in \R^{n \times m}$ be a matrix. 
    Given some approximate procedures to retrieve unit estimates 
    $\overline{\ve{u}}_i$
    of the left singular vectors 
    $\ve{u}_i$
    such that 
    $\norm{\overline{\ve{u}}_i - \ve{u}_i} \leq \delta$,
    the error on 
    $\m{D}_{X}^{-1/2}\m{U}$
    can be bounded as  $\norm{ \m{D}_{X}^{-1/2}\m{U}- \m{D}_{X}^{-1/2}\overline{\m{U}}}_F \leq \norm{\m{D}_{X}^{-1/2}}_F\sqrt{k}\delta$.
    Similarly, $\norm{\m{D}_{Y}^{-1/2}\m{V} - \m{D}_{Y}^{-1/2}\overline{\m{V}}}_F \leq \norm{\m{D}_{Y}^{-1/2}}_F\sqrt{k}\delta.$
\end{lemma}
\begin{proof}
It suffices to note that
$\norm{\m{D}_X^{-1/2}\overline{\m{U}} - \m{D}_X^{-1/2}\m{U}}_F \leq 
    \norm{\m{D}_X^{-1/2}}_F\norm{\overline{\m{U}} - \m{U}}_F
    \leq
    \norm{\m{D}_X^{-1/2}}_F\sqrt{k}\delta$. Similar conclusions can be drawn for $\norm{\m{D}_{Y}^{-1/2}\m{V} - \m{D}_{Y}^{-1/2}\overline{\m{V}}}_F$.
\end{proof}

\subsection{Latent semantic analysis}
\label{subsec:LSArelw}
Latent semantic analysis is a data representation method used to represent words and text documents as vectors in Euclidean spaces. 
Using these vector spaces, it is possible to compare terms, documents, and terms and documents.
LSA spaces automatically model synonymy and polysemy \cite{deerwester1990indexing}, and their applications in machine learning range from topic modeling to document clustering and retrieval.

\paragraph{Model} The input of LSA is a contingency table of $n$ words and $m$ documents $\m{A} \in \R^{n \times m}$. 
Inner products of rows are a measure of words similarity, and can be computed at once as  $\m{A}\m{A}^T=\m{U}\Sig^2\m{U}^T$. 
Inner products of columns $\m{A}^T\m{A}=\m{V}\Sig^2\m{V}^T$ are a measure of documents similarity, and the $a_{ij}$ entry of $\m{A}=\m{U}\Sig\m{V}^T$ is a measure of similarity between word $i$ and document $j$.
We can use SVD to express words and documents in new spaces where we can compare them with respect to this similarity measure. 
In particular, we can compute: \begin{enumerate*}
    \item a representation for word comparisons $\m{U}\Sig \in \R^{n\times k}$;
    \item a representation for document comparisons $\m{V}\Sig \in \R^{m\times k}$;
    \item two representations for word and document comparisons $\m{U}\Sig^{1/2} \in \R^{n\times k}$ and $\m{V}\Sig^{1/2} \in \R^{m\times k}$.
\end{enumerate*} 
When using LSA for document indexing, like in a search engine, we need to represent the query as a vector in the document space.
In this case, instead of increasing $\m{A}$'s size and recomputing the document space, the new vector can be expressed as $\ve{v}_q^T = \ve{x}_q^T\m{U}\Sig^{-1}$, where $\ve{x}_q \in \R^{n}$ is obtained using the same criteria used to store a document in $\m{A}$.
The representation of the query can then be used to compare the query to the other documents in the document representation space.
Finally, factor score ratios play an important role in LSA too. For instance, the columns of $V$ can be seen as latent topics of the corpus. The importance of each topic is proportional to the corresponding factor score ratio.
This paragraph only stresses how computing the SVD of $\m{A}$ is connected to LSA.
For a better introduction to LSA and indexing, we invite the reader to consult the original paper \cite{deerwester1990indexing}.\\

Even in this case, the cost of extracting the orthogonal factors and the factor scores is bounded by $\widetilde{O}\left(\left( \frac{1}{\gamma^2} + \frac{k(n+m)}{\theta\delta^2}\right)\frac{\mu(\m{A})}{\epsilon}\right)$.
In some applications, the data analyst might use a fixed number of singular values and vectors, regardless of the factor score ratios. 
In \citet{deerwester1990indexing}, $k=100$ is found to be a good number for document indexing.
Similarly, we believe that if we scale the singular values by the spectral norm, it is possible to empirically determine a threshold $\theta$ to use in practice.
Determining such threshold would reduce the complexity of model computation to the one of Theorem \ref{TheoMio:top-k_sv_extraction}: $\widetilde{O}\left( \frac{k(n+m)}{\theta\delta^2}\frac{\mu(\m{A})}{\epsilon}\right)$.

For what concerns the error bounds, we already know that it is possible to retrieve an approximation $\overline{\m{U}\Sig}$ and $\overline{\m{V}\Sig}$ with precision $\sqrt{k}(\epsilon+\delta\norm{\m{A}})$ (Lemma \ref{Lemma:accuracyUSeVS}), where $\delta$ is the precision on the singular vectors and $\epsilon$ the precision on the singular values. 
To provide bounds on the estimations of $\m{U}\Sig^{1/2}$, $\m{V}\Sig^{1/2}$, and $\m{U}\Sig^{-1}$ we introduce Lemma \ref{Lemma:accuracyUS12eVS12} and Lemma \ref{Lemma:accuracyUE-1eVE-1}.
\begin{lemma} [Accuracy of $\overline{\m{U}\Sig}^{1/2}$ and $\overline{\m{V}\Sig}^{1/2}$] 
\label{Lemma:accuracyUS12eVS12}
    Let $\m{A} \in \R^{n \times m}$ be a matrix. 
    Given some approximate procedures to retrieve estimates $\overline{\sigma}_i$ of the singular values $\sigma_i$ such that
    $\abs{\overline{\sigma}_i - \sigma_i} \leq \epsilon$
    and unitary estimates 
    $\overline{\ve{u}}_i$
    of the left singular vectors 
    $\ve{u}_i$
    such that 
    $\norm{\overline{\ve{u}}_i - \ve{u}_i} \leq \delta$,
    the error on 
    $\m{U}\Sig^{1/2}$
    can be bounded as
    $\norm{\m{U}\Sig^{1/2} - \overline{\m{U}}\overline{\Sig}^{1/2}}_F \leq \sqrt{\sum_j^k \left(\delta\sqrt{\sigma_j} + \frac{\epsilon}{2\sqrt{\theta}}\right)^2}$.
    Similarly,
    $\norm{\m{V}\Sig^{1/2} - \overline{\m{V}}\overline{\Sig}^{1/2} }_F \leq \sqrt{\sum_j^k \left(\delta\sqrt{\sigma_j} + \frac{\epsilon}{2\sqrt{\theta}}\right)^2}$. Both are bounded by $\sqrt{k}\left(\delta\sqrt{\norm{\m{A}}} + \frac{\epsilon}{2\sqrt{\theta}}\right)$
\end{lemma}
We prove this result for $\norm{\overline{\m{U}}\overline{\Sig}^{1/2} - \m{U}\Sig^{1/2}}_F$.

\begin{proof} 

We start by bounding $\abs{\sqrt{\overline{\sigma}_i} - \sqrt{\sigma_i}}$. Let's define $\epsilon = \gamma \sigma_i$ as a relative error:
\begin{align}
    \abs{\sqrt{\sigma_i + \epsilon} - \sqrt{\sigma_i}} & =
    \abs{\sqrt{\sigma_i + \gamma \sigma_i} - \sqrt{\sigma_i}} =
    \abs{\sqrt{\sigma_i}(\sqrt{1 + \gamma} - 1)} \nonumber\\
    &= \sqrt{\sigma_i}\abs{\frac{(\sqrt{1 + \gamma} - 1)(\sqrt{1 + \gamma} + 1)}{\sqrt{1 + \gamma} + 1}} \nonumber\\ 
    &=
    \sqrt{\sigma_i}\abs{\frac{\gamma + 1 - 1}{\sqrt{1 + \gamma} + 1}} \leq
    \sqrt{\sigma_i}\frac{\gamma}{2}.
\end{align}
By definition $\gamma = \frac{\epsilon}{\sigma_i}$ and we know that $\sigma_{min} \geq \theta$:
\begin{align}
    \abs{\sqrt{\overline{\sigma}_i} - \sqrt{\sigma_i}} \leq \frac{\sqrt{\sigma_i}}{\sigma_1}\frac{\epsilon}{2} = \frac{\epsilon}{2\sqrt{\sigma_i}} \leq \frac{\epsilon}{2\sqrt{\theta}}.
\end{align}

Using the bound on the square roots, we can bound the columns of $\overline{\m{U}}\overline{\Sig}^{1/2}$:
\begin{align}
    \left\lVert\sqrt{\overline{\sigma_i}}\overline{\ve{u}}_i - \sqrt{\sigma_i}\ve{u}_i\right\rVert &\leq
    \left\lVert\left(\sqrt{\sigma_i} + \frac{\epsilon}{2\sqrt{\theta}}\right)\overline{\ve{u}}_i - \sqrt{\sigma_i}\ve{u}_i\right\rVert = \nonumber\\
    \left\lVert\sqrt{\sigma_i}(\overline{\ve{u}}_i - \ve{u}_i) +  \frac{\epsilon}{2\sqrt{\theta}}\overline{\ve{u}}_i\right\rVert &\leq 
    \sqrt{\sigma_i}\delta + \frac{\epsilon}{2\sqrt{\theta}}
\end{align}

From the error bound on the columns we derive the bound on the matrices:
\begin{align}
    \left\lVert\overline{\m{U}}\overline{\Sig}^{1/2} - \m{U}\Sig^{1/2}\right\rVert_F 
    &= 
    \sqrt{\sum_j^k \left( \left\lVert \sqrt{\overline{\sigma}_j}\overline{\ve{u}}_j - \sqrt{\sigma_j} \ve{u}_j \right\rVert \right)^2} \nonumber \\ 
    &\leq 
    \sqrt{\sum_j^k \left(\delta\sqrt{\sigma_j} + \frac{\epsilon}{2\sqrt{\theta}}\right)^2}
    \leq
    \sqrt{k}\left(\delta\sqrt{\norm{\m{A}}} + \frac{\epsilon}{2\sqrt{\theta}}\right).
\end{align}
\end{proof}

\begin{lemma} [Accuracy of $\overline{\m{U}\Sig}^{-1}$ and $\overline{\m{V}\Sig}^{-1}$] 
\label{Lemma:accuracyUE-1eVE-1}
    Let $\m{A} \in \R^{n \times m}$ be a matrix. 
    Given some approximate procedures to retrieve estimates $\overline{\sigma}_i$ of the singular values $\sigma_i$ such that
    $\abs{\overline{\sigma}_i - \sigma_i} \leq \epsilon$
    and unitary estimates 
    $\overline{\ve{u}}_i$
    of the left singular vectors 
    $\ve{u}_i$
    such that 
    $\norm{\overline{\ve{u}}_i - \ve{u}_i} \leq \delta$,
    the error on
    $\m{U}\Sig^{-1}$
    can be bounded as
    $\norm{ \m{U}\Sig^{-1} - \overline{\m{U}}\overline{\Sig}^{-1} }_F \leq \sqrt{k}\left(\frac{\delta}{\theta} + \frac{\epsilon}{\theta^2 - \theta\epsilon}\right)$.
    Similarly,
    $\norm{ \m{V}\Sig^{-1} - \overline{\m{V}}\overline{\Sig}^{-1} }_F \leq \sqrt{k}(\frac{\delta}{\theta} + \frac{\epsilon}{\theta^2 - \theta\epsilon})$.
\end{lemma}

We prove this result for  $\norm{\overline{\m{U}}\overline{\Sig}^{-1} - \m{U}\Sig^{-1}}_F$.

\begin{proof}
We start by bounding $\abs{ \frac{1}{\overline{\sigma}_i} - \frac{1}{\sigma_i}}$. 
Knowing that $\sigma_{min} \geq \theta$ and $\epsilon < \theta$:
\begin{align}
    \abs{\frac{1}{\overline{\sigma}_i} - \frac{1}{\sigma_i}} \leq 
    \abs{\frac{1}{\sigma_i - \epsilon} - \frac{1}{\sigma_i}} \leq
    \frac{\epsilon}{\theta^2 - \theta\epsilon}.
\end{align}

From the bound on the inverses, we can obtain the bound on the columns of $\overline{\m{U}}\overline{\Sig}^{-1}$:
\begin{align}
    \norm{\frac{1}{\overline{\sigma_i}}\overline{\ve{u}}_i - \frac{1}{\sigma_i}\ve{u}_i} \leq 
    \norm{\left(\frac{1}{\sigma_i} \pm \frac{\epsilon}{\theta^2 - \theta\epsilon}\right)\overline{\ve{u}}_i - \frac{1}{\sigma_i}\ve{u}_i} \leq 
    \frac{1}{\sigma_i}\delta + \frac{\epsilon}{\theta^2 - \theta\epsilon} \leq 
    \frac{\delta}{\theta} + \frac{\epsilon}{\theta^2 - \theta\epsilon}.
\end{align}

To complete the proof, we compute the bound on the matrices:
\begin{align}
    \norm{\overline{\m{U}}\overline{\Sig}^{-1} - \m{U}\Sig^{-1}}_F = 
    \sqrt{\sum_j^k \left( \norm{\frac{1}{\overline{\sigma}}_j\overline{\ve{u}}_{j} - \frac{1}{\sigma_j} \ve{u}_{j} } \right)^2} \leq 
    \sqrt{k}\left(\frac{\delta}{\theta} + \frac{\epsilon}{\theta^2 - \theta\epsilon}\right)
\end{align}
\end{proof}

\section{Experiments} \label{Section:experiments}
All of our experiments are numerical and can be carried out on classical computers. \footnote{The code of the experiments is available at https://github.com/ikiga1/qadra.}
We have analysed the distribution of the factor score ratios in the MNIST, Fashion MNIST, CIFAR-10, Tiny Imagenet and Research Papers datasets. 
They decrease exponentially fast (figures in the appendix), confirming the low rank nature of the data.
Focusing on MNIST, Fashion-MNIST, and CIFAR-10, we have simulated PCA's dimensionality reduction for image classification. 
The datasets have been shifted to row mean 0 and normalized so that $\sigma_{max}=1$.
We have simulated Algorithm \ref{alg_main:factor_score_estimation} by sampling $1/\gamma^2 = 1000$ times from the state $\sum_i^r\lambda_i \ket{\ve{u}_i} \ket{\ve{v}_i} \ket{\overline{\sigma}_i}$ to search the first $k$ principal components that account for a factor score ratios sum $p=0.85$.
The simulation occurs by sampling with replacement from the discrete probability distribution given by the $\lambda_i$. 
We then estimated the measured $\lambda_i$ using the Wald estimator (see the proof of Theorem \ref{TheoMio:factor_score_estimation}) and searched for the most important $k$.\footnote{Note that, in practice, one could also estimate the factor score ratios as $\overline{\lambda}_i = \frac{\overline{\sigma}_i}{\norm{A}_F}$. This method should require less measurements: a bound on the necessary number of measurements can be obtained via the coupon collector's problem with non-uniform probabilities.}
In all cases, sampling the singular values has been enough to decide how many to keep.
However, as $p$ increases, the gap between the factor score ratios decreases and the quality of the estimation of $k$ or $\theta$ decreases.
As discussed in Section \ref{Subsec:quality_repr}, it is possible to detect this problem using Theorem \ref{TheoMio:check_explained_variance} and solve it with a binary search for $\theta$ (Theorem \ref{Theorivisto:binarysearch}).
We have tested the quality of the representation by observing the accuracy of 10-fold cross-validation k-nearest neighbors with $k=7$ as we introduce error in the representation's Frobenius norm (see Figure \ref{Fig:MNIST_Classification}).
To introduce the error, we have added truncated Gaussian noise to each element of $\m{U}\Sig$ to have $\norm{\m{U}\Sig - \overline{\m{U}}\overline{\Sig}} \leq \xi = \sqrt{k}(\epsilon+\delta)$ (Lemma \ref{Lemma:accuracyUSeVS}).
The parameter $\delta$ has been estimated using the bound above, choosing the error so that the accuracy drops no more than $0.01$ and fixing $\epsilon$ to a number that allows for correct thresholding.
Table \ref{table:parameters} summarizes the run-time parameters.
The results show that Theorems \ref{TheoMio:factor_score_estimation}, \ref{TheoMio:check_explained_variance}, \ref{Theorivisto:binarysearch} are already advantageous on small datasets, while Theorem \ref{TheoMio:top-k_sv_extraction} requires bigger datasets to express its speed-up.
We have also simulated the creation of the state at step \ref{algTop-k:ampamp2} of Algorithm \ref{alg_main:topk_sv} to test the average number of measurements needed to collect all the singular values as $\epsilon$ increases.
The analysis has confirmed the run-time's expectations.
To end with, we have tested the value of $\alpha$ (Definition \ref{def:pca_assumption}, Claim \ref{claim:pcarepr}) for the MNIST dataset, fixing $\varepsilon=0$ and trying $p \in \{0.1,0.2,0.3,0.4,0.5,0.6,0.7,0.8,0.9\}$. 
We have observed that $\alpha = 0.97 \pm 0.03$, confirming that the run-time of Corollary \ref{Coro:qPCAvector} can be assumed $\widetilde{O}(\mu(\m{V})^{(k)})$ for the majority of the data points of a PCA-representable dataset.

We point out that more experiments on the run-time parameters have been extensively discussed in other works that rely on the same parameters \cite{kerenidis2020mnist, kerenidis2020gaussian}. 
These works study the scaling of the parameters as the dataset size increases, both in features and samples, and conclude that the parameters of interest are almost constant.
In addition to the existing experiments, we have studied the trend of the run-time parameters on the Tiny Imagenet dataset as the number of samples scales.
While the spectral norm increases, the other run-time parameters become constant after a certain number of samples.
Figure \ref{Fig:Imagenet_scaling} shows that the algorithms discussed in Section \ref{Subsec:quality_repr} are already of practical use for small datasets, while the singular vector extraction routines of Section \ref{Subsec:extraction} require larger datasets to be convenient over their classical counterparts.
We refer the interested reader to the appendix for more details about the experiments.

\begin{figure}[t]
    \begin{center}
        \centerline{\includegraphics[width=0.5\columnwidth]{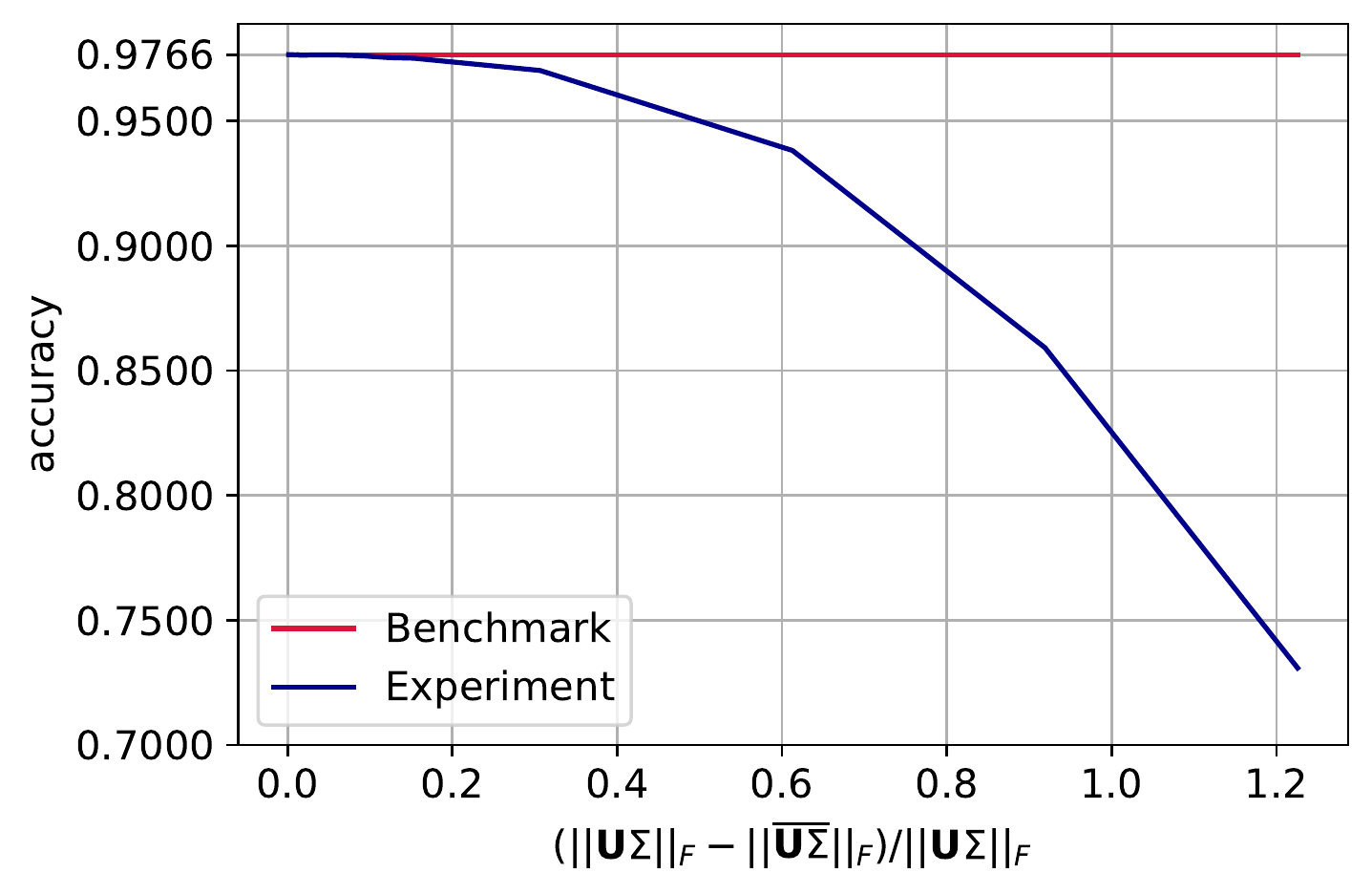}}
        \caption{Accuracy of 10-fold cross-validation using K-Nearest-Neighbors, with 7 neighbors, on the MNIST dataset after PCA's dimensionality reduction (0.8580\% of variance retained). The \emph{benchmark} accuracy was computed with an exact PCA. The \emph{experiment} line shows how the classification accuracy decreases as error is introduced in the Frobenius norm of the representation.}
        \label{Fig:MNIST_Classification}
    \end{center}
\end{figure}

\begin{table}[t]
\caption{Summary of the run-time parameters. The parameters that depend on $k$ have been computed using the estimated $k$.}

\label{table:parameters}
    \begin{center}
        \begin{small}
            \begin{sc}
                \begin{tabular}{ccccr}
                    \toprule
                    Parameter & MNIST & F-MNIST & CIFAR-10 \\
                    \midrule
                    $\mu(\m{A}) = \norm{\m{A}}_F$ & 3.2032 & 1.8551 & 1.8540\\
                    Estimated $k$ & 62 & 45& 55\\
                    Exact $k$ & 59 & 43& 55\\
                    Estimated $p$ & 0.8510 & 0.8510& 0.8510\\
                    Exact $p$ & 0.8580 & 0.8543& 0.8514\\
                    Thrs. $\epsilon$ & 0.0030 & 0.0009& 0.0006\\
                    $\theta$ & 0.1564 & 0.0776& 0.0746\\
                    $\delta$ & 0.1124 & 0.0106& 0.0340\\
                    \bottomrule
                \end{tabular}
            \end{sc}
        \end{small}
    \end{center}
\end{table}

\begin{figure}[t]
    \begin{center}
        \centerline{\includegraphics[width=\linewidth]{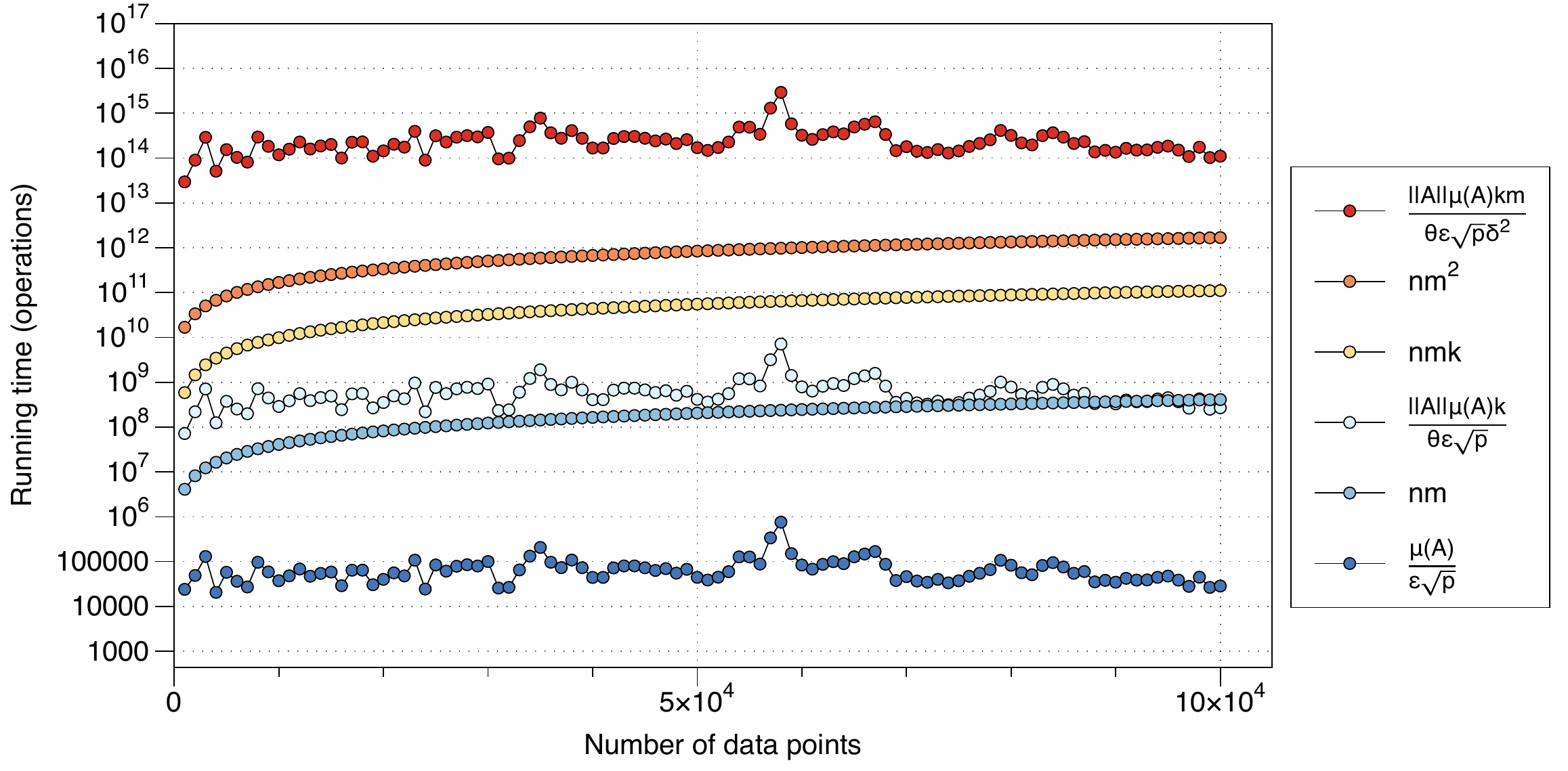}}
        \caption{Run-time comparison on Imagenet as the number of data points increases. The plots have been computed setting $\delta = 0.1$ and $p = 0.85$ and are logarithmic w.r.t. the y axis.}
        \label{Fig:Imagenet_scaling}
    \end{center}
\end{figure}

\section{Conclusions}
In this paper, we formulate many eigenvalue problems in machine learning within a useful framework, filling the gap left open by previous literature with new algorithms.
Our new procedures fill the gap by estimating the quality of a representation and extracting a classical description of the top-k singular values and vectors.
We have shown how to use the new tools to extract the information needed by SVD-based data representation algorithms, computing theoretical error bounds for three machine learning applications.
Besides identifying the proper quantum tools and formalizing the novel quantum algorithms, the main technical difficulty was analyzing how the error propagates to bound the algorithms' run-time properly.

We do not expect run-time improvements that exceed poly-logarithmic factors or constant factors, using similar techniques. 
For non-zero singular values and dense singular vectors, the run-time of the extraction can not be smaller than $kz$, as one needs to read vectors of size $kz$. 
The $\delta^2$ parameter is a tight bound for the $\ell_2$ norm of the vectors, as it is a result of Chernoff's bound.
The parameter $\epsilon$ is a tight error bound from phase estimation, which is necessary to distinguish the singular vectors.
$\theta$ is the condition number of the low-rank approximation of the matrix, and it is necessary to amplify the amplitudes of the smallest singular values. 

As future work, we deem it interesting to explore quantum algorithms for incremental SVD or for datasets whose points are available as a data streaming. 
It might be possible to reduce the overhead due to tomography and achieve greater speed-ups in these settings.
It also remains an open question whether there are particular applications and dataset distributions for which the singular vector extraction algorithms offer a practical advantage over their classical counterparts.
Finally, an appropriate resource estimation that takes into consideration different quantum hardware architectures, noise models, and error correction codes is out of the scope of this paper and is left for future work.

\backmatter

\bmhead{Acknowledgments}
A.B. and S.Z. thank Prof. Ferruccio Resta and Prof. Donatella Sciuto for their support.
A.L. has been supported by QuantERA ERA-NET Cofund in Quantum Technologies implemented within the European Union’s Horizon 2020 Programme (QuantAlgo project), the ANRT, and Singapore's National Research Foundation, the Prime Minister’s Office, Singapore, the Ministry of Education, Singapore under the Research Centres of Excellence program under research grant R 710-000-012-135.

\section*{Declarations}

\subsection*{Competing interests: } The authors have no competing interests to declare.
\subsection*{Funding: } Not applicable.
\subsection*{Ethics approval: } Not applicable.
\subsection*{Consent to participate: } Not applicable.
\subsection*{Consent for publication: } Not applicable.
\subsection*{Availability of data and materials: } Not applicable.
\subsection*{Code availability: } Not applicable.
\subsection*{Authors' contributions: } Not applicable.

\begin{appendices}

\section{Experiments}

\subsection{Factor score ratios distribution in real data}
Throughout the work, we often claim that real datasets for machine learning are low-rank and that the distribution of their singular values is so that a few of them are much bigger than the rest. 
To verify this fact, we have selected five datasets for machine learning and investigated the distribution of the factor score ratios $\frac{\sigma_i^2}{\sum_j^r \sigma_j^2}$ in all of them. 
We briefly describe the datasets and our pre-processing steps.

\paragraph{MNIST}
MNIST \cite{lecun1998gradient} is probably the most used dataset in image classification. 
It is a collection of $70000$ images of $28 \times 28 = 784$ pixels. 
Each image is a black and white hand-written digit between 0 and 9 and it is paired with a label that specifies the digit. 
Since the images are black and white, they are represented as arrays of 784 values that encode the lightness of each pixel.
The dataset, excluding the labels, can be encoded in a matrix of size $70000 \times 784$. 

\paragraph{Fashion MNIST}
Fashion MNIST \cite{xiao2017} is a recent dataset for benchmarking in image classification.
Like the MNIST, it is a collection of 70000 images composed of $28 \times 28 = 784$ pixels. 
Each image represents a black and white fashion item among \{T-shirt/top, Trouser, Pullover, Dress, Coat, Sandal, Shirt, Sneaker, Bag, Ankle boot\}. 
Each image is paired with a label that specifies the item represented in the image. 
Since the images are black and white, they are represented as arrays of 784 values that encode the lightness of each pixel.
The dataset, excluding the labels, can be encoded in a matrix of size $70000 \times 784$. 

\paragraph{CIFAR-10}
CIFAR-10 \cite{krizhevsky2009learning} is another widely used dataset for benchmarking image classification. 
It contains 60000 colored images of $32 \times 32$ pixels, with the values for each of the 3 RGB colors.
Each image represents an object among \{airplane, automobile, bird, cat, deer, dog, frog, horse, ship, truck\} and is paired with the appropriate label. 
We use all the images, reshaping them to unroll the three channels in a single vector. 
The resulting size of the dataset is $60000 \times 3072$.

\paragraph{Tiny Imagenet}
Tiny Imagenet \cite{le2015tiny} is a subset of Imagenet, a large dataset for image classification.
It is a collection of $100000$ colored images of $64 \times 64$ pixels. 
Tiny Imagenet contains images of $200$ object classes. Each class is composed of $500$ images.
We process the dataset to have only black and white images. 
Though the size is considerably less than the one of Imagenet, its complexity is higher than CIFAR-10's.
The dataset, excluding the labels, can be encoded in a matrix of size $100000 \times 4096$. 

\paragraph{Research Paper}
Research Paper \cite{researchpaper} is a dataset for text classification, available on Kaggle. 
It contains 2507 titles of papers together with the labels of the venue where they have been published. 
The labels are \{WWW, INFOCOM, ISCAS, SIGGRAPH, VLDB\}. 
We pre-process the titles to compute a contingency table of $papers \times words$: the value of the $i^{th}-j^{th}$ cell is the number of times that the $j^{th}$ word is contained in the $i^{th}$ title. 
We remove the English stop-words, the words that appear in only one document, and those that appear in more than half the documents. 
The result is a contingency table of size $2507 \times 2010$.\\

Except for Research Paper, all the datasets have been shifted to row mean $0$ and normalized so that $\sigma_{max} = 1$.
Figure \ref{fig:factor_score_ratio} shows the factor score ratios distributions in these datasets.
The rapid decrease is exponential and confirms the expectations. 
\begin{figure}[ht]
    \begin{subfigure}{0.5\linewidth}
        \centering
        \includegraphics[width=0.8\linewidth]{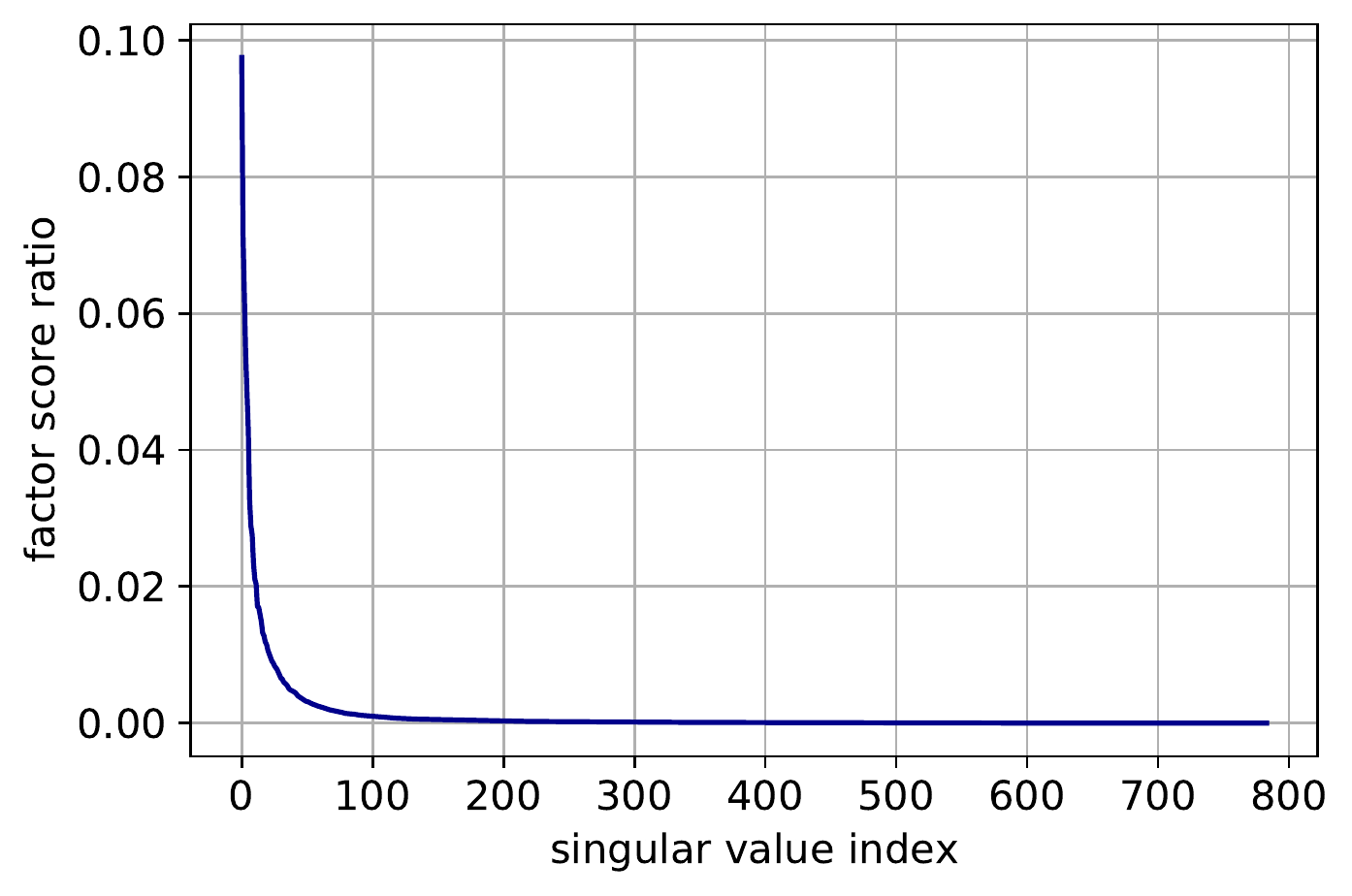}
        \caption{MNIST.}
        \label{fig:factor_score_ratio_cifar10}
    \end{subfigure}
    \begin{subfigure}{0.5\linewidth}
        \centering
        \includegraphics[width=0.8\linewidth]{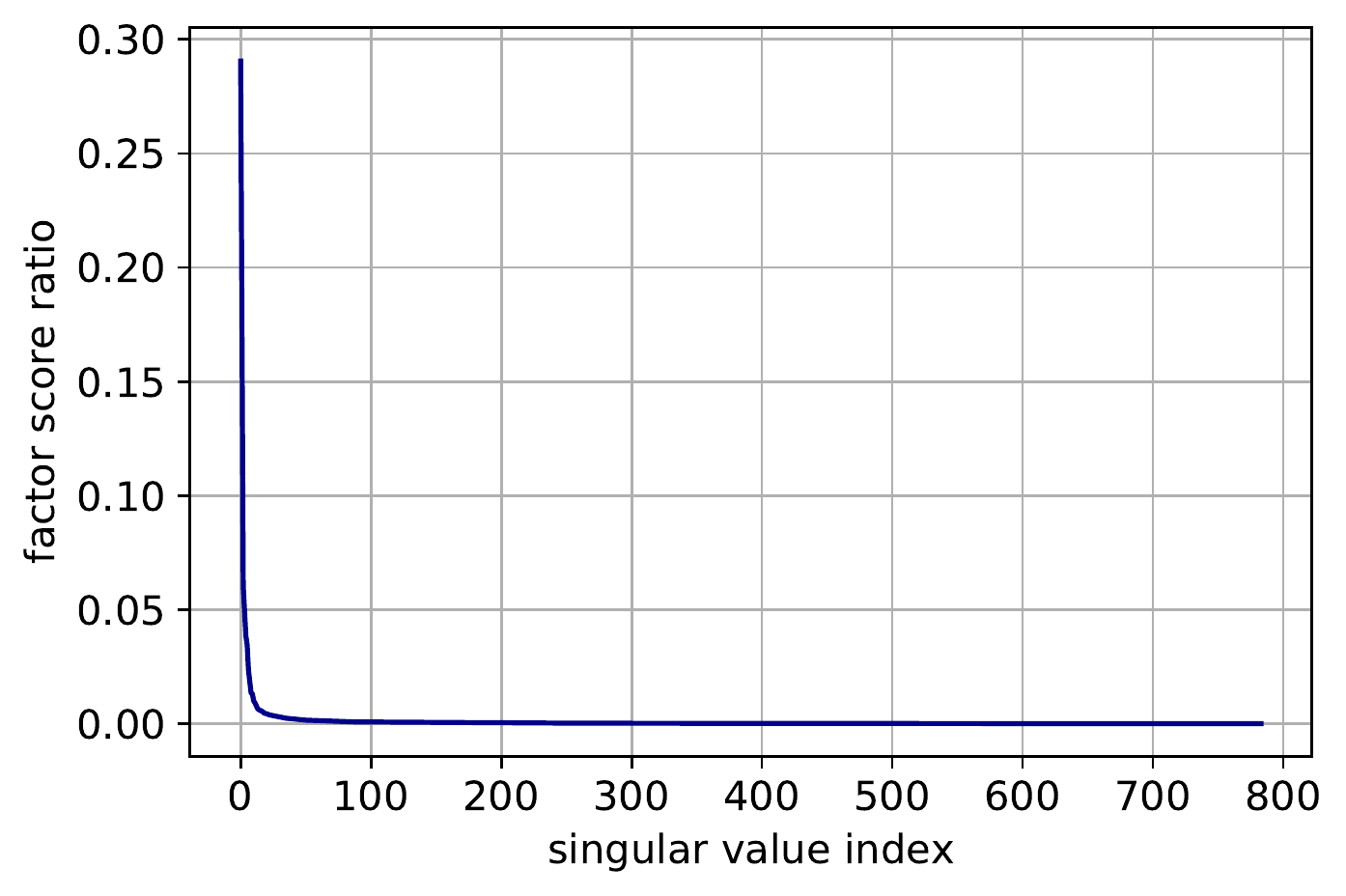}
        \caption{Fashion MNIST.}
        \label{fig:factor_score_ratio_conf}
    \end{subfigure}
    \begin{subfigure}{0.5\linewidth}
        \centering
        \includegraphics[width=0.8\linewidth]{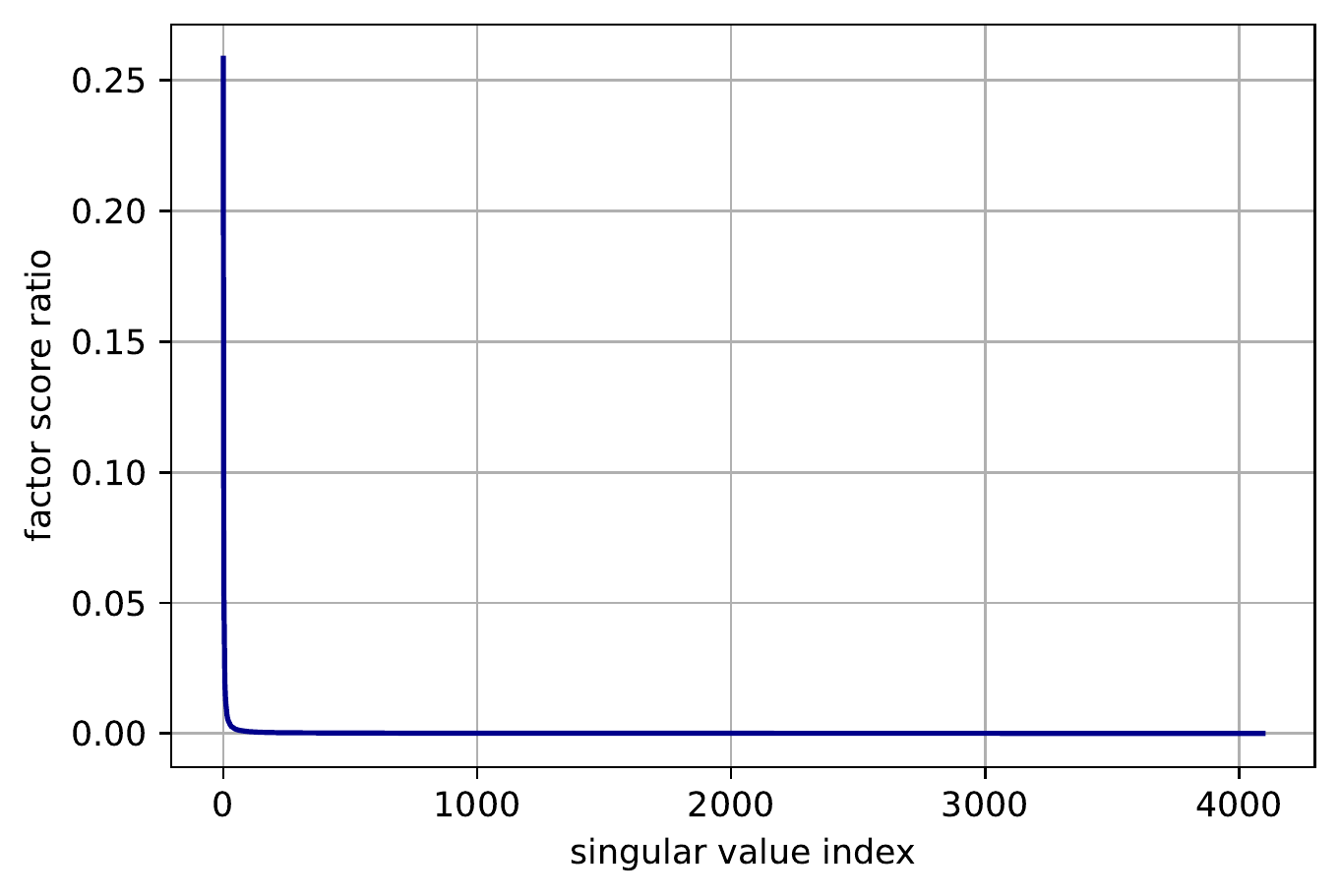}
        \caption{Tiny Imagenet.}
        \label{fig:factor_score_ratio_tiny}
    \end{subfigure}
    \begin{subfigure}{0.5\linewidth}
        \centering
        \includegraphics[width=0.8\linewidth]{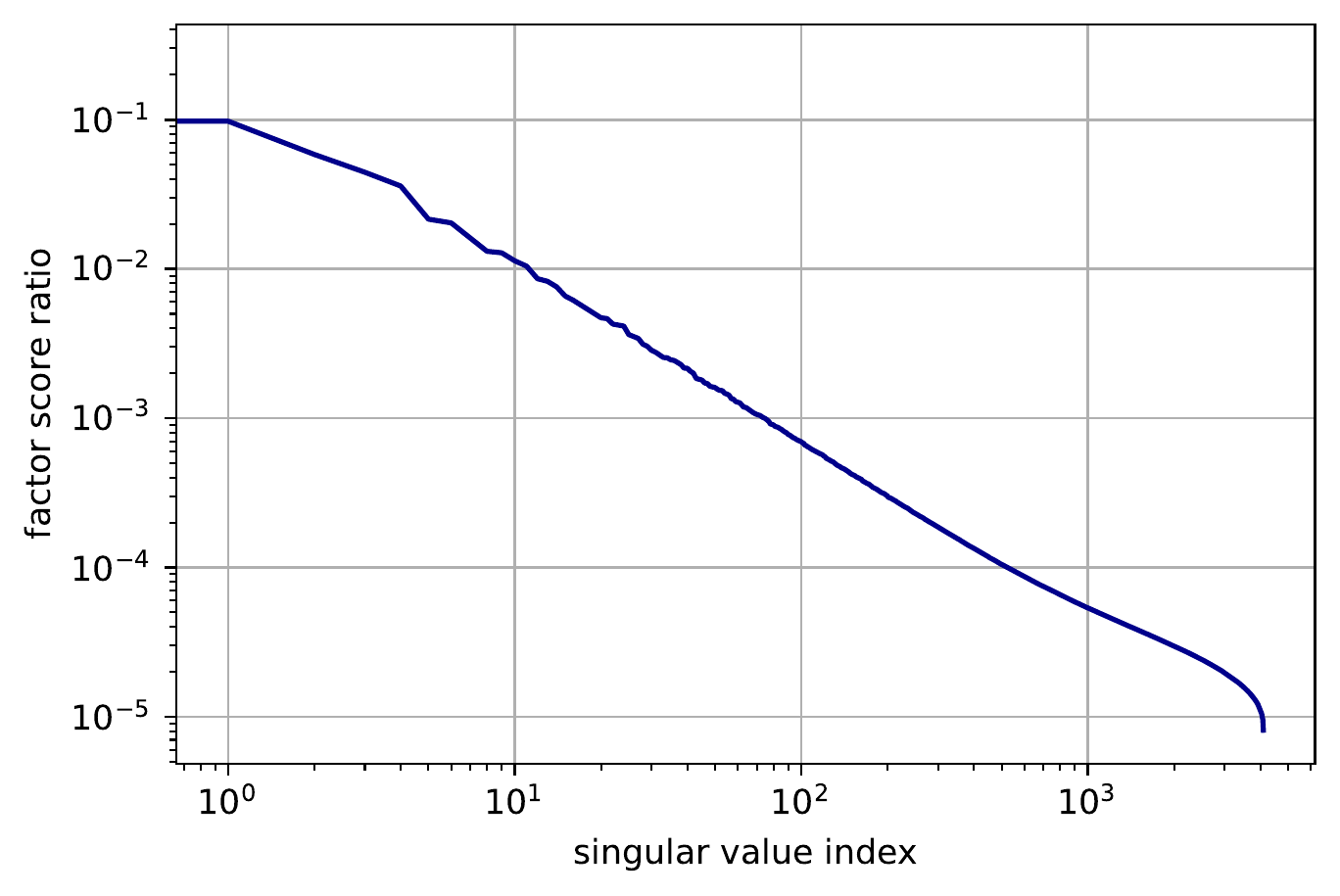}
        \caption{Tiny Imagenet log-log scale.}
        \label{fig:factor_score_ratio_tiny_loglog}
    \end{subfigure}
    \begin{subfigure}{0.5\linewidth}
        \centering
        \includegraphics[width=0.8\linewidth]{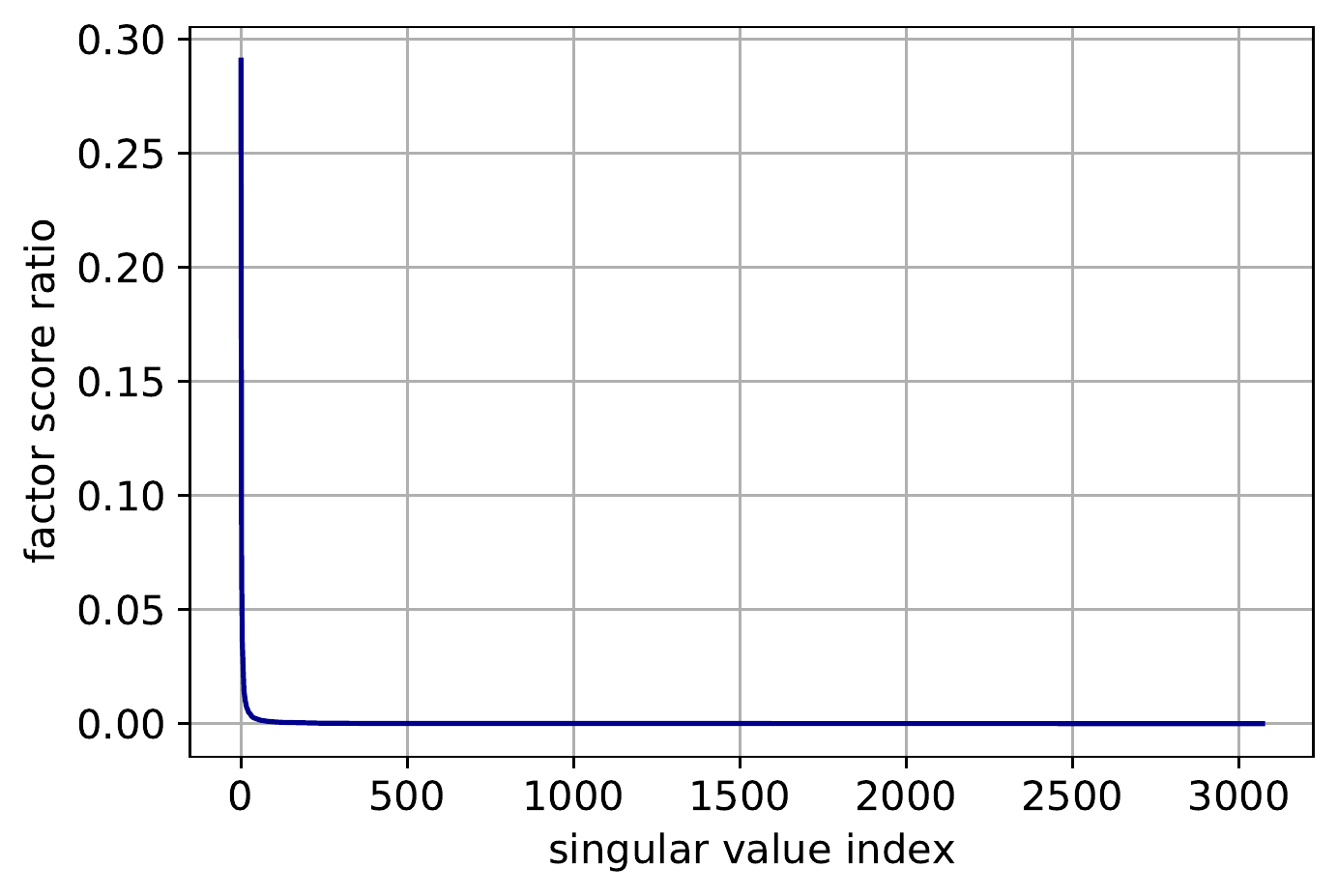}
        \caption{CIFAR-10.}
        \label{fig:factor_score_ratio_MNIST}
    \end{subfigure}
    \begin{subfigure}{0.5\linewidth}
        \centering
        \includegraphics[width=0.8\linewidth]{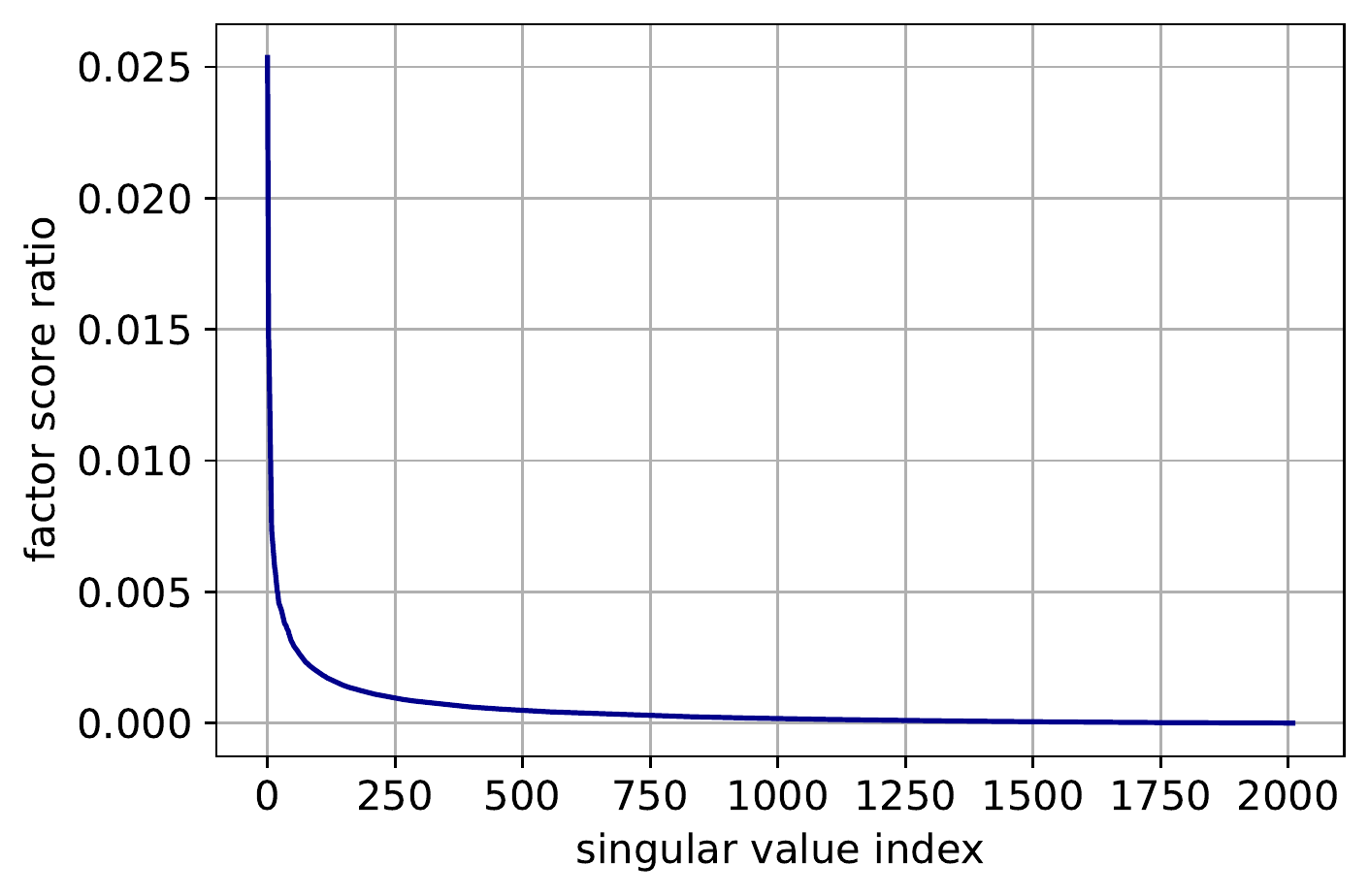}
        \caption{Research Papers.}
        \label{fig:factor_score_ratio_Research}
    \end{subfigure}

    \caption{Factor score ratios distributions in machine learning datasets.}
    \label{fig:factor_score_ratio}
\end{figure}

\subsection{Run-time parameters}
We have computed the run-time parameters on the Tiny Imagenet dataset, maintaining the number of features steady (i.e., 4096 black and white pixels) and observing how the parameters scale as we consider an increasing number of data points.
The results are shown in Figure \ref{fig:runtimeparam}.
In these plots, epsilon is half the gap between the least singular value to retain and the one below, leading to correct thresholding, while theta is computed as the least singular value to retain. 
Although we would fine-tune $\theta$ and $\epsilon$ better in practice, the trend and the order of magnitudes of these parameters would remain like our plots. 
We have computed the best $\mu(\m{A})$ over a finite set of $p \in [0,1]$, and for any number of data points, the Frobenius norm was the most convenient.
Finally, in this experiment, we did not estimate $\delta$.
This is because $\delta$ can only be estimated with respect to a specific classification task.
We did not run classification on this dataset for practical computational reasons. 
However, the following sections contain more run-time parameters for image classification datasets on smaller datasets, including estimates for $\delta$.

\begin{figure}[ht]
    \begin{subfigure}{0.5\linewidth}
        \centering
        \includegraphics[width=0.8\linewidth]{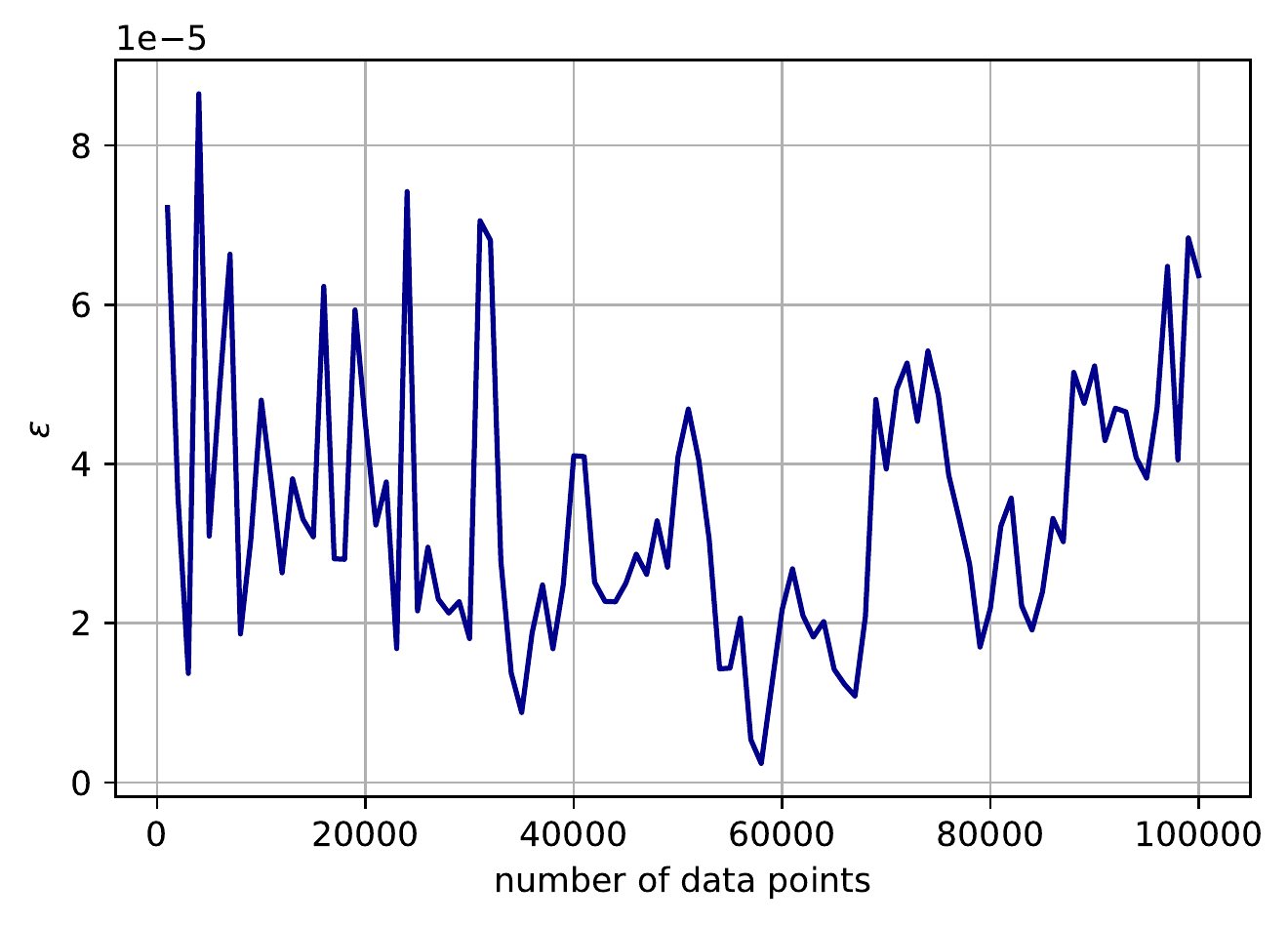}
        \caption{Threshold $\epsilon$.}
        \label{fig:im_eps}
    \end{subfigure}
    \begin{subfigure}{0.5\linewidth}
        \centering
        \includegraphics[width=0.8\linewidth]{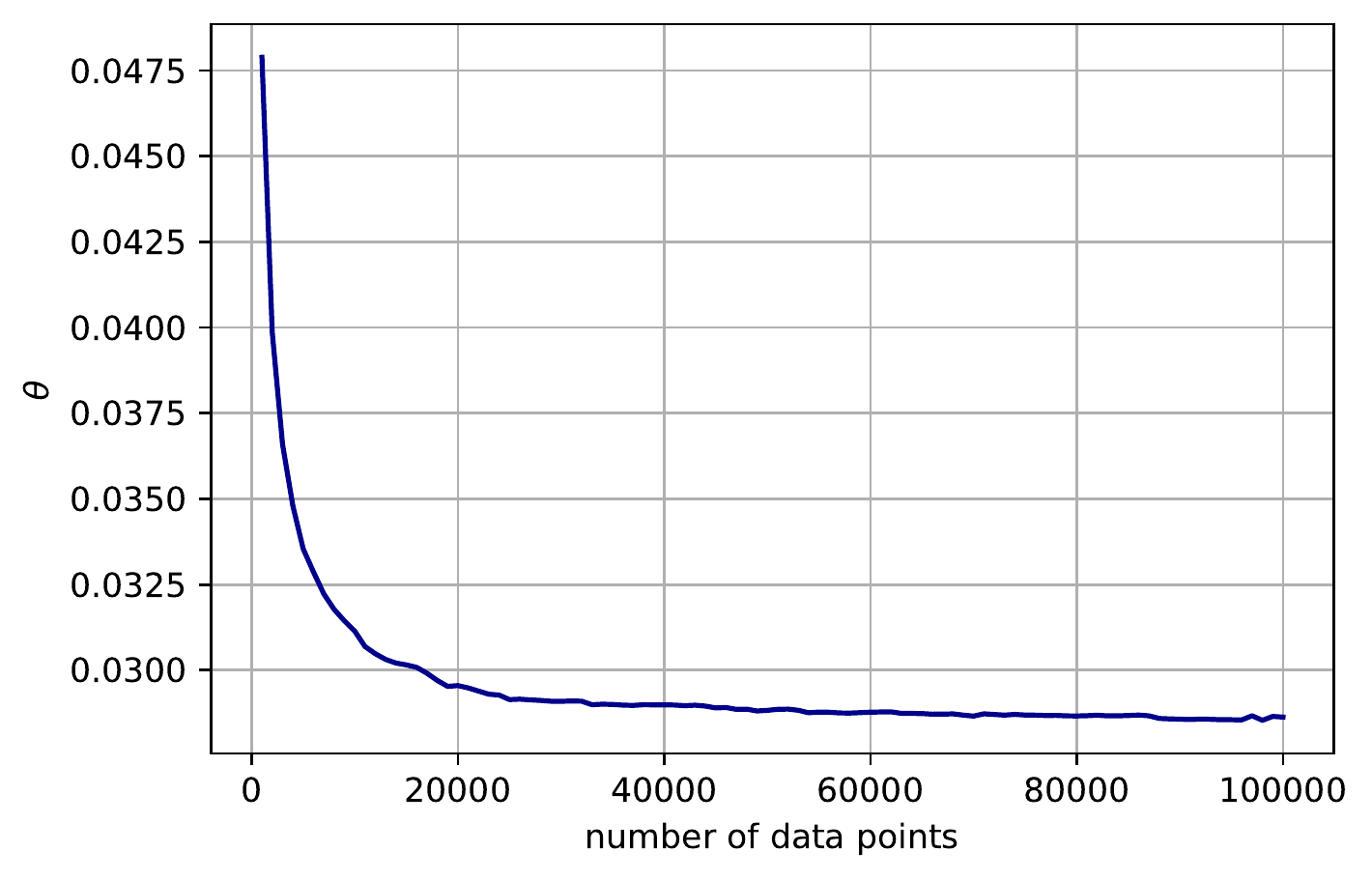}
        \caption{$\theta$.}
        \label{fig:im_theta}
    \end{subfigure}
    \begin{subfigure}{0.5\linewidth}
        \centering
        \includegraphics[width=0.8\linewidth]{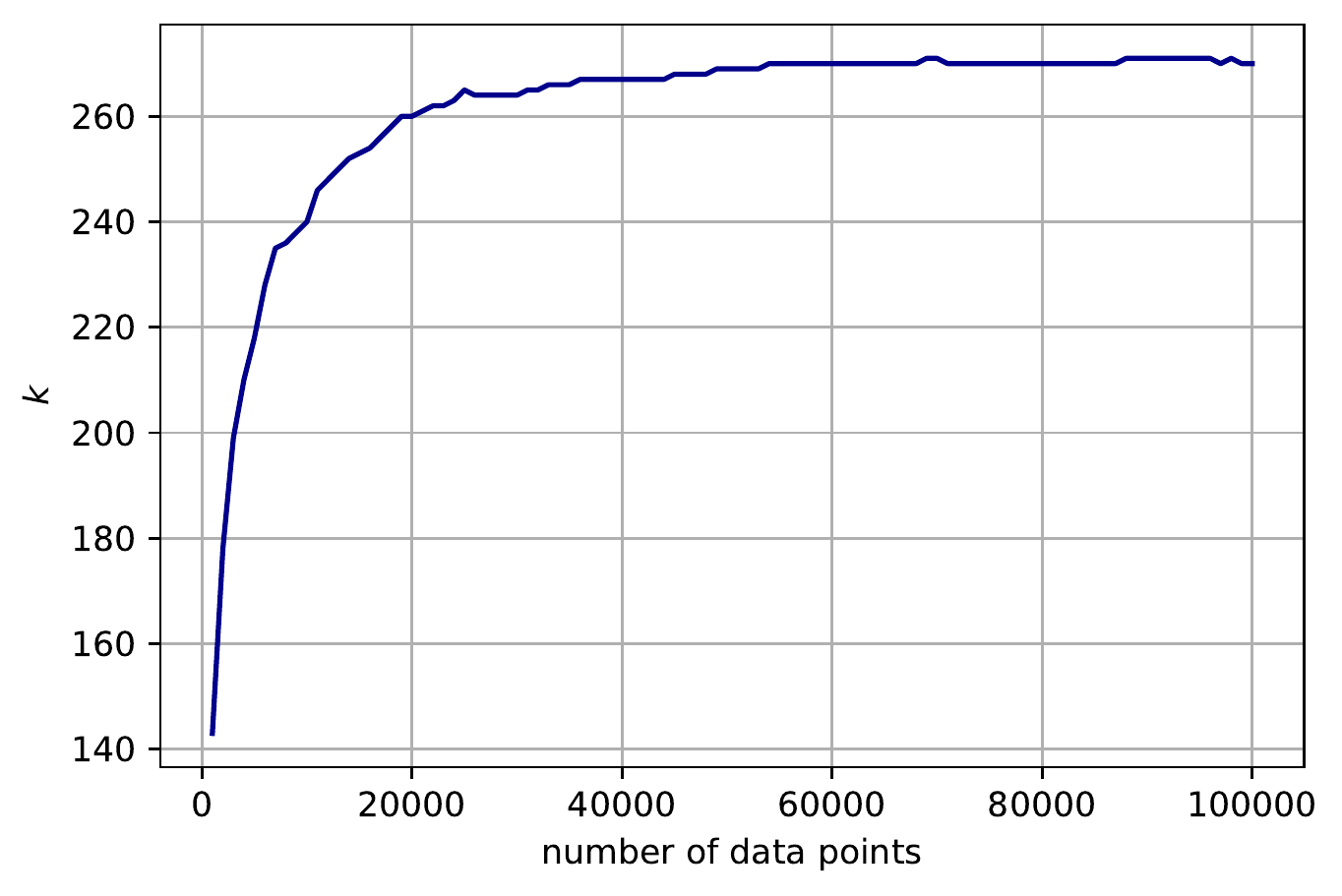}
        \caption{Number of principal components $k$.}
        \label{fig:im_ks}
    \end{subfigure}
    \begin{subfigure}{0.5\linewidth}
        \centering
        \includegraphics[width=0.8\linewidth]{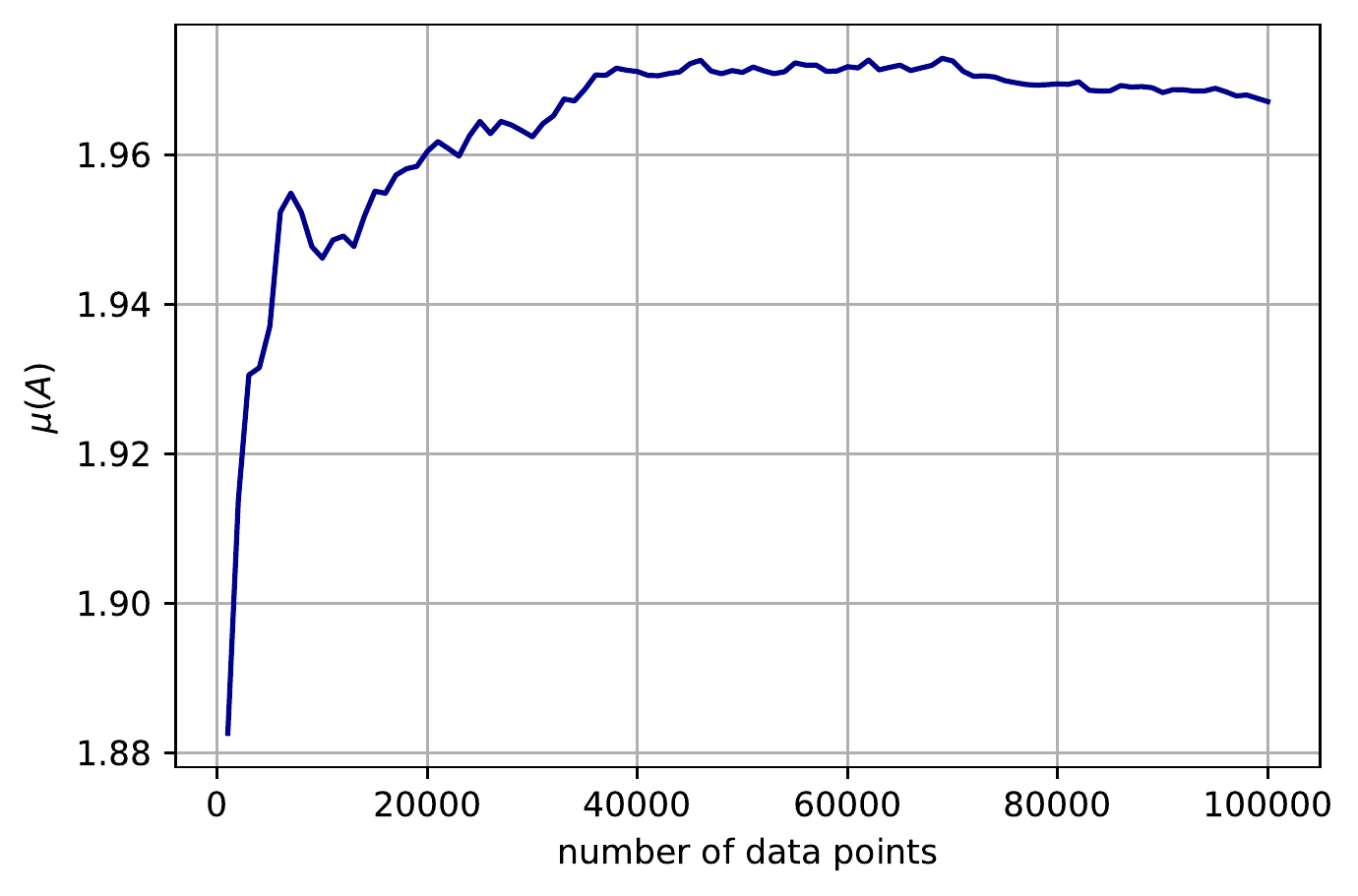}
        \caption{$\mu(\m{A}) = \norm{\m{A}}_F$.}
        \label{fig:im_mu}
    \end{subfigure}
    \begin{subfigure}{0.5\linewidth}
        \centering
        \includegraphics[width=0.8\linewidth]{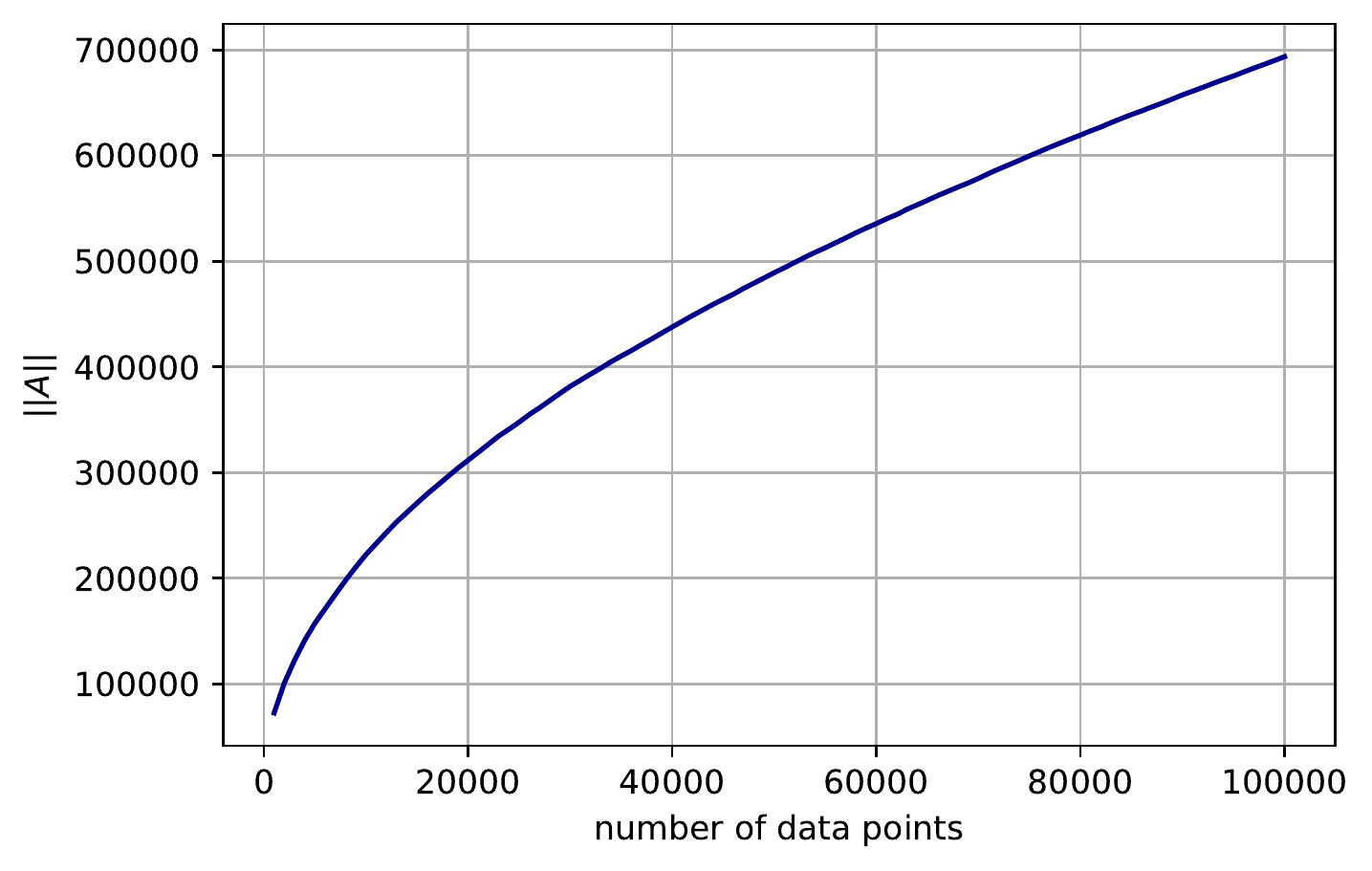}
        \caption{Spectral norm.}
        \label{fig:im_specnorm}
    \end{subfigure}
    \caption{Tiny Imagenet run-time parameters.}
    \label{fig:runtimeparam}
\end{figure}

From the plots, we can see that the spectral norm increases with the number of data points and that the thresholding epsilon is independent of this quantity. 
All the other parameters asymptotically approach a constant after introducing a certain number of data points. 
Our intuition suggests that the number of data points after which the parameters are constant depends on the number of classes in the dataset.
Indeed, this quantity should be related to the amount of information that a new data point adds to the dataset. 
The reader might find it weird that the Frobenius norm, in Figure \ref{fig:im_mu}, slightly decreases towards the end.
However, this trend is justified by the fact that we compute these parameters after the dataset is divided by the spectral norm, and this parameter continues to increase (Figure \ref{fig:im_specnorm}).
The fact that $\mu(A)$ is a positive homogeneous function makes it so that scaling by the spectral norm does not improve the overall run-time.
If we did not divide the dataset by the spectral norm, we would have seen the effect of its trend in $\epsilon$, $\theta$, and $\mu(\m{A})$.
The decrease of $\mu$ after the normalization corresponds to a decrease of $\epsilon$ and $\theta$, making the overall run-time remain the same.

We have used this data to generate the run-time plots in the main text (Figure \ref{Fig:Imagenet_scaling}). 
In that figure, we can see that the algorithms of Section \ref{Subsec:quality_repr} are already convenient on datasets of this size.
In contrast, the ones for singular vector extraction of Section \ref{Subsec:extraction} require datasets of greater size to show their potential.

\subsection{Image classification with quantum PCA}
To provide the reader with a clearer view of our new algorithms and their use in machine learning, we provide experiments on quantum PCA for image classification. 
We perform PCA on the three datasets for image classification (MNIST, Fashion MNIST, and CIFAR 10) and classify them with a K-Nearest Neighbors model. 
First, we simulate the extraction of the singular values and the percentage of variance explained by the principal components (top $k$ factor score ratios' sum) using the procedure from Theorem \ref{TheoMio:factor_score_estimation}. 
Then, we study the error of the model extraction, using Lemma \ref{Lemma:accuracyUSeVS}, by introducing errors on the Frobenius norm of the representation to see how this affects the accuracy. 

\paragraph{Estimating the number of principal components}
We shift MNIST, Fashion MNIST, and CIFAR-10 to row mean $0$ and divide them by their spectral norm.
We directly simulate Theorem \ref{TheoMio:factor_score_estimation} to decide the number of principal components needed to retain 0.85 of the total variance. 
For each dataset, we classically compute the singular values with an exact classical algorithm and simulate the quantum state $\frac{1}{\sqrt{\sum_j^r \sigma_j^2}} \sum_i^r \sigma_i\ket{\sigma_i}$ 
to emulate the measurement process of Algorithm \ref{alg_main:factor_score_estimation}.
After initializing the random object with the correct probabilities, we measure it $\frac{1}{\gamma^2} = 1000$ times and estimate the factor score ratios with a frequentist approach (i.e., dividing the number of measurements of each outcome by the total number of measurements). 
Measuring $1000$ times guarantees us an error of at most $\gamma=0.03$ on each factor score ratio. 
In practice, the error is much smaller.
To determine the number of principal components to retain, we sum the factor score ratios until the percentage of explained variance becomes more significant than $0.85$.
We report the results of these experiments in Table \ref{table:principal_components}.
We obtained good results for all the datasets, estimating no more than three extra principal components than needed.
\begin{table}
    \caption{Results of the estimation of the number of principal components to retain. The parameter $k$ is the number of components needed to retain at least $p=0.85$ of the total variance. The parameter $p$ is computed with respect to the estimated $k$.}
    \label{table:principal_components}
    \begin{center}
    \begin{small}
    \begin{sc}
    \begin{tabular}{ccccr}
        \toprule
        Parameter & MNIST & F-MNIST & CIFAR-10 \\
        \midrule
        Estimated $k$ & 62 & 45& 55\\
        Exact $k$ & 59 & 43& 55\\
        Estimated $p$ & 0.8510 & 0.8510& 0.8510\\
        Exact $p$ & 0.8580 & 0.8543& 0.8514\\
        $\gamma$ & 0.0316 & 0.0316 & 0.0316\\
        \bottomrule
    \end{tabular}
    \end{sc}
    \end{small}
    \end{center}
\end{table}
We could further refine the number of principal components using Theorems \ref{TheoMio:check_explained_variance}, \ref{Theorivisto:binarysearch}. 
When we increase the percentage of variance to retain, the factor score ratios become smaller and the estimation worsens.
When the factor score ratios become too small to perform efficient sampling, it is possible to establish the threshold $\theta$ for the smaller singular value to retain using Theorems \ref{TheoMio:check_explained_variance} and \ref{Theorivisto:binarysearch}.
Suppose one is interested in refining the exact number $k$ of principal components, rather than $\theta$.
In that case, it is possible to obtain it using a combination of the Theorems \ref{TheoMio:check_explained_variance}, \ref{Theorivisto:binarysearch} and the quantum counting algorithm in time that scales with the square root of $k$ (Theorem \ref{TheoMio:counting}) to find a good trade-off.
Once one sets the number of principal components, the next step is to use Theorem \ref{TheoMio:top-k_sv_extraction} to extract the top singular vectors. 
To do so, we can retrieve the threshold $\theta$ from the previous step by checking the gap between the last singular value to retain and the first to exclude. 

\paragraph{Studying the error in the data representation}
We continue the experiment by checking how much error in the data representation a classifier can tolerate.
We compute the exact PCA representation for the three datasets and the 10-fold Cross-validation error using k-Nearest Neighbors with $7$ neighbors.
For each dataset, we introduce errors in the representation and check how the accuracy decreases.
To simulate the error, we perturb the exact representation by adding truncated Gaussian error (zero mean and unit variance, truncated on the interval $[\frac{-\xi}{\sqrt{nm}},\frac{\xi}{\sqrt{nm}}]$) to each matrix entry. 
The graph in Figure \ref{fig:error_matrix} shows the distribution of the simulated error on $2000$ approximation of a matrix $\m{A}$, such that $\norm{\m{A} - \overline{\m{A}}} \leq 0.1$. 
The distribution is still Gaussian, centered almost at half the bound.
\begin{figure}
    \center
    \begin{subfigure}{0.5\linewidth}
        \includegraphics[width=0.8\linewidth]{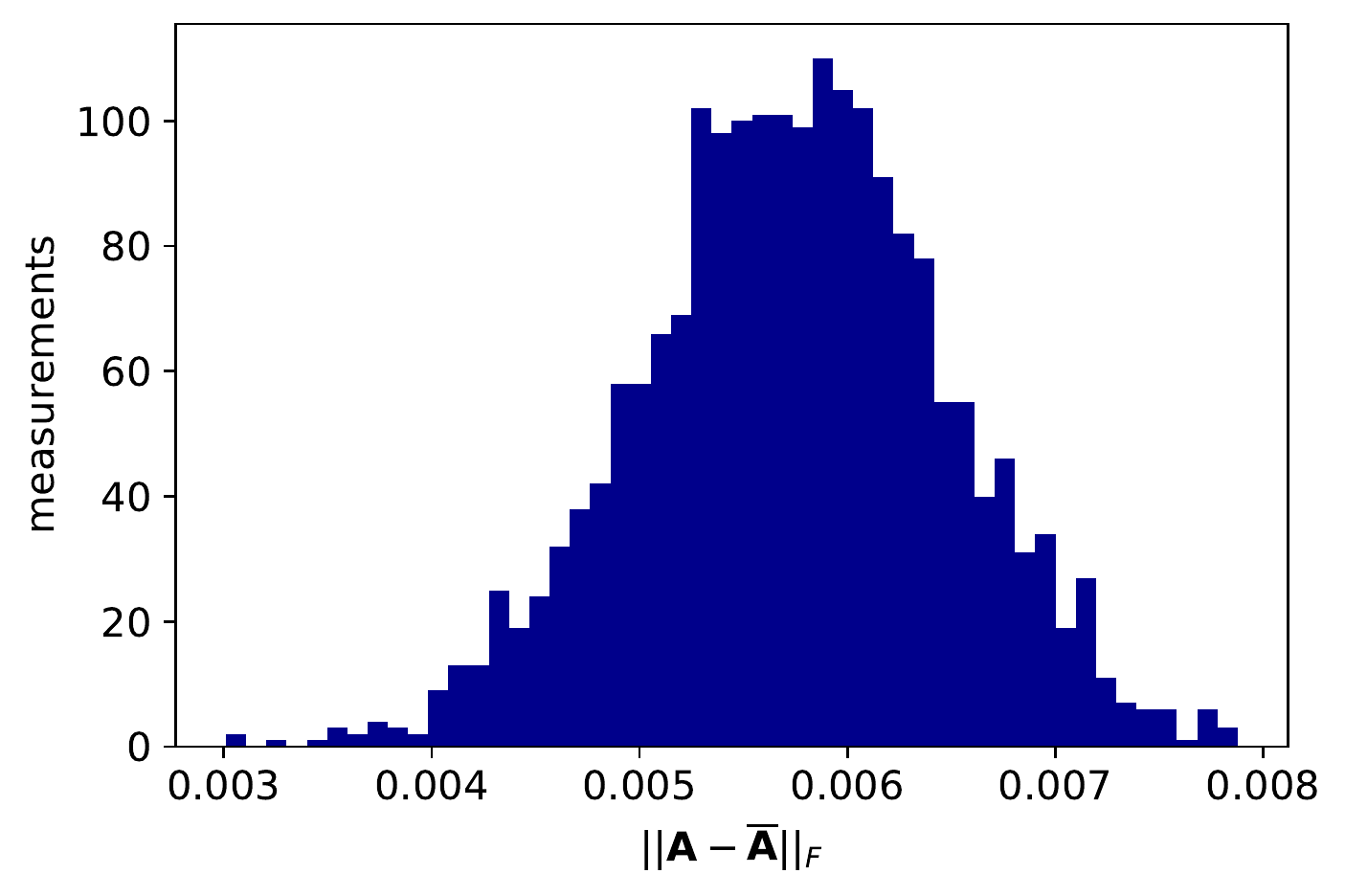}
    \end{subfigure}

    \caption{Introducing some error in the Frobenius norm of a matrix $\m{A}$. The error was introduced such that $\norm{\m{A} - \overline{\m{A}}} \leq 0.01$. The figure shows the distribution of the error over $2000$ measurements.}
    \label{fig:error_matrix}
\end{figure}
The results show a reasonable tolerance of the errors; we report them in two sets of figures.
Figure \ref{fig:class_bound} shows the drop of accuracy in classification as the error bound increases. Figure \ref{fig:class_relative} shows the accuracy trend against the approximation's error. 
\begin{figure}[h]
    \begin{subfigure}{0.328\linewidth}
        \centering
        \includegraphics[width=\linewidth]{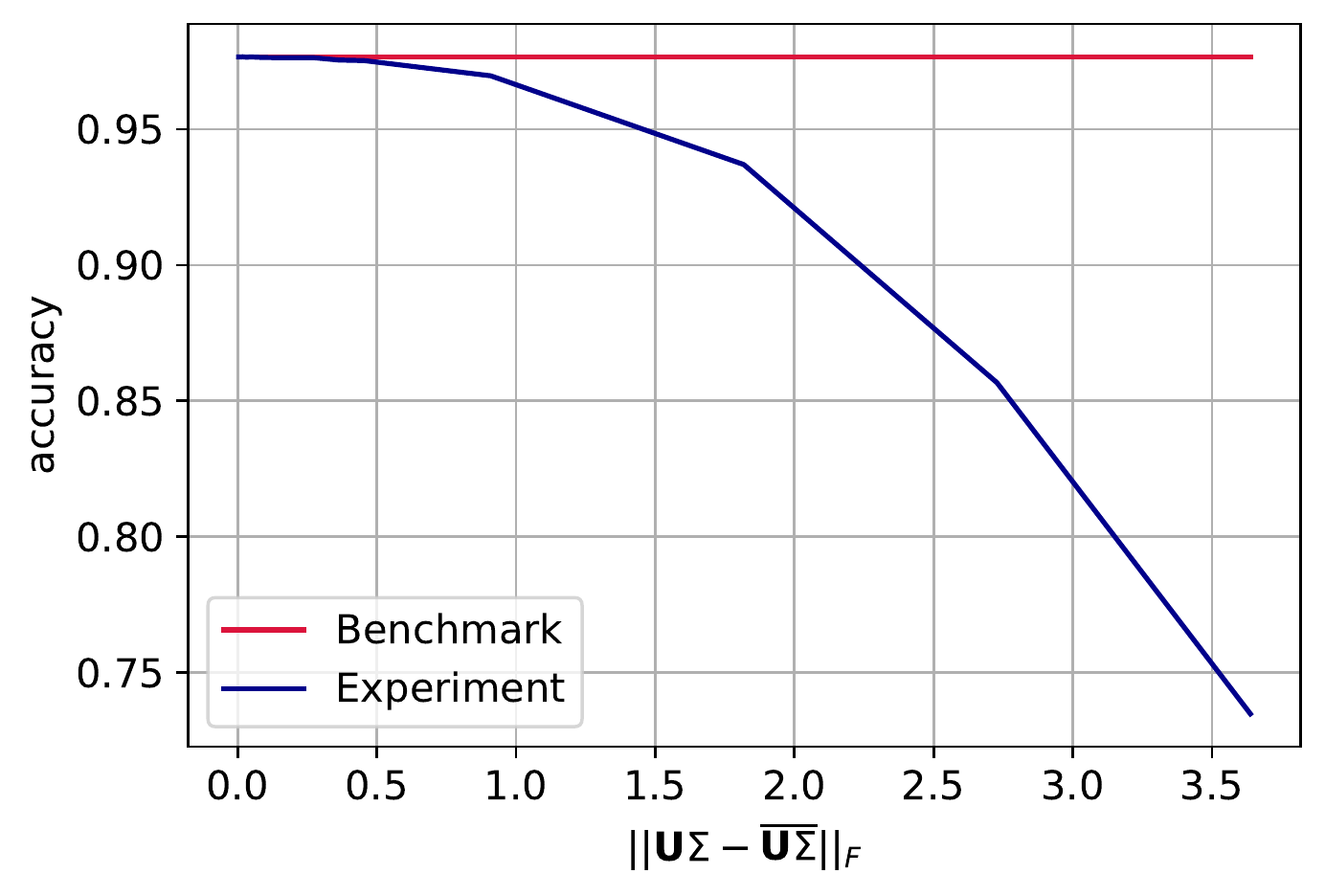}
        \caption{MNIST.}
        \label{fig:class_bound_mnist}
    \end{subfigure}
    \begin{subfigure}{0.328\linewidth}
        \centering
        \includegraphics[width=\linewidth]{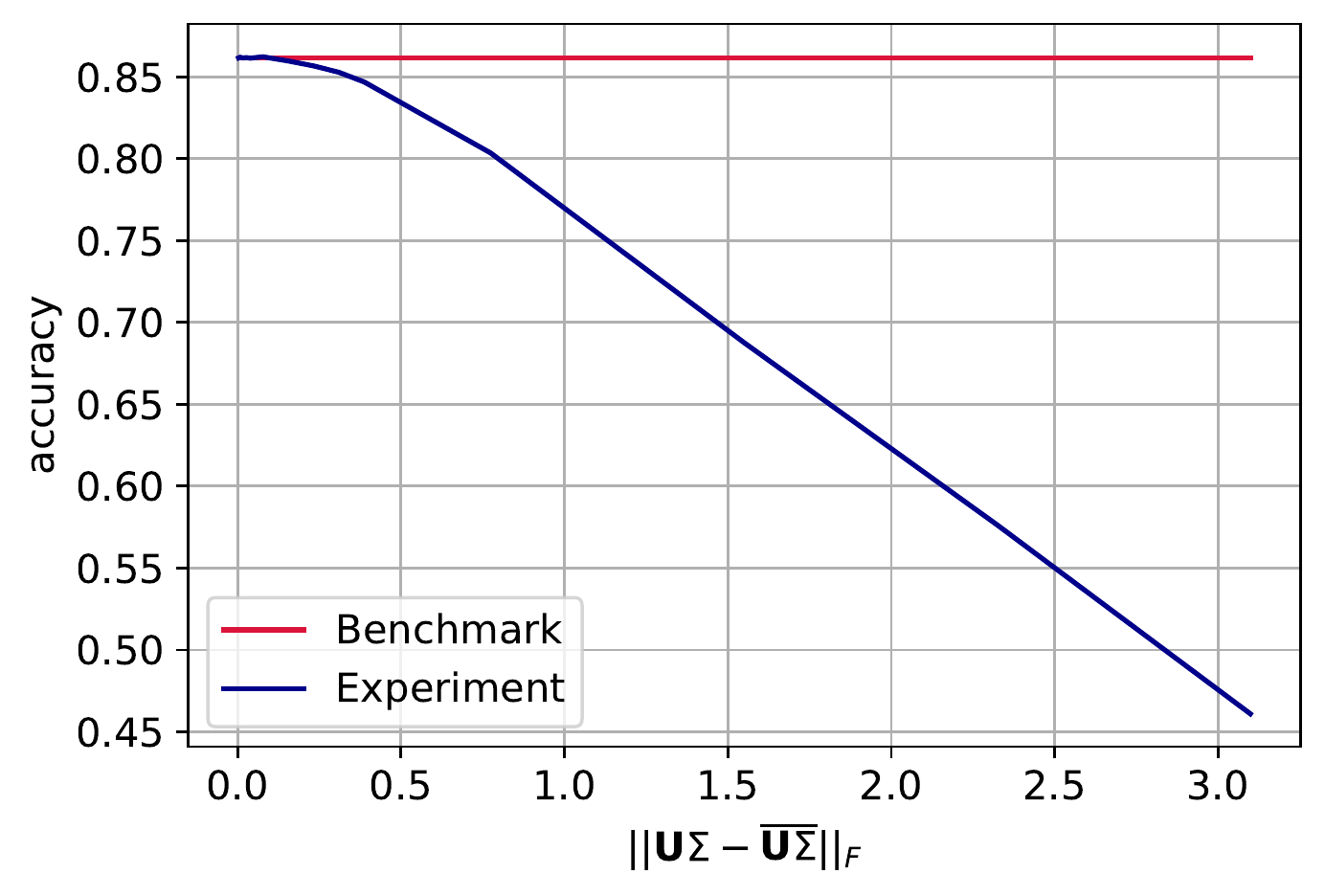}
        \caption{Fashion MNIST.}
        \label{fig:class_bound_Fmnist}
    \end{subfigure}
    \begin{subfigure}{0.328\linewidth}
        \includegraphics[width=\linewidth]{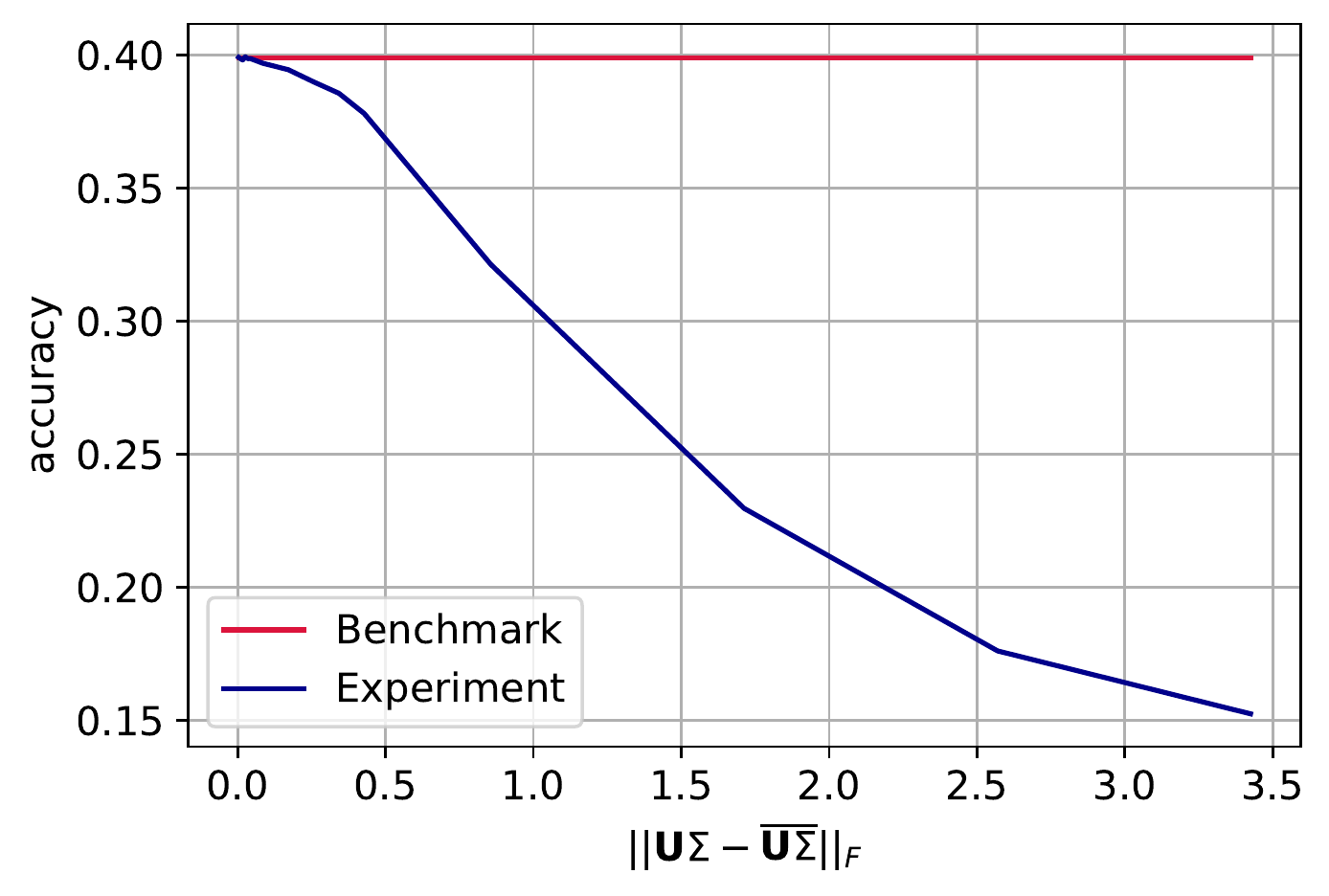}
        \caption{CIFAR-10.}
        \label{fig:class_bound_cifar}
    \end{subfigure}

    \caption{Classification accuracy of $7$-Nearest Neighbor on three machine learning datasets after PCA's dimensionality reduction. The drop in accuracy is plotted with respect to the \emph{bound} on the Frobenius norm of the difference between the exact data representation and its approximation.}
    \label{fig:class_bound}
\end{figure}
\begin{figure}[h]
    \begin{subfigure}{0.328\linewidth}
        \centering
        \includegraphics[width=\linewidth]{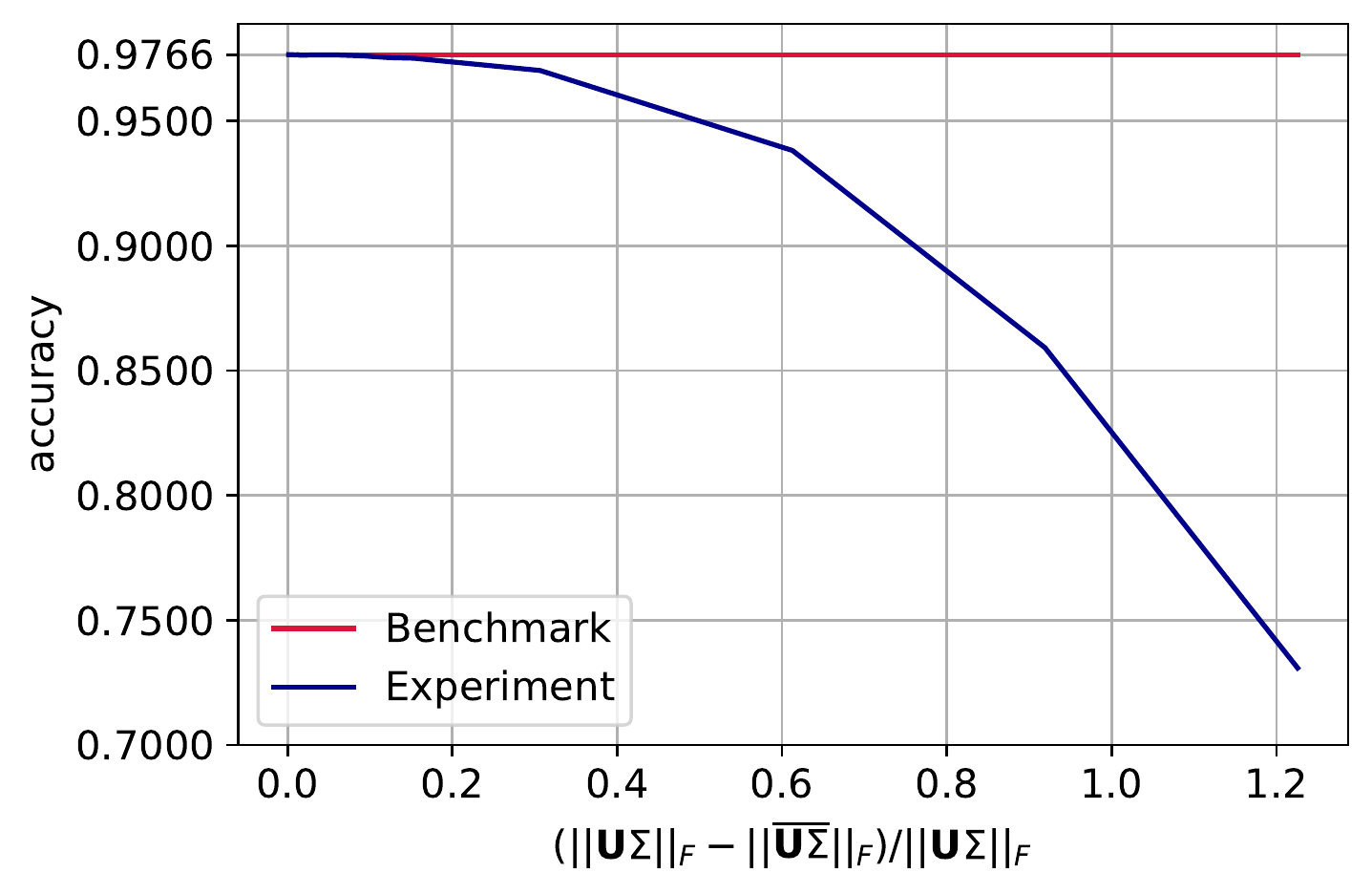}
        \caption{MNIST.}
        \label{fig:class_relative_mnist}
    \end{subfigure}
    \begin{subfigure}{0.328\linewidth}
        \centering
        \includegraphics[width=\linewidth]{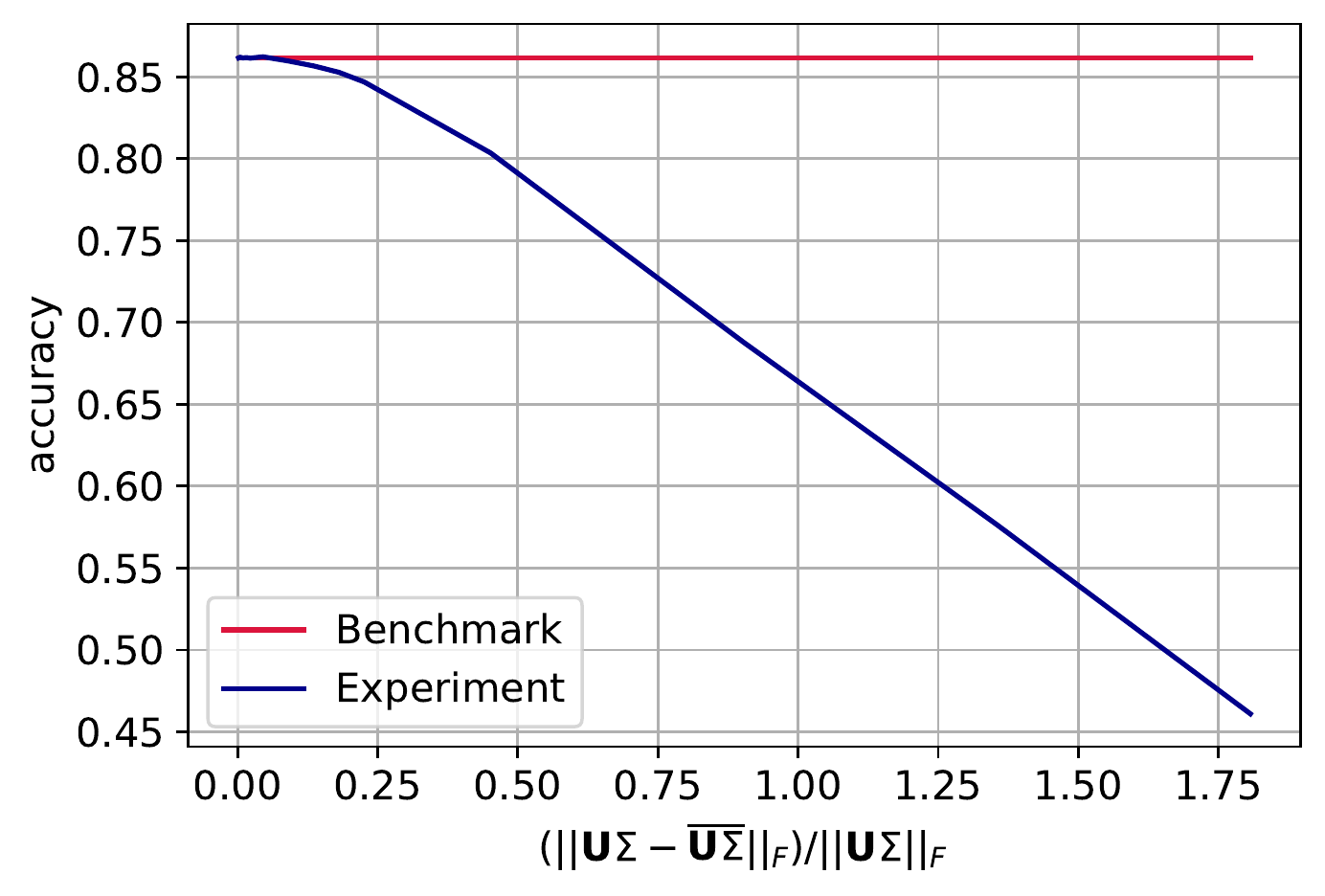}
        \caption{Fashion MNIST.}
        \label{fig:class_relative_Fmnist}
    \end{subfigure}
    \begin{subfigure}{0.328\linewidth}
        \includegraphics[width=\linewidth]{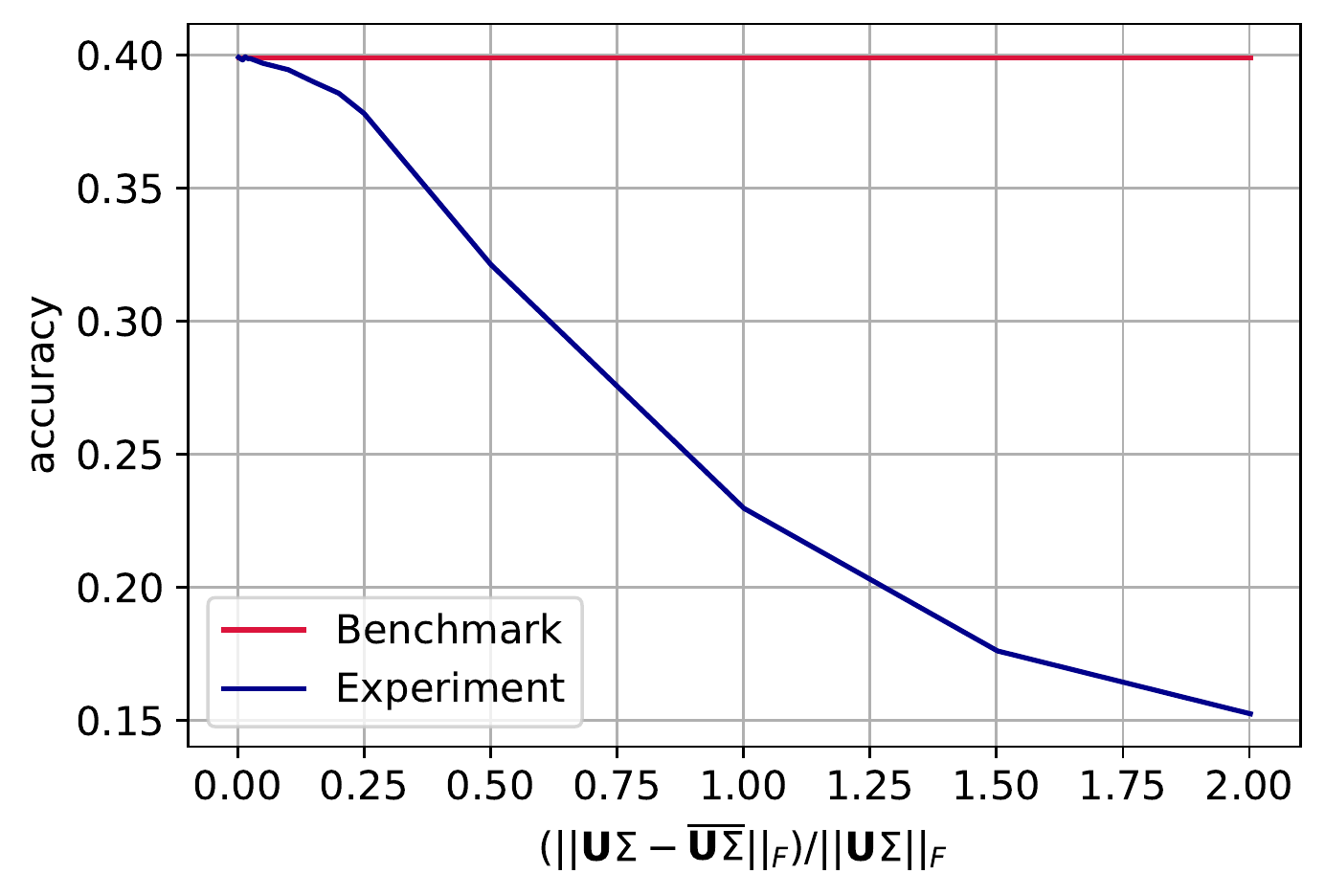}
        \caption{CIFAR-10.}
        \label{fig:class_relative_cifar}
    \end{subfigure}

    \caption{Classification accuracy of $7$-Nearest Neighbor on three machine learning datasets after PCA's dimensionality reduction. The drop in accuracy is plotted with respect to the \emph{effective} Frobenius norm of the difference between the exact data representation and its approximation.}
    \label{fig:class_relative}
\end{figure}

\paragraph{Analyzing the run-time parameters}
\label{SuppExperiments:run-time_param}
As discussed in Section \ref{Section:applications}, the model extraction's run-time is 
$\widetilde{O}\left(\left( \frac{1}{\gamma^2} + \frac{kz}{\theta\sqrt{p}\delta^2}\right)\frac{\mu(\m{A})}{\epsilon}\right)$, where $\m{A} \in \R^{n\times m}$ is PCA's input matrix, $\mu(\m{A})$ is a parameter bounded by $\min(\norm{\m{A}}_F, \norm{\m{A}}_\infty)$, $k$ is the number of principal components retained, $\theta$ is the value of the last singular value retained, $\gamma$ is the precision to estimate the factor score ratios, $\epsilon$ bounds the absolute error on the estimation of the singular values, $\delta$ bounds the $\ell_2$ norm of the distance between the singular vectors and their approximation, and $z$ is either $n$, $m$ depending on whether we extract the left singular vectors, to compute the classical representation, or the right ones, to retrieve the model and allow for further quantum/classical computation.
This run-time can be further lowered using Theorem \ref{Theorivisto:binarysearch} if we are not interested in the factor score ratios.
This paragraph aims to show how to determine the run-time parameters for a specific dataset.
We enrich the parameters of Table \ref{table:principal_components} with the ones in Table \ref{table:supp_parameters}, and we discuss how to compute them. 
From the previous paragraphs, it should be clear how to determine $k$, $\theta$, $\gamma$, and $p$, and it is worth noticing again that $1/\sqrt{p} \simeq 1$.
We have computed $\mu(\m{A})$ over a finite set of values $p \in [0,1]$ and have seen that $\norm{\m{A}}_F$ is the best $\mu(\m{A})$ (this is true for CIFAR-10, Fashion MNIST, Tiny Imagenet, and Research Papers as well).
To compute the parameter $\epsilon$ one should consider the epsilon that allows for a correct singular value thresholding. 
We refer to this as the thresholding $\epsilon$ and
set it as the difference between the last retained singular value and the first that is excluded.
For the sake of completeness, we have run experiments to check how the Coupon Collector's problem changes as $\epsilon$ increases. 
Recall that in the proof of Theorem \ref{TheoMio:top-k_sv_extraction}, we use 
$\frac{1}{\sqrt{\sum_{i}^k \frac{\sigma_i^2}{\overline{\sigma}_i^2}}}\sum_i^k\frac{\sigma_i}{\overline{\sigma}_i} \ket{\ve{u}_i}\ket{\ve{v}_i}\ket{\overline{\sigma}_i} \sim \frac{1}{\sqrt{k}}\sum_i^k \ket{\ve{u}_i}\ket{\ve{v}_i}\ket{\overline{\sigma}_i}$ to say that the number of measurements needed to observe all the singular values is $O(k\log(k))$, and this is true only if $\epsilon$ is small enough to let the singular values distribute uniformly. 
We observe that the thresholding $\epsilon$ always satisfies the Coupon Collector's argument, and we have plotted the results of our tests in Figure \ref{fig:coupon_collector}.
\begin{figure}[h]
    \begin{subfigure}{0.328\linewidth}
        \centering
        \includegraphics[width=\linewidth]{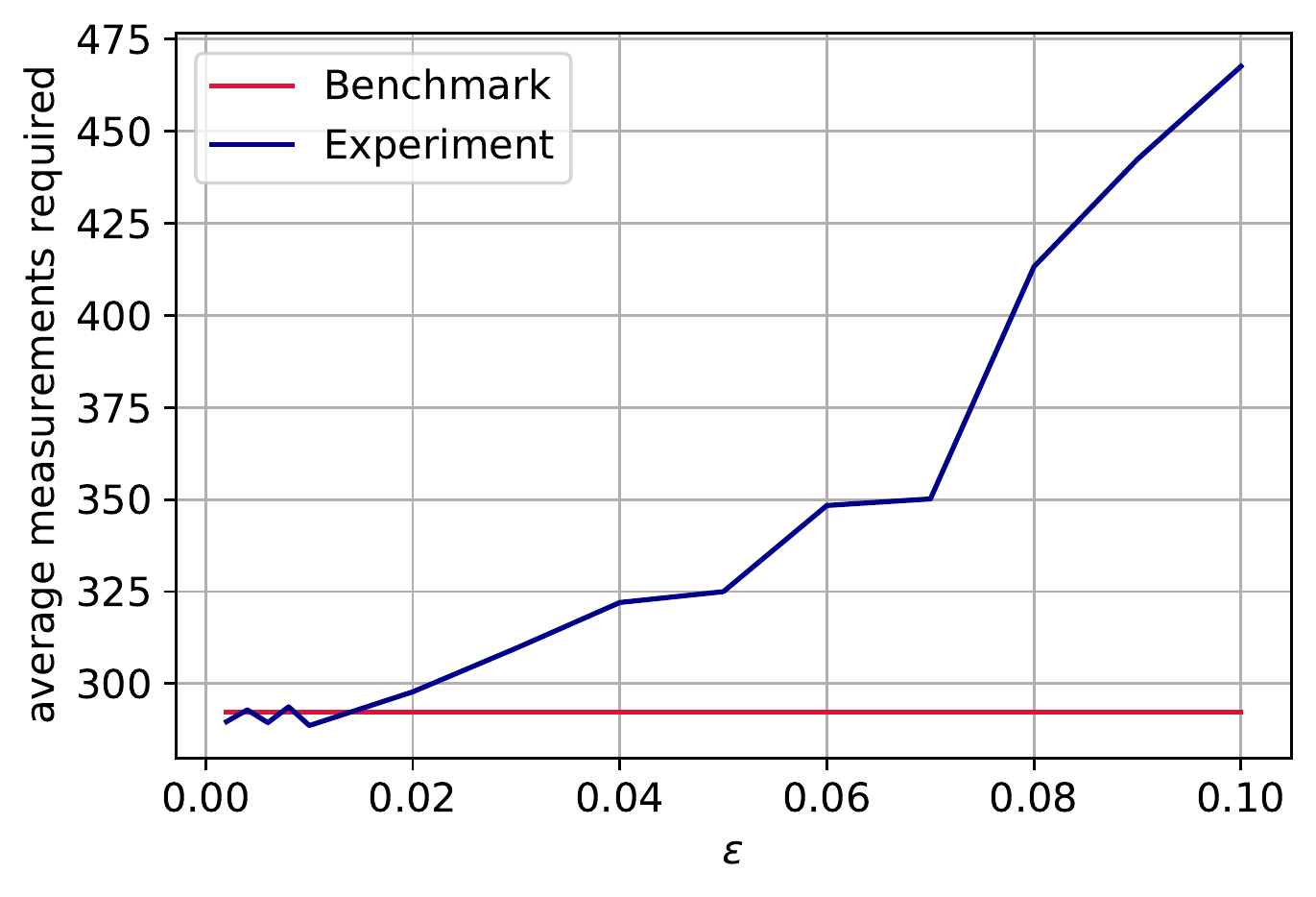}
        \caption{MNIST.}
        \label{fig:coupon_mnist}
    \end{subfigure}
    \begin{subfigure}{0.328\linewidth}
        \centering
        \includegraphics[width=\linewidth]{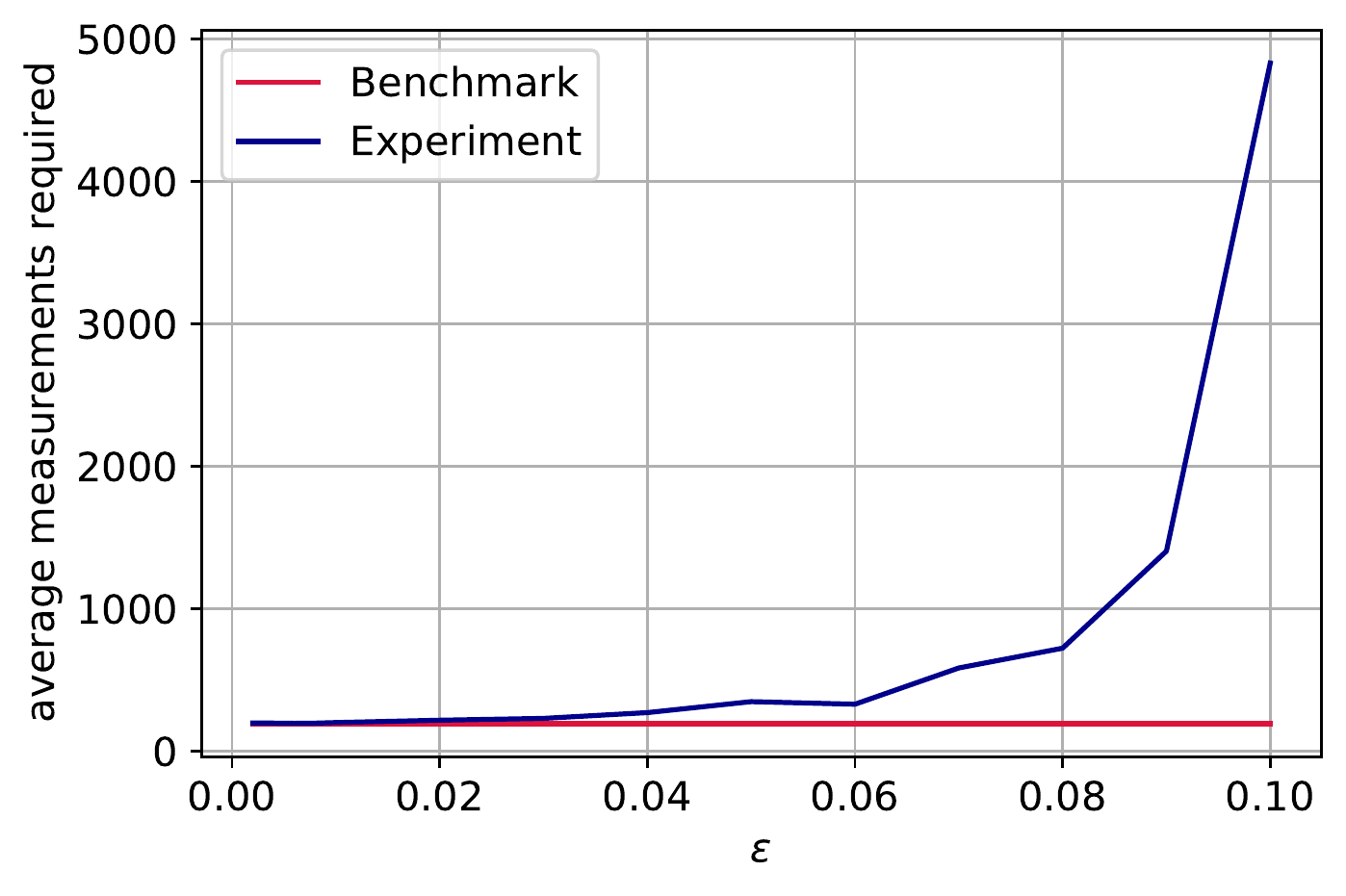}
        \caption{Fashion MNIST.}
        \label{fig:coupon_Fmnist}
    \end{subfigure}
    \begin{subfigure}{0.328\linewidth}
        \includegraphics[width=\linewidth]{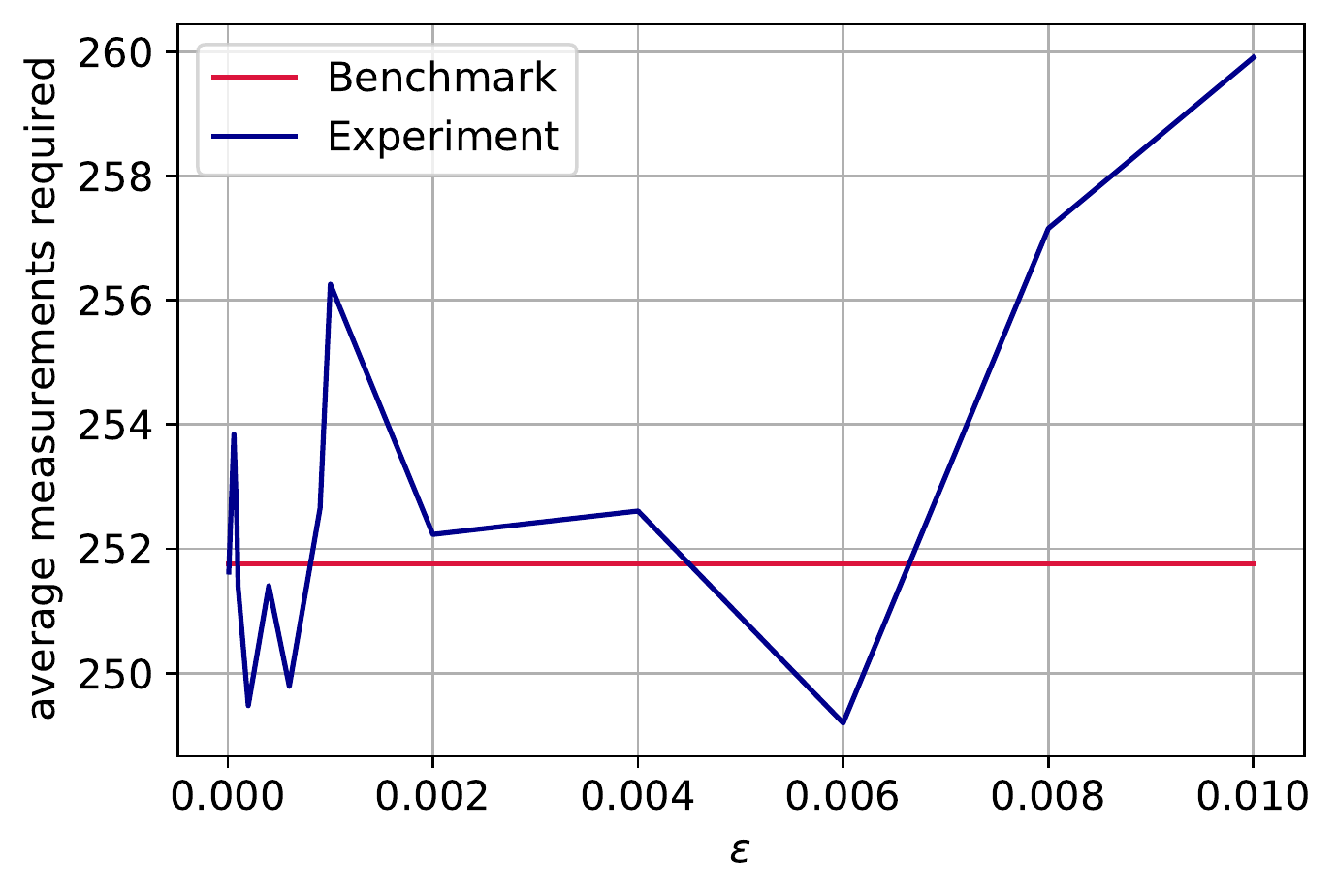}
        \caption{CIFAR-10.}
        \label{fig:coupon_cifar}
    \end{subfigure}

    \caption{Number of measurements needed to obtain all the 
    $k$ singular values from the quantum state $\frac{1}{\sqrt{\sum_i^r\sigma_i^2}}\sum_i^k \frac{\sigma_i}{\overline{\sigma}_i}\ket{\overline{\sigma}_i}$, where $\norm{\sigma_i - \overline{\sigma}_i} \leq \epsilon$, as $\epsilon$ increases. The benchmark line is $k\log_{2.4}(k)$.}
    \label{fig:coupon_collector}
\end{figure}
Furthermore, we have computed $\delta$ by using the fact that $\norm{\m{A} - \overline{\m{A}}} \leq \sqrt{k}(\epsilon + \delta)$ (Lemma \ref{Lemma:accuracyUSeVS}). 
We have computed an estimate for $\delta$ by inverting the equation and considering the thresholding $\epsilon$. 
In particular, we have fixed $\norm{\m{A} - \overline{\m{A}}}$ to the biggest value in our experiments so that the accuracy doesn't drop more than $1\%$. 
\begin{table}
\caption{Run-time parameters.}
\label{table:supp_parameters}
    \begin{center}
        \begin{small}
            \begin{sc}
                \begin{tabular}{ccccr}
                    \toprule
                    Parameter & MNIST & F-MNIST & CIFAR-10 \\
                    \midrule
                    $\mu(\m{A})$ & 3.2032 & 1.8551 & 1.8540\\
                    Thrs. $\epsilon$ & 0.0030 & 0.0009& 0.0006\\
                    $\theta$ & 0.1564 & 0.0776& 0.0746\\
                    $\delta$ & 0.1124 & 0.0106& 0.0340\\
                    \bottomrule
                \end{tabular}
            \end{sc}
        \end{small}
    \end{center}
\end{table}
These results show that Theorem \ref{TheoMio:factor_score_estimation}, \ref{TheoMio:check_explained_variance}, and \ref{Theorivisto:binarysearch} can already provide speed-ups on datasets as small as the MNIST.
Even though their speed-up is not exponential, they still run sub-linearly on the number matrix entries even though all the entries are taken into account during the computation, offering a polynomial speed-up with respect to their traditional classical counterparts. 
On the other hand, Theorem \ref{TheoMio:top-k_sv_extraction} requires bigger datasets. 
These algorithms are expected to show their full speed-up on big low-rank datasets that maintain a good distribution of singular values.
As a final remark, the parameters have similar orders of magnitude.

\section{Related works} \label{Section:related_work}
One of the first papers that faced the problem of performing the eigendecomposition of a matrix with a quantum computer is the well-known \citet{lloyd2014qpca}, which leveraged the intuition that density matrices are covariance matrices whose trace has been normalized. In this work, the authors assume to have quantum access to a matrix in the form of a density matrix and develop a method for fast density matrix exponentiation that enables preparing the eigendecomposition of the input matrix in time logarithmic on its dimensions. 
However, this algorithm requires the input matrix to be square, symmetric, and sparse or low-rank. 
More recently, the works of Kerenidis et al. on recommendation systems  \cite{kerenidis2016recommendation} and least-squares \cite{kerenidis2020gradient} have used a different definition of quantum access to a matrix (the one used throughout this work) and defined the task of singular value estimation. Their singular value decomposition scales better with respect to the error parameters, eliminates the dependency on the condition number, and does not have requirements on the input matrix. 
Several recent works, such as \citet{lin2019improved, rebentrost2018svd, gu2019quantum}, have improved or extended the quantum singular value decomposition techniques. 
Almost none of them have provided a formal analysis of an algorithm that ensures classical access to the singular vectors, values and the amount of variance explained by each.
There have also been attempts at creating near-term quantum algorithms for singular value decomposition. 
These works propose quantum circuits for singular value decomposition of quantum states on noisy intermediate-scale quantum (NISQ) devices using variational circuits \cite{bravo2020quantum, wang2020variational}. 
However, the complexity of such methods is unclear, and recent works have questioned the efficacy of the speed-ups of variational quantum algorithms due to (entanglement and noise-induced) barren plateaus in the optimization landscape \cite{wang2020noise, marrero2020entanglement}.

In classical computer science, most diffused implementations of PCA, CA, and LSA available \cite{scikit-learn} relays on ARPACK \cite{lehoucq1998arpack} or similar packages, which implement improvements of the Lanczos method, like the Implicitly restarted Arnoldi method (IRAM) \cite{sorensen1997implicitly}, an improvement upon the simple Arnoldi iteration, which dates back to 1951 (a more general case of Lancsoz algorithm, which works only for Hermitian matrices). The run-time of these algorithms is bounded by $O(nmk \frac{\ln (m/\epsilon)}{\sqrt{\epsilon}})$, where $\epsilon$ is an approximation error related to the relative spectral gap between eigenvalues \cite{saad1992numerical}. 

The realization of quantum procedures that provide exponential speed-ups in linear algebra tasks has given inspiration for the realization of classical quantum-inspired algorithms that try to achieve the same run-time as their quantum counterparts. 
The process of transforming a quantum algorithm into a classical algorithm with a similar speed-up is usually referred to as ``dequantization''.  
In our case, the comparison with dequantized algorithms is often not easy, as they solve problems that are different from ours. 
Most of these works are based on a famous algorithm by Frieze, Kannan, and Vempala, which computes a low-rank approximation of a matrix in time that is sub-linear in the number of entries \cite{fkv2004low_rank_approx, chia2020sampling, arrazola2019quantum}.
Such algorithms promise exponential speed-ups over the traditional SVD algorithm for low-rank matrices. 
However, the high polynomial dependency of the run-times on the condition number, the rank, and the estimation error makes them advantageous only for matrices of extremely large dimensions, with low ranks and small condition numbers. 
The research described in \citet{arrazola2019quantum} observed that the dependencies like $O\left(\frac{\norm{A}_F^6}{\epsilon^6}\right)$ are far from being tight in real implementations, but still order of magnitudes slower than the best classical algorithms. 

Concomitantly to our work, a new important result \cite{chepurko2020quantum} was able to lower the complexity of these dequantizations by better leveraging all the previous literature of classical algorithms in randomized linear algebra and re-framing them into a more complete mathematical framework. 
Indeed, previous sample-based dequantizations were just doing a form of leverage score sampling. 
These new algorithms seem to be tighter than previous results and offer a better comparison with quantum algorithms, solving problems related to ours.
While we believe that it is not possible to have classical algorithms with run-times comparable to the ones of Theorems \ref{TheoMio:factor_score_estimation}, \ref{TheoMio:check_explained_variance}, \ref{Theorivisto:binarysearch} (see the relationships between LLSD, SUES, and DQC1 in \citet{cade2018quantum}) and Corollaries \ref{Coro:qPCAvector} and \ref{Coro:qPCAmatrix}, we have found that the work of \citet{chepurko2020quantum} may question the practical advantage of our Theorem \ref{TheoMio:top-k_sv_extraction} over a classical counterpart.
At first sight, their Theorem 33 might seem relevant for this work, as it provides a set of linearly independent rows of the input matrix. 
We stress that this problem is not related to finding the singular vectors provided by SVD, which are linearly independent and orthonormal. 
Moreover, even after further orthonormalization processing (e.g., Gram–Schmidt), the computed row basis wouldn't necessarily be the one provided by SVD.
This is why we cannot compare the run-time of this procedure to our Theorem \ref{TheoMio:top-k_sv_extraction}. 
On the other hand, Theorem 37 is more similar to our Theorem \ref{TheoMio:top-k_sv_extraction} but still aims to solve a different problem.
While ours provides estimates $\norm{v_i - \overline{v}_i} \leq \epsilon, \forall i \in [k]$ (which we recall are also relative-error estimates, as $\norm{v_i}=1$), their Theorem 37 provides a rank-$k$ projector matrix $Q^{(k)}$, with orthonormal columns, such that $\norm{A - AQ^{(k)}Q^{(k)T}}_F^2 \leq (1+\epsilon')\norm{A - A_k}_F^2$ in time $\widetilde{O}(nnz(\m{A})+\frac{k^{w-1}m}{\epsilon'} + \frac{k^{1.01}m}{\epsilon'^2})$.
While it is easy to see that $Q \rightarrow V$ as $\epsilon \rightarrow 0$, it is not easy to see how $\norm{\m{Q} - \m{V}}_F$ varies as $\epsilon$ varies and that becomes even less clear if we are interested in the error on a specific singular vector.
If the run-time of this algorithm is shown to be better than its quantum equivalent, it would still be great to include it in our framework instead of Theorem \ref{TheoMio:top-k_sv_extraction} and continue to take advantage of the speed-ups of the other quantum procedures. 
One downside of using the dequantized subroutines would be that, in general, the $\widetilde{O}(nnz(\m{A}))$ data pre-processing step is different from the one required to provide efficient quantum access. 
Even though it can be possible that a classical algorithm could extract the singular vectors with a run-time comparable to the quantum one, using it would require paying additional costs both in time and space.
Those costs arise from the need for an ad hoc data structure that would not be adequate to provide competitive speed-ups with respect to the other available quantum machine learning and data analysis algorithms. 
We believe that both the classical and quantum versions of singular vectors extraction may be used in the future, depending on the computational capabilities available to the interested data analysts.

\subsection{Principal component analysis} Probably no other algorithm in ML has been studied as much as PCA, so the literature around this algorithm is vast \cite{halko2011algorithm,jolliffe2016principal}. To mention an improvement upon the standard Lanczos method for PCA \cite{wang2020improved}, the authors used more Lanczos iterations to improve the numerical stability of PCA, by obtaining a better description of the Krylov subspace (i.e., more iterations help obtain a more orthonormal base).   
As mentioned, the problem of PCA has been studied previously within the model of quantum computation. 
\citet{lin2019improved, he2020exact} focus on a circuit implementation of qPCA, whose run-time has been superseded by more recent techniques used in this paper. The work of \citet{yu2019quantumdatacompression} faces the problem of performing PCA for dimensionality reduction on quantum states achieving an exponential advantage over the best known classical algorithms. However, their algorithm is somewhat impractical, due to the overall error dependence, which can be of $\widetilde{O}(\epsilon^{-5})$. Furthermore they use old Hamiltonian simulation techniques, superseded by the techniques that we use in our paper. 
To our knowledge, there are no works that provide a theoretical analysis of the run-time for the procedure needed to \emph{select} the number of singular vectors needed to retain enough variance, obtain a classical description of the model, and map new data points in the new feature space with theoretical guarantees on the run-time (which we believe cannot be improved, as in this work we show that the run-time for this mapping is almost constant). 

\subsection{Correspondence analysis} While correspondence analysis has been really popular in the past, so much that entire books have been written about it \cite{clausen1998applied, greenacre2017correspondence}, it seems to have become out of fashion in the last decades, probably overshadowed by the wave of results in deep learning.
The novel formulation of \citet{hsu2019correspondence} gives a new perspective of CA. 
The authors connects correspondence analysis to the principal inertia components theory, making it relevant also in tasks that concern privacy in machine learning \cite{wang2019privacy}.
As said before, similarly and independently from us, \citet{koide2020dequantumcorrespondence} have extended the dequantized subroutines to perform canonical correspondence analysis. This algorithm is not expected to beat the performance of our quantum algorithm, let alone the performance of the best classical algorithm for CA. 

\subsection{Latent semantic indexing} LSA was first introduced in \citet{deerwester1990indexing}, which spurred a flurry of applications \cite{landauer2013handbook}. 
Some notable works are streaming and/or distributed algorithms for incremental LSA \cite{vrehuuvrek2011subspace, cavanagh2009parallel,zhang2017index}.
While these work might offer inspiration for new quantum algorithms, their distributed nature make it an unfair comparison with a single-QPU quantum algorithm. 
LSA with neural networks has also been explored in the past years \cite{yu2008latent}, albeit without guarantees on the run-time or the approximation error. 
During the preparation of this manuscript we discovered a previous work on quantum LSA, which pointed  at the similarities between quantum states and LSA, albeit without offering any practical algorithm \cite{qlsa2021caicedo}. 

\end{appendices}



\end{document}